\documentclass[12pt]{article}
\usepackage[utf8]{inputenc}
\usepackage[margin=1in]{geometry}
\usepackage{multirow, multicol}
\usepackage{booktabs}
\usepackage{dsfont,amssymb,amsmath,amsthm,centernot,graphicx,accents}
\usepackage{enumitem}
\usepackage{natbib}
\usepackage[colorlinks = true, citecolor = blue]{hyperref}
\usepackage{cleveref}
\usepackage{setspace}
\usepackage{comment}

\newtheorem{theorem}{Theorem}
\newtheorem*{theorem*}{Theorem}
\Crefname{theorem}{Theorem}{Theorems}
\crefname{theorem}{theorem}{theorems}

\newtheorem{proposition}{Proposition}
\newtheorem*{proposition*}{Proposition}
\Crefname{proposition}{Proposition}{Propositions}
\crefname{proposition}{proposition}{propositions}

\newtheorem*{corollary*}{Corollary}
\Crefname{corollary}{Corollary}{Corollaries}
\crefname{corollary}{corollary}{corollaries}

\newtheorem*{claim*}{Claim}
\Crefname{claim}{Claim}{Claims}
\crefname{claim}{claim}{claims}

\newtheorem{lemma}{Lemma}
\newtheorem*{lemma*}{Lemma}
\Crefname{lemma}{Lemma}{Lemmas}
\crefname{lemma}{lemma}{lemmas}

\newtheorem*{conjecture*}{Conjecture}
\Crefname{conjecture}{Conjecture}{Conjectures}
\crefname{conjecture}{conjecture}{conjectures}

\theoremstyle{definition}

\newtheorem{definition}{Definition}
\newtheorem*{definition*}{Definition}
\Crefname{definition}{Definition}{Definitions}
\crefname{definition}{definition}{definitions}

\newtheorem*{assumption*}{Assumption}
\Crefname{assumption}{Assumption}{Assumptions}
\crefname{assumption}{assumption}{assumptions}

\newtheorem{manualassumptioninner}{Assumption}
\newenvironment{manualassumption}[1]{%
  \renewcommand\themanualassumptioninner{#1}%
  \manualassumptioninner
}{\endmanualassumptioninner}
\crefname{manualassumptioninner}{assumption}{assumptions}
\Crefname{manualassumptioninner}{Assumption}{Assumptions}

\theoremstyle{definition}
\newtheorem{remark}{Remark} 

\newcommand{\E}{E}
\newcommand{\R}{\mathds{R}}

\newcommand{\Q}{\mathbb{Q}}

\newcommand{\indep}{\rotatebox[origin=c]{90}{$\models$}}

\renewcommand{\P}{\mathbb{P}}

\newcommand{\Var}{\textup{Var}}
\newcommand\ubar[1]{\underaccent{\bar}{#1}}

\newcommand{\PObs}{\P^{\textup{Obs}}}

\newcommand{\Pobs}{\PObs}

\newcommand{\PTrue}{\P^{\textup{True}}}
\newcommand{\PTarget}{\P^{\textup{Target}}}

\newcommand{\ProbManip}{\eta}

\newcommand{\PP}{\mathbb{P}}
\newcommand{\op}{o_{\mathbb{P}}}
\newcommand{\Op}{O_{\mathbb{P}}}

\DeclareMathOperator*{\argmin}{arg\,min}

\allowdisplaybreaks

\title{Sensitivity Analysis for Linear Estimators}

\date{\today}
\author{Jacob Dorn and Luther Yap\footnote{This material is based upon work supported by the National Science Foundation Graduate Research Fellowship Program under Grant No. DGE-2039656 and by grant T32 HS026116 from the Agency for Healthcare Research and Quality. Any opinions, findings, and conclusions or recommendations expressed in this material are those of the author and do not necessarily reflect the views of the National Science Foundation or the Agency for Healthcare Research and Quality. Anthropic's Claude Opus was used to assist with programming (simulation and data analysis) and copyediting.}}

\begin{document}

%\onehalfspacing 
%\doublespacing

\maketitle

\begin{abstract}
We propose a framework for building sensitivity analyses for linear estimators. We measure the degree of identification failure by the change of measure between the observed distribution and a hypothetical target distribution under which a standard estimator would identify the parameter of interest. Under $L_\infty$ bounds on that derivative, the identified set has closed-form sharp bounds that require no numerical optimization and that are computable and continuous when the derivative is unbounded, nesting  Manski-type worst-case bounds and existing characterizations that require bounded changes of measure. The framework nests several selection models for average potential outcomes, yielding the established worst-case bounds of each design. The framework yields a new model family for regression discontinuity (RD) that delivers the first sensitivity analysis on the conditional average treatment effect under manipulation. We show that a bootstrap procedure yields valid confidence intervals. An empirical application shows that a modest restriction on manipulation delivers bounds far more informative than worst-case bounds.
\end{abstract}

%%%%%%%%%%%%%%%%%%%%%%%%%%
%%%%%%%%%%%%%%%%%%%%%%%%%%
\section{Introduction}

Sensitivity analysis is often constructed one design at a time. Selection-on-observables and difference-in-differences designs possess a rich set of competing proposals for quantifying robustness to unobserved confounding. Regression discontinuity possesses only extreme worst-case analyses, and the proposals available for instrumental variables depend on the scale of the outcome. Some estimands lack any sensitivity analysis at all.

For all of the design-by-design theoretical work, the practical upshot of these frameworks is similar. There is some violation --- unobserved confounding, nonparallel trends, a manipulated running variable, or a nonexcluded instrument --- that means that the data we see may not come from the family of distributions needed to point-identify a causal effect. When unified this way, the shared character of these design-by-design analyses  raises the question of whether a single framework can deliver sensitivity analysis across linear estimators. A unified framework removes the need to develop case-specific methods, allows researchers to compare sensitivity across identification strategies on a common scale, and allows a technique developed for one design to transfer to others. We provide such a framework.

This paper proposes building up a sensitivity analysis for any linear estimator via a hypothetical Radon-Nikodym derivative, asking how much one would need to reweight the data to change the implied causal parameter by a given degree. Measuring identification violations via data reweighting offers a unit-free measure of violations, because the Radon-Nikodym derivative is invariant to invertible transformations. Further, the class of feasible reweightings naturally yields a sensitivity analysis with bounds that are provably sharp. Strong point-identifying restrictions are the case in which weight is constrained to be one, and Manski-type domain restrictions are the case in which the weight is unbounded. Between the two extremes, the degree of restriction offers a simple measure of the degree to which the point-identifying assumptions must be violated in order to overturn a given conclusion.

We focus on discrete upper and lower reweighting constraints applied to linear estimators. These constraints are computationally convenient, as the constraint decouples across covariate levels. There is a known general result from the literature on the Average Potential Outcome (APO) when the upper constraint is finite: the worst case causal parameter reduces to reweighting the data based on comparing the observed outcome to a specific quantile function  \citep{tan2006distributional,  dorn2022sharp,dvds,frauen2023sharp,tan2024model}. Previous results for models with unbounded upper constraints have been available only on a case-by-case basis and have required integrating the conditional outcome quantile function with respect to the quantile index \citep{masten2018identification,MastenPoirierZhang}. 
We show how to smoothly extend the previous bounded-weight characterization to unbounded upper constraints without integrating over quantiles. The key is to express the bounds in terms of the conditional value at risk above the threshold quantile, which coincides with the threshold quantile as the upper constraint grows and the threshold tends to the outcome support upper bound. 
%At the boundary the threshold quantile and the conditional value at risk coincide, so the same closed-form reweighting applies whether or not the implied Radon--Nikodym derivative is bounded.
By covering the unbounded case, the sensitivity parameter provides a scale that ranges from point identification to the worst case bounds of \citet{Manski1990} for the APO and \citet{gerard2020bounds} for RD, so our analysis nests those results.
%Covering the unbounded case is what turns the sensitivity parameter into a complete scale. Its endpoint is where the framework must reproduce the worst-case bounds ---\citet{Manski1990} for the APO, \citet{gerard2020bounds} for RD---so that our analysis nests those results, and where a breakdown point remains defined. 

A natural question for our measure is whether the Radon-Nikodym bound can be interpreted in terms of economic primitives. While sharpness is guaranteed for bounds on the Radon-Nikodym derivatives, an underlying model used to derive Radon-Nikodym constraints can have further restrictions that make a Radon-Nikodym-based bound conservative. In models with selection on unobservables, the Radon-Nikodym derivative can typically be linked to the underlying selection process. As a result, one can construct the degree of allowed reweighting from the degree of allowed selection on unobservables, with the specific link varying on a case-by-case basis. We offer such a link in our three applications. In models with two feasible treatment statuses, our reweighting-based bounds are generally sharp for an underlying selection bound; in models with more feasible treatment statuses, we show that maintained restrictions on the underlying selection process can have further content that our proposed simpler procedure will not exploit. 

Our framework immediately yields new sensitivity analyses for a wide variety of estimands.  In our first and primary application, we study bounds on causal effects from sharp RD with one-sided manipulation. \citet{mccrary2008manipulation} proposes a test for the RD assumption of no manipulation. When \citeauthor{mccrary2008manipulation}'s test fails,  \cite{gerard2020bounds} propose a worst-case bound on the conditional average treatment effect for non-manipulators: a Conditional Local Average Treatment Effect (CLATE). We show how to interpret \citeauthor{gerard2020bounds}'s manipulation bounds as a Manski-type limit. As a byproduct, we show that more moderate distributional distances can yield meaningful bounds on the more standard conditional average treatment effect (CATE), which to the best of our knowledge has been an open issue in the literature. 
The appendix works out two further designs in the same way: differences in differences under selection into treatment, where the bounds are sharp, and instrumental variables under exclusion failure, where they are valid but conservative under primitive assumptions.

We show that bootstrap confidence intervals remain valid, asymptotically covering the causal effect with at least nominal probability. There are three sources of statistical uncertainty to capture: sampling variation, the estimated quantile at the index implied by estimated weights, and the estimated weight constraints. Sampling variation is standard in our setting. The estimated quantile is non-standard, but errors in the estimated index and errors in the estimated quantile at that index push the implied bounds outwards; we therefore follow the \citet{dvds} strategy of showing that a bootstrap with the correct quantile function would achieve target coverage, and show that feasible the bootstrap with estimated quantiles produces confidence intervals that are at least as wide as needed for coverage. 
%This argument also covers the third source of uncertainty. 
%Further, our bounds orthogonalize the quantile estimate, so that a nonparametric quantile estimate of controlled complexity is sufficient to achieve an asymptotically exact confidence interval, even if the quantile estimates are held fixed across bootstraps. The effect of weight constraint estimation on statistical inference is context-dependent and typically must be reflected in the bootstrap. 
We focus on RD with manipulation, which is our least standard application. We adapt \cite{gerard2020bounds}'s trick of tilting the estimated manipulation share upwards at a diminishing rate, which moves the reported bounds outwards. We show that the confidence intervals cover the true causal effect with at least nominal coverage. 

In simulations calibrated to the sharp RD design, our confidence sets attain at least their nominal level even when the data are drawn from the sensitivity model boundary at which inference is most likely to fail. The bound estimator displays the expected outer robustness: coverage is conservative once the identified set has positive width, but far narrower than  \citet{gerard2020bounds} bounds.

Our empirical application quantifies the value of our procedure for RD with manipulation. 
%We first apply our procedure to study the incumbency advantage, in which previous work has evaluated the contribution of incumbent victory to the incumbent party winning in the following election. 
%A simple plug-in procedure detects that the conclusion of incumbency advantage is insensitive to the no-manipulation assumption, including for the CATE that previous proposals could not bound.
%While the results are largely similar across estimands, we find that the CLATE bounds are higher than the CATE bounds, suggesting that non-manipulators are more likely to win the next election than manipulators, but we find that the sign of the estimated CATE is robust to manipulation odds movements larger than a factor of $10$. 
We study the effects of council size on expenditure. Larger municipalities in Bavaria have a discrete increase in council size and in total government expenditures. However, statistical evidence is consistent with 11\% of expanded councils at the cutoff being manipulators. The resulting worst-case bounds are highly uninformative. Our sensitivity analysis yields meaningfully tighter bounds, estimating that if unobserved factors were able to change manipulation odds by a factor of $2.71$, then the point estimate for the effect of council expansion on expenditures could be zero; odds ratio movements of only $2.1$ suffice for the confidence interval to include zero. %and a breakdown point for confidence intervals and point estimates of an odds-factor of $1.65$ and $2.1$, respectively.
%: the positive effect on total expenditure survives until the odds of manipulation are allowed to vary by a factor of about $\Lambda = 2.1$, with the confidence-set breakdown near $\Lambda = 1.65$. 

\subsection{Related Work}

%We now mention some related work. 
%Our proposal nests many previous proposed sensitivity analyses for IPW estimator. A sensitivity analysis is an approach to partial identification that begins from assumptions that point-identify the structural estimand of interest, and then considers increasing relaxations of those assumptions \citep{MolinariChapter}. Our proposal is to instead relax the implication of the point-identifying assumptions in terms of Radon-Nikodym derivatives. The two approaches coincide for the Average Potential Outcome (APO): notably, \citet{tan2006distributional}'s original marginal sensitivity model starts with a Radon-Nikodym bound rather than the equivalent odds ratio bound. 

%Our specific analysis of worst-case bounds extends \citet{dvds}'s sharp characterization of bounds beyond \citet{tan2006distributional}'s marginal sensitivity model. \citet{tan2024model} and \citet{frauen2023sharp} previously extended this characterization to families that bound the Radon-Nikodym derivative of interest. 
We generalize existing APO worst-case bounds to also include unbounded Radon-Nikodym derivative, so that we can include a compact characterization of bounds under \citet{masten2018identification}'s conditional c-dependence model as a special case. 
Our paper uses the same Distributionally Robust Optimization (DRO) ideas as \citet{dvds}, \citet{frauen2023sharp}, and \citet{tan2024model}, but we generalize the class of assumptions that can be considered in the unconfoundedness case to allow infinite weight upper bounds as elements (and not merely limits) of the set, and show how to apply these ideas in settings beyond APO with unobserved confounding. Our proposed bootstrap unifies \citet{zhao2019sensitivity}'s percentile bootstrap dominance argument, \citet{dorn2022sharp}'s quantile dominance argument, and \citet{gerard2020bounds}'s estimated manipulation dominance argument to show that RD confidence intervals asymptotically cover with at least nominal probability. 
%There is rich work in this literature under other sensitivity assumptions like $f$-divergences and Total Variation distance  \cite[e.g.][]{HuangAndPimentel2022, freidling2022sensitivity, LongStoryShort, ChristensenAndConnault, IshikawaEtAlUpdate}. These other assumptions also fit within our framework, because our target distribution constructions are independent of the $L_\infty$ sensitivity assumptions that we analyze.

Our framework relates to existing work on sensitivity analysis for other applications. Our proposed sensitivity analysis for sharp RD applies when data on the running variable fails tests for manipulation \citep{mccrary2008manipulation}. Our proposal nests both an exogeneity-type assumption and \citet{gerard2020bounds}'s bounds as special cases. There is other work in the RD context on partial identification bounds under manipulation \citep{ishihara2020manipulation} but to our knowledge, our proposal is the first sensitivity analysis for manipulation that allows a continuous measure of manipulation severity instead of taking the worst-case bounds. 
%There are sensitivity analysis for exclusion failure with instrumental variables \citep{ramsahai2012causal, freidling2022sensitivity}, but to our knowledge our proposal is the first sensitivity analysis whose underlying assumptions are invariant to invertible transformations of variables. 

Two related frameworks are \cite{bertsimas2022distributionally} and \cite{ChristensenAndConnault}, both of which are limited to discrete covariates. In the discrete covariate case, relative to \citeauthor{bertsimas2022distributionally}, we propose justifying distributional distances in terms of underlying structural models of economic failure and propose constructing target distributions that apply in settings beyond average treatment effects, and offer a specific analysis of $L_\infty$ bounds; and relative to \citeauthor{ChristensenAndConnault}, we analyze a nominally different class of identification failures and provide closed-form bounds for specifically linear estimators.

Our reweighting parameter admits a breakdown-point reading: the smallest relaxation of the point-identifying assumption under which a given conclusion is overturned, which we report in our applications. This connects to a literature on inference for breakdown points and breakdown frontiers \citep{MastenPoirierZhang, huang2023design, zhao2019sensitivity}, and to sensitivity summaries that report the amount of confounding needed to overturn a result. We treat these quantities as descriptive by-products of the identified set instead of objects of inference; formal inference on the breakdown frontier is beyond the scope of this paper.

\textbf{Notation}. We use $\bar{\R}$ to refer to the extended real number line $\R \cup \{ -\infty, \infty \}$. For a real-valued random variable $Z$ and a distribution $\Q$, we use the notation $E_{\Q}[Z] = \int z d \Q(z)$ and we write that the expected value of $Z$ exists under $\Q$ if $E_{\Q}[|Z|]$ is finite or $E_{\Q}[Z]$ is well-defined as exactly one of positive or negative infinity. We write that the expected value of $Z$ exists (without specifying the distribution) if the expected value exists under the observed distribution $\PObs$, where $\PObs$ is defined below. Similarly, we sometimes suppress the dependence of expectations when referring to the expectation under the observed distribution $\PObs$.  We write $\{ a, b \}_+ = \max\{a, b\}$ and $\{a, b\}_- = \min\{a, b\}$, and abuse notation by writing $\{a\}_+ = \max\{a, 0\}$. For random variables $Y \in \R^1$ and $R \in \R^d$ and a function $t: \R^d \to [0, 1]$, we refer to ``the" conditional quantile function $Q_{t(R)}(Y \mid R)$, which is any minimizer  of $E_{\PObs}[ t(R) \{ Y - Q(R), 0 \}_+ - (1-t(R)) \{ Y - Q(R), 0\}_- \mid R]$ in functions $Q$ from the domain of $R$ to the extended real line $\R \cup \{-\infty, \infty\}$; when $t(R) \in \{ 0,1 \}$, we take the infimum and supremum of the distribution's support respectively. To consolidate notation,  when $b$ is infinite, we evaluate the interval $[a, b]$ as $[a, \infty)$. If our estimation procedure calls for an estimate of a nuisance function $f$ that depends on another nuisance function $g$, we use the notation $\hat{f}$ to denote the full estimated nuisance function, which may include a composition of nuisance estimates. We use $1\{ . \}$ to refer to the indicator function that takes value 1 if true and 0 otherwise. We use $\mathbb{P}$ to denote both the probability mass function for a discrete variable and the cumulative distribution function for a continuous (or mixed) variable, and use $d \mathbb{P}$ to denote the probability density of the a continuous (or mixed) variable.

\section{Framework for Sensitivity Analysis}\label{sec:Framework}

%In the following section, we illustrate our generic framework for sensitivity analysis: we propose linking a structural model of identification failures to an implied statistical distribution bound and then conduct partial identification under the distributional bound. We characterize the identified set under a family of $L_\infty$ distributional distances that is especially tractable for estimation, because the optimization problem can be solved explicitly. We leave specific structural models to \Cref{sec:applications}.

\subsection{General Setting}

We begin with our general partial identification setting. We use the APO application as a running example. 

%We consider three different but related probability measures. $\PObs$ is the observed distribution over $(R, Y)$ where $Y \in \R^1$ is the outcome of interest and $R \in \R^d$ are observable quantities that do not include $Y$. We call $R$ the regressors. With $Y(r)$ denoting potential outcomes that are functions of $r$, $\PTrue$ is the true distribution over $(Y(r),R)$ that generates $Y=Y(R)$ in $\PObs$. Finally, $\PTarget$ is a target distribution over $(R, Y)$ that point identifies a causal estimand of interest, i.e., $E_{\PTarget} [\lambda(R) Y]$ is our object of interest for some real-valued function $\lambda$. In particular, $\PTarget$ is a modified distribution of the variables $(R,Y)$ such that if we applied the baseline estimand to this measure, we would obtain the correct parameter of interest.
Let $Y \in \R^1$ denote the outcome of interest, $R \in \mathcal{R}$ denote observable quantities that do not include $Y$, and $Y(r)$ denote potential outcomes. We call $R$ the regressors, and the observed outcome is $Y = Y(R)$. 

We consider three different but related probability measures. There is a true distribution $\PTrue$ over $( \{Y(r)\}_{r \in \mathcal{R}},R, \xi)$, where $\xi$ are unobserved variables that do not include potential outcomes. This $\xi$ will not feature in our running APO example, but it will become relevant in other applications, for example to include manipulator status in RD. While $\xi$ itself is unobserved, some features of $\xi$ may be point- or partially-identified. For instance, in the instrumental variables framework with potential treatments $T(z)$ where $z$ is the instrument, $\xi$ could include $T(z)$, which is itself unobserved, but the proportion of compliers, $\Pr(T(1)=1, T(0)=0)$, is identified in the population under monotonicity. 

The object of interest is a linear reweighting over potential outcomes with treatment status $r$, $E_{\PTrue} \left[\lambda(R) \int_r  Y(r)d G(r \mid R ) \right] $, where $\lambda$ is a real-valued function and $G$ is some known distribution weighting over values of $r$. We observe data $(Y, R)$ from the observable distribution $\PObs$, which is a coarsening of $\PTrue$. We face the fundamental problem of causal inference and do not observe $Y(r)$.

In our applications, for every distribution $\PTrue$, there exists a target distribution $\PTarget$ that point identifies the structural estimand through the moment \\
$$E_{\PTarget} [\lambda(R) Y] = E_{\PTrue} \left[\lambda(R) \int_r  Y(r)d G(r \mid R ) \right].$$ $\PTarget$ is not necessarily unique, and it measures the same space (i.e., of $(R,Y)$) as $\PObs$. The observed distribution $\PObs$ may not be able to point-identify the structural estimand under the researcher's preferred estimator, for example if $R$ is a selected treatment or if $R$ includes an instrument that fails the exclusion restriction.

%There is a target distribution $\PTarget$ over $(R, Y)$ that point identifies a causal estimand of interest, where $Y \in \R^1$ is the outcome of interest and $R \in \R^d$ are observable quantities that do not include $Y$. We call $R$ the regressors. However, the researcher only has access to data from an observed distribution $\PObs$. $\PObs$ is also a distribution over $(R, Y)$, but $\PObs$ may not be able to point-identify the causal estimand under the researcher's preferred estimator, for example if $R$ is a selected treatment or if $R$ includes an instrument that fails the exclusion restriction.
%possess the wrong conditional distribution of outcomes.
\begin{manualassumption}{Support}\label{assum:ObservableDistProperties}
    The marginal distribution of $R$ is the same under $\PTarget$ and $\PObs$. The distribution $\PTarget$ is absolutely continuous with respect to $\PObs$. 
\end{manualassumption}

\Cref{assum:ObservableDistProperties} has two components. 
The first component---that the $R$-marginal agrees under $\PTarget$ and $\PObs$---is not a restriction on the data-generating process. 
Since $R$ is observed, we \emph{construct} $\PTarget$ by fixing its $R$-marginal at the observed $\PObs(R)$ and reweighting only $Y \mid R$; the construction $\PTarget(R,Y) = \PObs(R)\, d\PTarget(Y \mid R)$ succeeds for any $\PTrue$.
Instead, this component restricts the class of estimands the framework can represent. 
Holding the $R$-marginal fixed confines all counterfactual content to the conditional outcome law, so every admissible structural estimand takes the form $\psi_0 = \E_{\PObs}[\lambda(R) \E_{\PTarget}[Y \mid R]]$, where the outer average taken against the observed $R$-law. 
Estimands defined over a counterfactual distribution of $R$---covariate-shift or identification failures that operate by distorting the frequencies of $R$ ---fall outside the framework. 
The second component is an absolute continuity assumption that ensures the existence of a Radon-Nikodym derivative $\frac{d \PTarget}{d \PObs}$. The assumption accommodates unbounded Radon-Nikodym derivatives, for example $Y \mid R \sim Unif(0, 1)$ under $\PObs$ and $f_{Y}(y \mid R) = 1\{y \in (0, 1]\} y^{-1/2} / 2$ under $\PTarget$. 
%The absolute continuity assumption is stronger than necessary for our results, and could be reduced to assuming that the support of $(Y, R)$ under $\PTarget$ is a subset of the support under $\PObs$. 
Absolute continuity rules out $\PTarget$ possessing mass points where $\PObs$ lacks mass points, except where such distributions can be achieved as limits of distributions within \Cref{assum:ObservableDistProperties}.

Our framework is inspired by the literature on sensitivity analysis for the APO $E_{\PTrue}[ E_{\PTrue}[ Y(1) \mid R] ]$. In this case, there is a distribution $\PTrue$ over $(X, T, Y(1), Y(0))$, but we only observe the distribution $\PObs$ over $(X, T, Y = Y(T))$, so that $R=(X,T)$. 
The researcher can accurately estimate $\PTrue( T \mid X )$, but cannot observe $\PTrue( Y(1) \mid X, T = 0)$. We propose the target distribution $\PTarget$ that first samples $(X, T) \sim \PTrue$ and then samples $Y \mid X, T$ from the distribution of $Y(T) \mid X$ under $\PTrue$. To be precise, the target distribution $\PTarget$ is defined as $\PTarget(X, T, Y) = \PObs(X, T) d\PTrue(Y(T) \mid X, T)$. The target distribution $\PTarget$ does not correspond to $\PTrue$, but if a causal estimator like an inverse propensity weighted (IPW) estimator were applied to $\PTarget$ instead of $\PObs$, the researcher would identify the APO under $\PTrue$. Our later examples illustrate the applicability of this framework for other structural estimands.

We refer to our target as the structural estimand, and the target value of the primary observational analysis as the primary estimand. We restrict ourselves to estimands that correspond to linear estimators. 
\begin{definition}\label{def:TargetEstimand}
    The \emph{structural estimand} is $\psi_0 = E_{\PTarget}[ \lambda(R) Y ]$ for some real-valued function $\lambda$. The \emph{primary estimand} is $E_{\PObs}[ \lambda(R) Y]$. 
\end{definition}
We generally assume that $\lambda$ is identified and the researcher has conducted a primary analysis that consistently estimates $\lambda$. For example, in the APO case, the $\lambda$ function is the inverse propensity $\frac{T}{\PObs(T = 1 \mid X)}$.  When the researcher is able to make strong enough assumptions that $\PTarget$ is equivalent to $\PObs$, then the primary estimand will be equal to the structural estimand. 
When the researcher is only able to bound the difference between $\PTarget$ and $\PObs$, then the researcher can partially identify the structural estimand through a restriction on the Radon-Nikodym derivative between the two distributions.
\begin{lemma}\label{lemma:ReweightingHolds}
    Suppose \Cref{assum:ObservableDistProperties} holds and $\lambda(R) Y$ is integrable under $\PTarget$. Then $$\psi_0 = E_{\PObs} \left[ \lambda(R) Y \frac{d \PTarget( Y \mid R) }{d \PObs(Y \mid R)} \right].$$
\end{lemma}

This reweighing characterization is useful for partial identification because any non-negative putative outcome reweighing $\bar{W} =  \frac{d \PTarget( Y \mid R) }{d \PObs(Y \mid R)}$ that satisfies $E[ \bar{W} \mid R] = 1$ almost surely will correspond to a well-defined distribution $\Q$ over $(R, Y)$.

The Radon-Nikodym derivative characterization in \Cref{lemma:ReweightingHolds} often maps to interpretable quantities. For example, in the APO case, the Radon-Nikodym derivative maps to odds ratios as follows\footnote{
To see why the above result is true, observe that, by applying Bayes' Rule and dropping conditioning on $X$, $d\PObs(Y \mid T=1) = d\PTrue(Y(1) \mid T=1) = \frac{\PTrue(T=1\mid Y(1)) d\PTrue(Y(1))}{\PObs(T=1)}$ and $d\PTarget(Y \mid T=1) = \frac{\PTarget(T=1\mid Y(1)) d\PTrue(Y(1))}{\PObs(T=1)} = d\PTrue(Y(1))$ as $\PTarget(T=1\mid Y(1)) = \PObs(T=1)$.
Hence,$\frac{d \PTarget(Y \mid R)}{d \PObs(Y \mid R)} = \frac{\PObs(T=1)}{\PTrue(T=1\mid Y(1))}$, which is easily verified to be the same as the RHS by factoring out $\PObs(T=1)$, then observing that $\PTrue(T = 0 \mid X, Y(1)) + \PTrue(T = 1 \mid X, Y(1))=1$.
}:
\begin{align}
    \lambda(R) \frac{d \PTarget(Y \mid R)}{d \PObs(Y \mid R)} & = \lambda(R) \PObs(T = 1 \mid X) \label{eq:MapToQuantitites} \\
    & + \lambda(R)\PObs(T = 0 \mid X) \frac{\PTrue(T = 0 \mid X, Y(1))}{\PTrue(T = 1 \mid X, Y(1))} \frac{\PObs(T = 1 \mid X)}{\PObs(T = 0 \mid X)}. \nonumber 
\end{align}

As a result, there is a bijection between structural models of treatment selection and causal models of treatment effects. More generally, principled restrictions on an underlying treatment selection model may imply, or be equivalent to, restrictions on the Radon-Nikodym derivative. The mapping from some structural model on objects such as $\frac{\PTarget(T = 0 \mid X, Y(1))}{\PTarget(T = 1 \mid X, Y(1))}$ to the Radon-Nikodym derivative as in \Cref{eq:MapToQuantitites} is application-dependent, but holds across sensitivity models.

Our general framework defines $\PTarget$ in this way and derives a correspondence between bounds on $\PTrue$ and bounds on the divergence between $\PTarget$ and $\PObs$. We provide formal results for this approach applied to a class of $L_\infty$ sensitivity models.

\subsection{Partial Identification Assumption and Result}

In this subsection, we characterize the sharp bounds on the identified set under an $L_\infty$ restriction on the largest and smallest Radon-Nikodym derivative.
\begin{definition}[Sensitivity assumption]\label{assum:SensitivityModel}
    For any pair of functions $\underline{w}, \bar{w}: \R^d \to \bar{\R}$ satisfying $0 \leq \underline{w}(R) \leq 1 \leq \bar{w}(R)$ almost surely, we define $\mathcal{M}( \underline{w}, \bar{w} )$ as the set of distributions $\Q$ over $( \{Y(r)\}_{r \in \mathcal{R}},R, Y=Y(R) , \xi)$ satisfying $ \frac{d \Q( R, Y )}{d \PObs( R, Y)} \in [ \underline{w}(R),  \bar{w}(R)]$ almost surely. 
\end{definition}
The space that the distribution $\mathbb{Q}$ measures is the union of the space measured by $\PTrue$ and $\PObs$. This family has several advantages. The restrictions on $\frac{d \Q}{d \PObs}$ decouple across values of $R$, enabling tractable characterizations of the identified set. The functions $\underline{w}(R)$ and $\bar{w}(R)$ bound the worst-case ratio between the target density and the observed density, which is allowed to be dependent on data $R$. These functions have a standard interpretation in the APO case, but the interpretation can be extended to other contexts. The family nests extreme assumptions as special cases: point identification corresponds to the case $\underline{w}(R) = \bar{w}(R) = 1$ almost surely, while Manski-type bounds that only restrict the support of potential outcomes to the support of observed outcomes correspond to $\bar{w}(R) = \infty$ with domain-appropriate $\underline{w}(R)$. In between, the structural estimand is only partially identified. When the outcome $Y$ is binary, the restriction can equivalently be viewed as a restriction on the conditional mean of $Y \mid R$. We show below that, as in the \cite{tan2006distributional} model that inspired this generalization, the resulting bounds are highly tractable for estimating sharp and valid bounds. 

\begin{remark}[Invariance to reparametrization]\label{rem:Invariance}
The model family $\mathcal{M}(\underline{w}, \bar{w})$ is invariant to invertible transformations of the outcome. In particular, the ratio $\frac{d\PTarget(Y \mid R)}{d\PObs(Y \mid R)}$ is the same when applying the framework to $\tilde{Y} = g(Y)$ for any invertible $g$. The sensitivity parameter is therefore unit-free: it is a statement about the probability masses, and it does not depend on whether the outcome is recorded in levels, logs, or any other invertible reparametrization.
\end{remark}

We adapt standard notation from \cite{HoAndRosen}. 
\begin{definition}\label{def:IdentifiedSet}
    The \emph{identified set}  is $\mathcal{I}(\underline{w}, \bar{w} ) = \{ E_{\Q}[ \lambda(R) Y  ] \mid \Q \in \mathcal{M}( \underline{w}, \bar{w} ) \}$. The \emph{sharp bounds} on the identified set are $\psi^-(\underline{w}, \bar{w} ) = \inf_{\psi \in \mathcal{I}(\underline{w}, \bar{w} )} \psi$ and $\psi^+(\underline{w}, \bar{w} ) =  \sup_{\psi \in \mathcal{I}(\underline{w}, \bar{w} )} \psi$. A \emph{valid} identified set is a superset of $\mathcal{I}( \underline{w}, \bar{w} )$, and a valid bound is an interval that contains $[\psi^-(\underline{w}, \bar{w} ), \psi^+(\underline{w}, \bar{w} )]$. We call bounds \emph{conservative} if they are valid but not sharp. 
\end{definition}
We abuse notation and simply write $\mathcal{I}$, $\psi^-$, and $\psi^+$ as shorthand for the identified set and bounds under a generic model family. As we verify below, the constraints we work under are convex, so that bounds suffice to characterize the identified set up to inclusion of endpoints. \Cref{def:IdentifiedSet} corresponds to the statistical partial identification bounds implied by an underlying structural model. An underlying primitive model underlying \Cref{assum:SensitivityModel} can sometimes, but not always, imply a narrower bound.

We will require some nuisance functions to characterize the identified set under the family of $L_\infty$ restrictions we allow. In particular:
\begin{definition}
    The threshold probability is $\tau(R) \equiv \frac{\bar{w}(R) - 1}{\bar{w}(R) - \underline{w}(R)}$. The threshold quantiles are $Q^+(R) \equiv Q_{\tau(R)}(\lambda(R) Y \mid R)$ and $Q^-(R) \equiv Q_{1-\tau(R)}(\lambda(R) Y \mid R)$. The likelihood shifting term is $$a(\bar{w}, \underline{w}, s) = (\bar{w} - \underline{w}) 1\{ s > 0 \} - (1-\underline{w}).$$ For a function $\bar{Q}: \R^d \to \bar{\R}$, define the pseudo-outcome $$\phi^+(\underline{w}, \bar{w}, r, y \mid \bar{Q}) \equiv \lambda(r) y + ( \lambda(r) y - \bar{Q}(r)) a( \bar{w}(r), \underline{w}(r), \lambda(r) y - \bar{Q}(r)),$$ and define $\phi^-(\underline{w}, \bar{w}, r, y \mid \bar{Q})$ analogously:
    \begin{align*}
        \phi^-(\underline{w}, \bar{w}, r, y \mid \bar{Q}) \equiv \lambda(r) y + ( \lambda(r) y - \bar{Q}(r)) a( \bar{w}(r), \underline{w}(r), \bar{Q}(r) - \lambda(r) y).
    \end{align*}
     Define the indicator variables $F \equiv 1\{ \bar{w}(R) \text{ finite}\}$, $G^+ \equiv 1\{ \lambda(R) Y - Q^+(R) > 0 \}$, $G^- \equiv 1\{ \lambda(R) Y - Q^-(R) < 0 \}$, and $H \equiv 1\{ \tau(R) < 1 \}$. 
    
\end{definition}
When $\bar{w}(R)$ is finite, the formula for $\tau(R)$ can be found in \citet{tan2024model} and \citet{frauen2023sharp}.
We will make additional regularity assumptions.
\begin{manualassumption}{Moments}\label{assum:QuantileMoments} 
    The expected values of $F |Q^+(R)|$, $F |Q^-(R)|$, $|\phi^+(\underline{w}, \bar{w}, R, Y \mid Q^+)|$ and $|\phi^-(\underline{w}, \bar{w}, R, Y \mid Q^-)|$ are well-defined.   
\end{manualassumption}
We allow infinite $E_{\PObs}[ Q^+(R) ]$ for completeness with unbounded outcomes. While $\bar{w}$ can be infinite, we rule out heavy tails that would yield infinite $E_{\PObs}[ 1\{ \bar{w}(R) \text{ finite}\} Q^+(R) ]$.

Our work focuses on the upper bound of the identified set. 
\begin{theorem}\label{thm:IDSetUB}
    Suppose Assumptions \ref{assum:ObservableDistProperties} and \ref{assum:QuantileMoments} hold. Then the sharp upper bound on the identified set is:
    \begin{align}
        \psi^+( \underline{w}, \bar{w} ) & = E_{\PObs}\left[ \lambda(R) Y + ( \lambda(R) Y - Q^+(R) ) a( \bar{w}(R), \underline{w}(R), \lambda(R) Y - Q^+(R) ) \right]. \label{eq:UBFormula}
    \end{align}
    The sharp lower bound $\psi^-( \underline{w}, \bar{w} ) $ follows symmetrically, i.e.,
    \begin{align}
        \psi^-( \underline{w}, \bar{w} ) & = E_{\PObs}\left[ \lambda(R) Y + ( \lambda(R) Y - Q^-(R) ) a( \bar{w}(R), \underline{w}(R), Q^-(R) - \lambda(R) Y ) \right]. \label{eq:LBFormula}
    \end{align} Further, the identified set is convex. 
\end{theorem}
The theorem states that the sharp upper bound can be written as a closed-form moment: after the constituent components such as $\lambda(.), Q^+(.),$ and $a(.)$ have been estimated, no further optimization problem needs to be solved. Proofs are in the appendix.

A proof sketch is in order. We abuse notation to write $\psi^+ = E_{\PObs}[\psi^+(R)]$, where $\psi^+(r)$ is a pointwise upper bound function satisfying:
\begin{align*}
    \psi^+(R) & = \sup_{W}  E_{\PObs}\left[ W \lambda(R) Y \mid R \right] \quad \text{s.t.} \quad W \in [\underline{w}(R), \bar{w}(R)] \quad \& \quad E_{\PObs}[W \mid R] = 1,
\end{align*}
and $\phi^+(R) \equiv E_{\PObs}[ \phi^+(\underline{w}, \bar{w}, R, Y \mid Q^+) \mid R ]$. The claim is $E_{\PObs}[\psi^+(R)] = E_{\PObs}[ \phi^+(R) ]$. We prove the stronger claim that  $\psi^+(R) = \phi^+(R)$ almost surely. For simplicity, this sketch will ignore the possibility that $\lambda(R) Y = Q^+(R)$ with positive probability and drop ``almost sure" caveats. As \cite{dvds} note, $\psi^+(R)$ can be mapped to a simple DRO problem over $(W - \underline{w}(R)) / (1 - \underline{w}(R)) \in [0, (\bar{w}(R) - \underline{w}(R)) / (1-\underline{w}(R))]$. The solution is 
\begin{align}
    \psi^+(R) = E_{\PObs}[ \underline{w}(R) \lambda(R) Y + (1-\underline{w}(R)) CVaR_{\tau(R)}^+(R) \mid R], \label{eq:CVaRFormula}
\end{align}
where $CVaR^+_{\tau(R)} =  E\left[ Q^+ + \frac{\{ \lambda(R) Y - Q^+\}_+}{1-\tau(R)} \mid R \right]$ is the level-$\tau(R)$ conditional value at risk of $\lambda(R) Y$.

We split on the event $\tau(R) < 1$: by previous work \citep{tan2024model, frauen2023sharp}, $\tau(R) < 1$ implies $\psi^+(R) = \phi^+(R)$. \Cref{thm:IDSetUB} is novel for the case $\tau(R) = 1$, which corresponds to either $\ubar{w}(R) = 1$ or $\bar{w}(R) = \infty$ ($e(X) \centernot{\in} (c, 1-c)$ in the case of conditional c-dependence). On that event, $\phi^+(R)$ evaluates to $E[ \underline{w}(R) \lambda(R) Y + (1 - \underline{w}(R)) Q^+(R) \mid R ]$ and $Q^+(R) = CVaR_{1}^+(R)$. As a result, $\psi^+(R) = \phi^+(R)$ even if $\bar{w}(R)$ is infinite with positive probability. Intuitively, at extreme probabilities, the conditional value at risk and quantile either converge (with bounded outcomes) or tend to the same infinite limit (with unbounded outcomes). In either case, replacing the conditional value at risk with the quantile has no effect.
In the other extreme case when $\underline{w}(R) = \bar{w}(R) = 1$, we have $a(\bar{w},\underline{w},s)=0$, so we have point identification regardless of the assignment of $\tau(R) \in (0,1)$. 

\begin{remark}[Extension to unbounded upper constraints]\label{rem:UnboundedUB}
\Cref{thm:IDSetUB} requires no finiteness of $\bar{w}$:  \Cref{eq:UBFormula} remains the sharp upper bound on any event where $\bar{w}(R) = \infty$, so that $\tau(R) = 1$. On such an event $Q^{+}(R)$ is the essential supremum of $\lambda(R) Y$ given $R$, so $\{\lambda(R) Y - Q^{+}(R)\}_{+} = 0$ almost surely and the pseudo-outcome reduces to
\[
    \phi^{+}(\underline{w}, \bar{w}, R, Y \mid Q^{+})
    = \underline{w}(R)\,\lambda(R) Y + \left(1-\underline{w}(R)\right) Q^{+}(R)
    \quad \text{a.s.}
\]
Since the level-$\tau(R)$ conditional value at risk collapses to the quantile there, $CVaR^{+}_{\tau(R)}(R) = Q^{+}(R)$, the same bound is delivered by the value-at-risk representation. One may therefore extend the bounded-$\bar{w}$ characterization to unbounded upper constraints. The only substantive requirement is \Cref{assum:QuantileMoments}, which permits an infinite $E_{\PObs}[Q^{+}(R)]$ while ruling out tails heavy enough to make $E_{\PObs}[F\,Q^{+}(R)]$ diverge. 
%This is what extends the characterizations of \citet{tan2024model} and \citet{frauen2023sharp}, which require a finite upper constraint, to families with unbounded Radon--Nikodym derivatives such as \citet{masten2018identification}'s conditional c-dependence.
\end{remark}

%While the reweighing characterization \Cref{eq:UBFormula} aids intuition in sensitivity analysis, we characterize the partial identification bounds using the conditional value at risk formula in \Cref{eq:CVaRFormula} directly to simplify estimation and computation.
While the reweighing characterization \Cref{eq:UBFormula} aids intuition in sensitivity analysis and offers the characterization we use in applications, the conditional value at risk formula in \Cref{eq:CVaRFormula} provides useful intuition for the behavior of the bounds and why the bounds can be continuous as $\bar{w}(R)$ tends to infinity.
%One could characterize the partial identification bounds using the conditional value at risk formula in \Cref{eq:CVaRFormula} directly, but the reweighing characterization \Cref{eq:UBFormula} is valuable for sensitivity analysis. 
We present the worst-case Radon-Nikodym derivative here for convenience.
\begin{lemma}\label{cor:WorstCaseWeights}
Suppose $\bar{w}(R)$ is bounded. Then one can construct a distribution $\Q^+ \in \mathcal{M}( \underline{w}, \bar{w} )$ and a random variable satisfying $\gamma(R) \in [\underline{w}(R), \bar{w}(R)]$ almost surely such that $\psi^+ = E_{\Q^+}[ \lambda(R) Y]$ and:
\begin{equation} \label{eqn:Wstar_sup}
    \frac{d \Q^+(R, Y)}{d \PObs(R, Y)} = W^{*}_{sup} =\begin{cases}
    \bar{w}(R) & \text{if } \quad \lambda(R)Y > Q^+(R) \\
    \underline{w}(R) & \text{if } \quad \lambda(R)Y < Q^+(R) \\
    \gamma(R) & \text{if } \quad \lambda(R)Y = Q^+(R). 
\end{cases}
\end{equation}
\end{lemma}
The construction in Equation (\ref{eq:UBFormula}) avoids the need to specify $\gamma(R)$.

A researcher with a primary estimate of $\lambda(R)$ and who is capable of quantile regression can easily construct a plug-in estimate of $\psi^+$. Even if the researcher fails to consistently estimate the additional nuisance parameter $Q^+$, they can still easily estimate valid bounds:
\begin{lemma}\label{cor:Robustness}
    Suppose Assumptions \ref{assum:ObservableDistProperties} and \ref{assum:QuantileMoments} hold and there is a putative quantile function $\bar{Q}$ such that the expectation of $|\phi^+(\underline{w}, \bar{w}, r, y \mid \bar{Q})|$ exists. Then replacing $Q^+$ with the potentially incorrect quantile function $\bar{Q}$ yields valid bounds:
    \begin{align*}
        \psi^+( \underline{w}, \bar{w} ) & \leq E_{\PObs}\left[ \lambda(R) Y + ( \lambda(R) Y - \bar{Q}(R) ) a( \bar{w}(R), \underline{w}(R), \lambda(R) Y - \bar{Q}(R) ) \right]. 
    \end{align*}
    An analogous result holds for the lower bound $\psi^-$.
\end{lemma}
The argument generally follows by \cite{tan2024model}'s decomposition of $(\lambda(R) Y - \bar{Q}(R)) a( \hdots )$ in terms of the Quantile Regression check function, although some care is needed to handle the case $\tau(R) = 1$.

A similar recipe as in this section extends to other sensitivity assumptions. With other restrictions on $\Q$, the model family in \Cref{assum:SensitivityModel} would change, the sharp upper bound in \Cref{thm:IDSetUB} would change and may lack a closed form, and the validity result \Cref{cor:Robustness} may or may not hold, but the fundamental logic would carry through.

\subsection{Illustration with Average Potential Outcomes} \label{sec:ipw}

To complete the exposition of our procedure, we show how the framework applies to our APO running example. 
Many of our results are restatements of recent work; the extension of partial identification bounds to unbounded Radon-Nikodym derivatives is new.

We assume there is a distribution $\PTrue$ over $(X, T, \{Y(t)\}_{t \in \{ 0,1 \}}, U)$, where $X$ are controls, $T$ is a discrete treatment, $Y(t)$ is the potential outcome corresponding to treatment level $t$, and $U$ are potential unobserved confounders. We only observe the coarsened distribution $\PObs$ over $(X, T, Y = Y(T))$, where $Y$ is the observed outcome. We write $I_t = 1\{T = t\}$.

We tailor our application to IPW estimation of the APO, and consider average treatment effects at the end of the section. We write the observable propensity $e(X) = \PObs(T = 1 \mid X)$, where we assume $e(X) \in (0, 1)$ almost surely. The primary estimate is the IPW estimator that first estimates  $e(X)$ and then estimates the APO as $E_{\PObs}[ \lambda(R) Y ]$, where $\lambda(R) = I_1 / e(X)$. When unconfoundedness holds given the observed covariates so that $I_t \indep Y(t) \mid X$, the IPW strategy consistently estimates the structural estimand $E_{\PTrue}[Y(1)] $. However, we will only assume unconfoundedness holds if the researcher had access to both the observed and unobserved covariates: $I_t \indep Y(t) \mid X, U$.\footnote{This is without loss of generality by setting $U = \{ Y(t) \}$.} We write $e(X,U) := \PTrue(T=1 \mid X,U)$ for the (partially identified) propensity score that conditions on $U$.

%The target distribution $\PTarget$ is defined as
%$$\PTarget(X, T, Y) = \PObs(X, T) \PTrue(Y \mid X, T).$$ 
%follows:
%\begin{enumerate}
%    \item Draw $X, T \sim \PObs$, where $\PObs$ is the observed distribution over covariates, treatment, and realized outcome 
%    \item Draw $Y$ from the distribution of $Y(T) \mid X$ under $\PTrue$, where $\PTrue$ is the true distribution over covariates, treatment, and potential outcomes \label{step:IPW:drawY}
%    \item Return $(X, T, Y)$
%\end{enumerate}
Recall the definition of $\PTarget$ for the APO: draw $R$ from the same distribution as $\PObs$, and draw $Y \mid R$ from $\PObs( Y \mid R )$ with probability $P( T \mid X )$ and from a certain reweighted distribution  (Equation \ref{eq:MapToQuantitites}) with probability $1-P(T \mid X)$. The distribution of $R$ is the same under $\PTarget$ and $\PObs$ and the distribution $\PTarget$ satisfies $E_{\PTrue}[Y(1)] = E_{\PTarget}[\lambda(R) Y]$, so that the structural estimand is $E_{\PTarget}[\lambda(R) Y]$. Notice that $\PTarget$ and $\PObs$ may have different distributions of $Y \mid X$: $\PObs(X, T, Y) = \PObs(X, T) d\PObs(Y \mid X, T)$ which differs from $\PTarget( X, T, Y)$ under unobserved confounding. %draw from $Y(T) \mid X, T$ in Step \ref{step:IPW:drawY}.  

%Suppose we have a structural model that implies 
Many researchers study models that imply 
there are functions $\ell(X), u(X)$ such that $$\ell(X) \leq \frac{e(X, U) / (1-e(X, U))}{e(X) / (1-e(X))} \leq u(X)$$ almost surely \citep{Manski1990, tan2006distributional, AronowLeeInterpretable,  zhao2019sensitivity, dvds, tan2023sensitivity, frauen2023sharp}. These models imply pointwise restrictions on $\PTrue(T = 1 \mid X, Y(1))$ and the Radon-Nikodym derivative $\frac{d \PTarget(Y \mid X, T=1)}{d \PObs(Y \mid X, T=1)}$. In particular, the selection assumption would imply the following almost sure bound:
\begin{align*}
    \frac{d \PTarget(Y \mid R)}{d \PObs(Y \mid R)} \in \left[ ( e(X) + (1-e(X)) u(X)^{-1}), \ \ ( e(X) + (1-e(X)) \ell(X)^{-1}) \right]. 
\end{align*}
Unconfoundedness corresponds to the special case $\ell(X) = u(X) = 1$. Manski bounds correspond to the special case $\ell(X) = 0$, $u(X) = \infty$. \cite{tan2006distributional}'s Marginal Sensitivity Model (called such for $e(X)$ implicitly marginalizing over $U$) corresponds to $\underline{w}(X, t) = \PObs(T = t \mid X) + \PObs(T \neq t \mid X) \Lambda^{-1}$ and $\bar{w}(X, t) = \PObs(T = t \mid X) + \PObs(T \neq t \mid X) \Lambda$ with their sensitivity parameter $\Lambda$. \cite{basit2023riskratiobased}'s risk ratio marginal sensitivity model, that restricts $e(X, U) \in \left[ \Gamma_0^{-1} e(X), \Gamma_1^{-1} e(X) \right]$, corresponds to $\underline{w}(R) = \Gamma_1$ and $\bar{w}(R) = \Gamma_0$ when $e(X) \leq \Gamma_1$, where $\Gamma_1$ and $\Gamma_0$ are their sensitivity parameters. \citeauthor{masten2018identification}'s conditional c-dependence model, that restricts $e(X,U) \in \left[ e(X) - c, e(X) + c\right] \cap [0, 1]$, corresponds to $\underline{w}(X, t) = \PObs(T = t \mid X) + \PObs(T \neq t \mid X) \frac{\max\{0, 1-(e(X)+c)\} / \min\{ 1, e(X) + c \}}{(1-e(X)) / e(X)}$ and $\bar{w}(X, t) = \PObs(T = t \mid X) + \PObs(T \neq t \mid X) \frac{\min\{1, 1-(e(X)-c)\} / \max\{ 0, e(X) - c \}}{(1-e(X)) / e(X)}$, where $c$ is their sensitivity parameter. 

\Cref{thm:IDSetUB} yields valid bounds on the identified set with the following quantities.
\begin{proposition}\label{prop:IPW_Y1}
In the APO case, our method can be implemented on $E_{\PTrue}[Y(1)]$ with:
\begin{align*}
    \lambda(R) = \frac{I_1}{e(X)}, \quad \underline{w}(R) = e(X) + (1-e(X)) u(X)^{-1}, \text{ and } \bar{w}(R) = e(X) + (1-e(X)) \ell(X)^{-1}.
\end{align*}
\end{proposition}

The result is intuitive. The $w$ bounds would be $1$ under the unconfoundedness assumption $\ell(X) = u(X) = 1$. As $e(X)$ gets closer to one, we see a greater share of the treated potential outcomes and the $w$ bounds grow closer to one.  %We can add the other $1-e(X)$ to varying degrees to fit the sensitivity assumption. 
An analogous approach can be used to bound $E[Y(0)] = E[Y(1-T)/(1-e(X,U))]$.

The result in \Cref{prop:IPW_Y1} extends the analysis of \citet{tan2024model, frauen2023sharp} to allow $\ell(X)$ to be arbitrarily small or equal to zero and $u(X)$ to be equal to one. As a result, it includes \cite{masten2018identification}'s conditional c-dependence assumption so long as $\lambda(R) Y$ is integrable under the target distribution. 
%This characterization of bounds under conditional c-dependence is also simpler than \cite{masten2018identification}'s characterization, which involves an integral over a transformation of the full quantile regression function, and by \Cref{cor:Robustness}, our characterization enables estimation of valid bounds under appropriate regularity conditions and consistency of the propensity score alone. An interesting question for future work is whether \cite{MastenPoirierZhang}'s proposed estimator for conditional c-dependence possesses similar validity guarantees. 

The resulting bounds are also sharp for the underlying selection constraints when bounding the average treatment effect (ATE). It is not difficult to verify that the bounds are sharp for each APO separately, because the bounds $\ubar{w}$ and $\bar{w}$ correspond to propensity scores between $0$ and $1$. One may also ask whether the separate APO bounds can be achieved by the same unobserved confounder, for example the potential outcomes. An extension of \cite{dorn2022sharp}'s sharpness argument shows that in fact the separate bounds can be achieved at the same time, so that these bounds also yield sharp bounds for the ATE. 

%Verifying \Cref{assum:QuantileMoments} is straightforward in this model. The binding case is the boundary at which $\bar{w}$ first becomes unbounded. At the boundary $\bar{w}(R)$ is almost surely finite but unbounded, which is exactly the configuration \Cref{thm:IDSetUB} accommodates and earlier characterizations do not.

\begin{remark}[Worst-case bounds]\label{rem:ManskiNesting}
Setting $\ell(X) = 0$ and $u(X) = \infty$ in \Cref{prop:IPW_Y1} gives $\ubar{w}(X,1) = e(X)$ and $\bar{w}(X,1) = \infty$, so $\tau(R) = 1$. If $Y$ has conditional support in $[y_{\min}, y_{\max}]$, then $Q^+(R) = y_{\max}/e(X)$, and since $\lambda(R) = 0$ when $T = 0$, \Cref{eq:UBFormula} contributes
\begin{align*}
    \E\left[ \phi^+( \ubar{w}, \bar{w}, R, Y \mid Q^+) \mid X \right] = e(X) \E[Y \mid X, T = 1] + (1 - e(X)) y_{\max},
\end{align*}
exactly \citeauthor{Manski1990}'s conditional worst-case bound, with the lower bound symmetric in $y_{\min}$. The identity holds at the endpoint, which is attained at finite sensitivity parameters: under \citeauthor{masten2018identification}'s conditional c-dependence, $\bar{w}(X,1) = \infty$ whenever $e(X) \leq c$. When the conditional support is unbounded the same expression returns $+\infty$. Note finally that $\ubar{w}(X,1) \geq e(X) > 0$ always, so the Manski case in this parametrization is $\ubar{w}(R) = e(X)$.
\end{remark}

%The conditional c-dependence sensitivity assumption of \citet{masten2018identification} can be expressed as $e_z(X, Y(z)) \in [e(X)-c, e(X)+c] \cap [0, 1]$. Then, using our existing notation, if $e(X) \in (c, 1-c)$, we have:
%\begin{align*}
%    u_z &= \frac{(e(X)+c)/(1-(e(X)+c))}{e(X)/(1-e(X))} \\
%    l_z &= \frac{(e(X)-c)/(1-(e(X)-c))}{e(X)/(1-e(X))}. 
%\end{align*}
%Then, bounds on the ATE can be obtained analogously using our framework.

\section{Application to Regression Discontinuity} \label{sec:RDManipulation}

In this section, we illustrate the applicability of our framework by applying it to potential failures of identifying assumptions: RD with manipulation. In \Cref{sec:AdditionalApplications}, we show how our framework applies to  DD without randomization or parallel trends, and IV without exogeneity or exclusion.  
If researchers are willing to make assumptions directly on the likelihood ratios of outcomes or if that likelihood ratio is an interpretable parameter, then the bounds are automatically sharp due to \Cref{thm:IDSetUB}.  
However, even when a generic reweighting is not immediately interpretable, researchers can map many structural parameters to our framework and still obtain valid bounds. In this section, since the first approach is trivial, we expand on the second approach and show cases when bounds are sharp. 
%In particular, we show that the sharp statistical bounds under our approach contain all restrictions implied by the underlying structural model for APO selection and RD manipulation, but not IV exclusion failure.  

%\subsection{Regression Discontinuity}\label{sec:RDManipulation}

%In the context of RD, 
We present non-trivial partial identification bounds for the CATE at the discontinuity; the extension to other linear estimators is relatively straightforward. We explain the setup, describe the estimands, then propose our sensitivity assumption and show how our framework applies. Finally, we show that our bounds are sharp for our structural model.

We work under a slight modification of \cite{gerard2020bounds}'s model. There is a full distribution $\PTrue$ over $(M, X(1), X(0), X, Y(1), Y(0), T, T(0))$, where $M$ is a variable corresponding to manipulation status, $X(m)$ is a potential running variable corresponding to manipulation status $M = m$, $Y(t)$ is a potential outcome corresponding to treatment status $T = t$, $T \in \{0, 1\}$ is the treatment status, and $T(0)$ is the potential treatment status for $M=0$. We face the fundamental problem of causal inference and only observe the coarsening $\PObs$ over $(X = X(M), Y = Y(T), T)$, so $R=(X,T)$ and $\xi = (M, X(1), X(0), T(0))$ in this context.

We study sharp RD designs. We assume there is a cutoff $c$ such that $\PTrue(T = 1 \mid X > c) = 1$, $\PTrue(T = 1 \mid X < c) = 0$, which mimics assumption (RD) of \citet{hahn2001identification}.\footnote{\cite{gerard2020bounds} extend their approach to fuzzy RD, where there are complex shape restrictions. $c$ has previously been used for $c$-dependence, but the contexts are distinct, so $c$ is used for the cutoff here without risk of confusion.} It is observationally testable and often obvious in applications: if $X$ is the net reported vote share of an election candidate, then the candidate wins if and only if $X > 0$. We use the notation ``$X = c$" to mean that we take the limit as $\varepsilon \to 0$ of $X \in [c - \varepsilon, c + \varepsilon]$ for bandwidth $\varepsilon$. (We assume that these and other relevant limits are well-defined.) Similarly, we let $X = c^+$ denote the limit corresponding to $X \in (c,c+\varepsilon]$ and let $X=c^-$ correspond to $X \in [c-\varepsilon,c)$. 
We use this notation to mimic the notation in \citet{gerard2020bounds} so that the results are immediately comparable. 
We operate in the bandwidth near the threshold, so indicators such as $1\{T=0 \}$ implicitly takes value zero for observations outside of the bandwidth. 
%In the rest of our exposition, we take the bandwidth $\varepsilon$ as a fixed object so that events $X=c^-$, $X=c^-$, and $X=c$ have strictly positive probabilities.

We assume that observations can be partitioned into manipulators and non-manipulators. We assume manipulation is one-sided, and without loss of generality assume manipulators choose treatment ($F_{X \mid M=1}(c) = 0$). 
%An interesting direction for future work are bounds in which manipulation can occur in both directions.  
As in \citeauthor{gerard2020bounds}, the probability of treated observations being manipulated, $$\ProbManip = \PTrue( M  = 1 \mid X = c^+ ),$$ is identified.\footnote{\citeauthor{gerard2020bounds} call this quantity $\tau$. Write $f$ for the density of the running variable and $f(c^-)$, $f(c^+)$ for its one-sided limits at the cutoff, which exist by \Cref{assum:RD}. Non-manipulators have a density that is continuous at $c$ and manipulators contribute mass only above it, so $f(c^-)$ is the non-manipulator density at $c$ and $f(c^+) - f(c^-)$ is the manipulator density just above it, giving $\ProbManip = 1 - f(c^-) / f(c^+)$. Both limits are functionals of the single density $f$: nothing in the definition depends on a bandwidth, or on how a sample is split to estimate them.}  We assume appropriate continuity of potential outcomes and running variables given the manipulation status, and discuss other regularity conditions in \Cref{assum:RD}.
%We use a slightly stronger version of \citeauthor{gerard2020bounds}'s assumption that non-manipulators' potential outcome CDFs are continuous at the cutoff and assume that non-manipulators' running variables are as good as random around the discontinuity: $( (Y(1), Y(0) ) \indep X \mid M = 0, X = c$, which mimics the local randomization approach of \citet{cattaneo2015randomization}.  
The assumptions imply that non-manipulator average treatment effects would be identified by the change in non-manipulator outcomes across the cutoff $c$. However, when $\ProbManip > 0$ so that there is manipulation, the distribution of $Y \mid X = c^+$ includes manipulators' treated potential outcomes, so that standard treatment effects are not point-identified. 

To ensure that objects are well-defined at the cutoff in light of discontinuities, we define, for events $A,B$,
\begin{align*}
    d\PObs \left( X= c^+ \mid X=c \right) &= \frac{d\PObs(X=c^+)}{d\PObs(X=c^+) + d\PObs(X=c^-)}, \\
    \Pr\left( A \mid X=c, B \right) &= d\PObs \left( X= c^+ \mid X=c \right) \Pr \left( A \mid X=c^+, B \right) \\
    &\quad + \left(1-d\PObs \left( X= c^+ \mid X=c \right) \right) \Pr \left( A \mid X=c^-, B \right), \\
    E\left[ A \mid X=c, B \right] &= d\PObs \left( X= c^+ \mid X=c \right)  E\left[  A \mid X=c^+, B \right] \\
    &\quad + \left(1-d\PObs \left( X= c^+ \mid X=c \right) \right)  E\left[  A \mid X=c^-, B \right].
\end{align*}

We may be interested in the conditional average treatment effect (CATE), the conditional local average treatment effect (CLATE), and the conditional average treatment effect on the treated (CATT):
\begin{equation} \label{eq:RD_estimands}
\begin{split}
    \psi^{CATE} & \equiv E\left[ Y(1) - Y(0) \mid X = c \right] \\
    \psi^{CLATE} & \equiv E\left[ Y(1) - Y(0) \mid X = c, M = 0 \right] \\
    \psi^{CATT} & \equiv E\left[ Y(1) - Y(0) \mid X = c^+ \right]. 
\end{split}
\end{equation}

$E\left[ Y(1) - Y(0) \mid X = c, M = 0 \right]$ is called the CLATE because it is an average treatment effect for the subgroup of potentially-assigned ``compliers" who receive treatment if and only if their $X_i$ is above the cutoff \citep{gerard2020bounds}. 

We write the structural estimands of interest in terms of conditional expectations of \textcolor{black}{observed variables} and \textcolor{black}{potential outcomes} as follows: 
\begin{lemma}\label{lemma:EstimandExpressions}
    Our main estimands of interest have the following expressions:
    \begin{align*}
    \psi_{CATE}& = \frac{1}{2 - \textcolor{black}{\ProbManip}} \textcolor{black}{E\left[ Y \mid X = c^+ \right]} + \frac{1-\textcolor{black}{\ProbManip}}{2-\textcolor{black}{\ProbManip}} \textcolor{black}{\E_{\PTrue} \left[ Y(1) \mid X = c, M = 0 \right]} \\
    & \qquad - \frac{\textcolor{black}{\ProbManip}}{2-\textcolor{black}{\ProbManip}} \textcolor{black}{\E_{\PTrue} \left[ Y(0) \mid X = c, M = 1 \right]} - \left(1-\frac{\textcolor{black}{\ProbManip}}{2-\textcolor{black}{\ProbManip}}\right) \textcolor{black}{E\left[ Y \mid X = c^- \right]} \\
    \psi_{CLATE} & = \textcolor{black}{\E_{\PTrue}[Y(1) \mid X = c, M = 0]} - \textcolor{black}{E[Y \mid X = c^-]}  \\
    \psi_{CATT} & = \frac{\textcolor{black}{E\left[ (2 T - 1) Y \mid X = c \right]} - \frac{\textcolor{black}{\ProbManip}}{2-\textcolor{black}{\ProbManip}} \textcolor{black}{\E_{\PTrue}[Y(0) \mid X = c, M = 1]}}{\textcolor{black}{\PObs( X = c^+ \mid X = c)}}.    
    \end{align*} 
\end{lemma}

Most quantities in the expressions above are point-identified. The only unidentified objects are the conditional expectations \textcolor{black}{$\E_{\PTrue} \left[ Y(1) \mid X = c, M = 0 \right]$} and \textcolor{black}{$\E_{\PTrue} \left[ Y(0) \mid X = c, M = 1 \right]$}.

The specific target distribution $\PTarget$ will depend on the estimand of interest as follows:
%\footnote{A formal construction would define the appropriate distributions at $X=x$ based on the distribution conditional on $|x-c|$ and then take the limit as $x \to c$ from either side, which would yield well-defined distributions by the maintained continuity assumptions.}
\begin{align*}
     d\PTarget_{CATE}(Y \mid X=c, T=t) & = d\PTrue( Y(t)  \mid X=c) \\
    d\PTarget_{CLATE}(Y \mid X=c, T=t) & = t d\PTrue( Y(1)  \mid X=c, M=0) + (1-t) d\PObs( Y \mid X = c^-) \\ 
    d\PTarget_{CATT}(Y \mid X=c, T=t) & = t d\PObs(Y  \mid X=c^+) + (1-t) d\PTrue( Y(0)  \mid X=c^-). 
\end{align*}
Then the target distribution for estimated $E \in \{CATE, CLATE, CATT\}$ is defined as $d\PTarget(X, T, Y) \equiv \PObs(X, T) d \PTarget_{E}(Y \mid X, T)$. 
%:
%\begin{enumerate}
%    \item Draw $(X, T, Y) \sim \PObs$
%    \item Depending on the estimand, proceed as follows:
%    \begin{itemize}
%        \item For the CLATE, if $T = 1$, re-draw $Y$ from the distribution of $Y(1) \mid X = c, M = 0$ under $\PTrue$
%        \item For the CATT, if $T = 0$, re-draw $Y$ from the distribution of $Y(0) \mid X=c^+$ under $\PTrue$ 
%        \item For the CATE, re-draw $Y$ from the distribution of $Y(T) \mid X = c$ under $\PTrue$
%    \end{itemize}
%    \item Return $(X, T, Y)$
%\end{enumerate}
The distribution of $R$ is the same under $\PTarget$ and $\PObs$ and $E_{\PTarget}[\lambda(R) Y]$ is equal to the target estimand for $\lambda(R)$ we define below.

We work under a relatively simple restriction on manipulation. Suppose we have a structural model that implies there are values $\Lambda_0, \Lambda_1 \geq 1$ such that for $t = 1, 0$:
\begin{align}
    \left. \frac{\PTrue( M = 1 \mid Y(t), X = c)}{\PTrue( M = 0 \mid Y(t), X = c)} \right/ \frac{\PTrue( M = 1 \mid X = c)}{\PTrue( M = 0 \mid X = c)} \in [\Lambda_t^{-1}, \Lambda_t]. \label{eq:RDManipulationBounds}
\end{align}
Define $\mathcal{M}_{\Lambda}(\Lambda_0, \Lambda_1)$ as the set of distributions $\Q$ over $(M, X(1), X(0), Y(1), Y(0), T)$ that marginalize to the distribution of $(X(M), Y(T), T = 1\{ X(M) > c\})$ under $\PTrue$, satisfy the restrictions of this section, and satisfy $\frac{\Q(M = 1 \mid Y(t), X=c)}{\Q(M = 0 \mid Y(t), X=c)} / \frac{\PTrue(M = 1 \mid X = c)}{\PTrue(M = 0 \mid X = c)} \in [\Lambda_t^{-1}, \Lambda_t]$ almost surely.\footnote{
To be precise, $\PTrue( M = 1 \mid X = c) = \PTrue( M = 1 \mid X = c^+) \PObs(X = c^+ \mid X = c) = \PTrue( M = 1 \mid X = c^+) \lim_{\varepsilon \to 0^+} \PObs( T = 1 \mid |X - c| \leq \varepsilon)$, which is well-defined by \Cref{assum:RD}.
}

The sensitivity parameter can be interpreted as the largest deviation in the odds ratio of manipulation given any potential outcome compared to the average. $\Lambda_t = 1$ corresponds to an unconfounded case in which the difference in regression values yields the CATE. $\Lambda_t = \infty$ corresponds to \cite{gerard2020bounds}'s assumption, in which the distribution of $Y(0) \mid M = 1, X=c$ is only constrained by the support of the potential outcome. In between, larger values of $\Lambda$ accommodate larger degrees of manipulation on potential outcomes.

\begin{proposition} \label{prop:RDEstimands}
In the RD case, our method can be implemented for the  CATE $E_{\PTrue}[Y(1) - Y(0) \mid X = c]$ with: 
\begin{align*}
    \lambda(R) & = \frac{1\{ T=1 \}}{\PObs(X = c^+ \mid X = c)} - \frac{1\{ T=0 \}}{\PObs(X = c^- \mid X = c)} \\
    \underline{w} (R) &= 1\{ T=1 \} \left[ \frac{1}{2-\ProbManip} + \frac{1-\ProbManip}{2-\ProbManip} \frac{1}{1-\ProbManip+\ProbManip\Lambda_1} \right] + 1\{ T=0 \} \left[ \left( 1 - \frac{\ProbManip}{2-\ProbManip}\right)  + \frac{\ProbManip}{2-\ProbManip} \Lambda_0^{-1} \right] \\
    \bar{w} (R) &= 1\{ T=1 \} \left[ \frac{1}{2-\ProbManip} + \frac{1-\ProbManip}{2-\ProbManip} \frac{1}{1-\ProbManip+\ProbManip\Lambda_1^{-1}} \right] + 1\{ T=0 \} \left[ \left( 1 - \frac{\ProbManip}{2-\ProbManip}\right) + \frac{\ProbManip}{2-\ProbManip} \Lambda_0 \right];
\end{align*} our method can be implemented  for the CLATE $E_{\PTrue}[ Y(1) \mid X = c, M = 0] - E_{\PObs}[ Y \mid X = c^-]$ with the same $\lambda(R)$ and:  
\begin{align*}
     \\
    \underline{w}(R) & = \frac{1\{ T=1 \}}{1 - \ProbManip + \ProbManip \Lambda_1} + 1\{ T=0 \}, \qquad \bar{w}(R) = \frac{1\{ T=1 \}}{1 - \ProbManip + \ProbManip \Lambda_1^{-1}} + 1\{ T=0 \};
\end{align*}
our method can be implemented for the CATT $E_{\PTrue}[Y(1) - Y(0) \mid X = c^+]$ with the same $\lambda(R)$ and: 
\begin{align*}
    \underline{w}(R) & = 1\{ T=1 \} + 1\{ T=0 \} \left( (1-\ProbManip) + \ProbManip \Lambda_0^{-1} \right), \\
    \bar{w}(R) & = 1\{ T=1 \} + 1\{ T=0 \} \left( (1-\ProbManip) + \ProbManip \Lambda_0 \right).
\end{align*}
\end{proposition}
The connection between \Cref{prop:RDEstimands} and the characterization from \Cref{lemma:EstimandExpressions} comes from correspondences that we derive in Appendix \Cref{lemma:ImpliedLikelihoodRatios}. Our CLATE bounds for $\Lambda_1=\infty$ are finite and identical to the bounds in \citet{gerard2020bounds} for sharp RD.

When $\Lambda_1$ is finite, we are also able to obtain meaningful CATE bounds. 

\begin{remark}
CATE bounds require some restriction on the untreated potential outcomes among the always-treated ($M=1$) population. As a result, we need some assumption to narrow these bounds beyond range restrictions. \Cref{eq:RDManipulationBounds} narrows these bounds by bounding the probability of manipulation conditional on potential outcomes, allowing us to learn about the potential outcomes of the $M=1$ observations from the untreated potential outcomes of the $M=0$ observations.
\end{remark}

Our analysis so far bounds manipulation on each potential outcome separately. There turns out to be no additional information available from a structural model that bounds manipulation on both potential outcomes simultaneously. 
\begin{proposition}\label{prop:RDSharpness}
    Suppose there is a finite $\Lambda \geq 1$ such that $\Lambda_1 = \Lambda_0 = \Lambda$. Let $\mathcal{M}_{\Lambda}'(\Lambda)$ be the set of distributions $\Q \in \mathcal{M}_{\Lambda}(\infty, \infty)$ satisfying: 
    \begin{align}
        \left. \frac{\Q( M = 1 \mid Y(1), Y(0), X = c)}{\Q( M = 0 \mid Y(1), Y(0), X = c)} \right/ \frac{\PTrue( M = 1 \mid X = c)}{\PTrue( M = 0 \mid X = c)} \in [\Lambda^{-1}, \Lambda]. 
        \label{eq:JointPOSelection}
    \end{align}
    Then $\mathcal{M}_{\Lambda}'(\Lambda) \subseteq \mathcal{M}_{\Lambda}(\Lambda, \Lambda)$. Further, take any distributions $\Q_1, \Q_0 \in \mathcal{M}_{\Lambda}(\Lambda, \Lambda)$ and write $\psi_t = E_{\Q_t}[Y(t) \mid X = c, M = 1-t]$. Then there is a distribution $\Q' \in \mathcal{M}_{\Lambda}'(\Lambda)$ satisfying $E_{\Q'}[Y(t) \mid X = c, M = 1-t] = \psi_t$ for $t = 1, 0$. 
\end{proposition}
\Cref{prop:RDSharpness} shows that the RD bounds are also sharp for constraints on the underlying manipulation model, i.e., separate bounds on both potential outcomes suffice to bound our structural estimands of interest. The proposition states that a distribution that satisfies the conditions of the model can be constructed to achieve the bounds. The quantities $\E_{\PTrue}[ Y(t) \mid X = c, M = 1-t]$ still identify our structural estimands of interest.
The high-level logic of \Cref{prop:RDSharpness} is fundamentally similar to Section \ref{sec:ipw}. Any given pair of Radon-Nikodym derivatives may not be simultaneously achievable for both potential outcomes. However, the pair corresponding to the worst-case bounds are simultaneously achievable. As a result, the identified set is the same whether we bound manipulation on one potential outcome, corresponding to the direct Radon-Nikodym bounds, or if we bound both potential outcomes simultaneously, corresponding to constraints on the underlying economic primitives.

\section{Estimation and Inference}\label{sec:EstimationAndInference}

We now estimate the bounds of \Cref{thm:IDSetUB} and construct confidence sets for the identified set. We focus on the upper bound throughout and write $\Psi \equiv \psi^{+}(\underline{w}, \bar{w})$; the lower bound is symmetric. We present three results in this subsection. First, the plug-in estimators converge to limits that are valid in that they either achieve the sharp bound or are conservative. Second, the \citet{imbens2004confidence} confidence set is valid when the sensitivity assumption is placed directly on the Radon-Nikodym derivative. Third, in the regression discontinuity context, where $w$'s are estimated, a tilting method achieves valid coverage.
%Three results follow. Plug-in estimation converges, and when the threshold quantile is estimated inconsistently the bounds err outward rather than inward. When the sensitivity assumption is placed directly on the Radon--Nikodym derivative, a bootstrap that fixes the estimated quantile across draws is asymptotically valid (though potentially too wide). When the weights must be estimated based on some underlying nuisance functions, validity is design specific; we give a procedure and a guarantee for RD under manipulation.

\subsection{Plug-in Estimation and One-Sided Robustness}

It is convenient to collect terms in \Cref{eq:UBFormula}. For a putative quantile function $\bar{Q}$, the pseudo-outcome $\phi^{+}(\underline{w}, \bar{w}, R, Y \mid \bar{Q})$ of \Cref{thm:IDSetUB} may be written as
\[
    \phi(\underline{w},\bar{w},\lambda,\bar{Q};Y) \equiv \underline{w}(R)\lambda(R)Y
    + \left(1-\underline{w}(R)\right)\bar{Q}(R)
    + \left(\bar{w}(R)-\underline{w}(R)\right)\{\lambda(R)Y - \bar{Q}(R)\}_{+}. 
\]
We write the $\bar{Q}$-implied upper bound as $\Psi(\bar{Q}) \equiv E_{\PObs}[ \phi(\underline{w},\bar{w},\lambda,\bar{Q};Y)]$, so that the sharp bound is $\Psi = \Psi(Q^{+}) \leq \Psi(\bar{Q})$. Hatted objects denote estimates and $\hat{E}$ the sample mean, so the plug-in estimator is $\hat{\Psi} \equiv \hat{E}[ \phi(\hat{\underline{w}},\hat{\bar{w}},\hat{\lambda},\hat{Q};Y)]$. 

\begin{manualassumption}{Estimation}\label{assum:Estimation}
Fix a nonrandom quantile function $\bar{Q}$, and write $M(R,Y)$ for $1 + |Y| + |\bar{Q}(R)|$.
\begin{enumerate}
    \item[(i)] \emph{(Sampling and stabilization.)} Sampling is iid, the sup-norm errors $\|\hat{\lambda}-\lambda\|_\infty$, $\|\hat{\underline{w}}-\underline{w}\|_\infty$ and $\|\hat{Q}-\bar{Q}\|_\infty$ are $\op(1)$, and so is $\|\hat{\bar{w}} \wedge K - \bar{w} \wedge K\|_\infty$ for each fixed $K < \infty$.
    \item[(ii)] \emph{(Moments.)} $\sup_r |\lambda(r)| < \infty$ and $E_{\PObs}[ M(R,Y) ] < \infty$.
    \item[(iii)] \emph{(Integrable weights.)} $\bar{w}(R) < \infty$ almost surely, $E_{\PObs}[ \bar{w}(R) M(R,Y) ] < \infty$, and
    \[
        \hat{E}\left[ \left| \hat{\bar{w}}(R) - \bar{w}(R) \right| M(R,Y) \right] = \op(1).
    \]
\end{enumerate}
\end{manualassumption}
The requirement on $\hat{Q}$ in (i) is that $\hat{Q}$ must tend to some fixed $\bar{Q}$, but $\bar{Q}$ need not equal $Q^{+}$. 
%\Cref{rem:StabilizationNeeded} in \Cref{sec:EstimationAppendix} shows the requirement cannot be dispensed with. 
Part (iii) is the only part that restricts $\bar{w}$, and it is used in the first part of the result below but not the second part.

\begin{proposition} \label{prop:consistency}
Suppose \Cref{assum:ObservableDistProperties}, \Cref{assum:QuantileMoments} and \Cref{assum:Estimation}\,(i)--(ii) hold.
\begin{enumerate}
    \item[(i)] \emph{(Integrable case.)} If \Cref{assum:Estimation}\,(iii) also holds, then $\Psi(\bar{Q})$ is finite, $\hat{\Psi} \xrightarrow{p} \Psi(\bar{Q})$, and $\Psi(\bar{Q}) \geq \Psi$ for every limit $\bar{Q}$. If in addition $\bar{Q} = Q^{+}$, then $\hat{\Psi} \xrightarrow{p} \Psi$.
    \item[(ii)] \emph{(Divergent case.)} If instead $\Psi(\bar{Q}) = +\infty$, then $\hat{\Psi} \xrightarrow{p} +\infty$: for every $M < \infty$, $\PObs(\hat{\Psi} > M) \to 1$. 
\end{enumerate}
\end{proposition}
Part (i) is a law of large numbers applied to the plug-in pseudo-outcome. 
To deal with possibly unbounded weights, \Cref{assum:Estimation}(iii) is stated in the averaged form and requires convergence. 
%Part (i) is a law of large numbers applied to the plug-in pseudo-outcome. Two features make that argument sufficient: $\bar{Q}$ enters $\phi$ only through an affine term and the hinge $\{\lambda(R)Y - \bar{Q}(R)\}_{+}$, so $\phi$ is Lipschitz in its arguments and the estimator has a kink rather than a discontinuity; and that Lipschitz property, with $\|\hat{Q}-\bar{Q}\|_\infty = \op(1)$, is what licenses replacing $\hat{Q}$ by $\bar{Q}$ inside the sample mean. The inequality $\Psi(\bar{Q}) \geq \Psi$ is \Cref{cor:Robustness}: an inconsistent quantile leaves the estimated bounds too wide for the sensitivity model rather than too narrow.
Part (ii) holds regardless of whether $\bar{Q}$ equals $Q^{+}$.

%The content of \Cref{assum:Estimation}(iii) is that the weight error is controlled in $L_1$ rather than in sup norm. A bounded $\bar{w}$ is the leading case, since a uniformly bounded $\hat{\bar{w}}$ and $\bar{w}$ give $\hat{E}[ | \hat{\bar{w}} - \bar{w} | M ] \leq \| \hat{\bar{w}} - \bar{w} \|_\infty \hat{E}[ M ] = \op(1) \Op(1)$; the weaker $L_1$ form is what an almost surely finite but unbounded $\bar{w}$ admits, as at the boundary $c = \eta_1$ of \Cref{cor:CDependenceIDdSet}. Part (ii) then covers models that admit unbounded reweighting, as c-dependence does past that boundary, and the two cases are mutually exclusive because (iii) forces $\Psi(\bar{Q})$ to be finite. Separating them by $\Psi(\bar{Q})$ rather than by $\bar{w}$ is what makes them fit together, since $\Psi(\bar{Q})$ is what the estimator targets: where $\bar{w}(R)$ is infinite, any $\bar{Q}$ short of the conditional essential supremum sends the hinge term to $+\infty$, so even a finite sharp bound is then estimated by a divergent sequence. The estimator diverges to its target rather than converging to a finite and incorrect limit, though for any fixed $n$ the realized $\hat{\Psi}$ is finite, since $\hat{\bar{w}}$ is finite in sample. The one configuration neither part covers is an unbounded $\bar{w}$ with a finite $\Psi(\bar{Q})$ but no $L_1$ control on the weight error; nothing can be said there, since the nuisance error is unrestricted in the direction that matters.

\subsection{Bootstrap Inference under a Direct Reweighting Assumption}\label{subsec:BootstrapDirect}

When the sensitivity assumption is placed directly on the Radon--Nikodym derivative, as in \Cref{assum:SensitivityModel}, rather than derived from a structural model, $\underline{w}(R)$, $\bar{w}(R)$ and $\tau(R)$ are known and the only estimated nuisance is the threshold quantile.
Write $\Psi^{+}( \bar{Q} ) \equiv E_{\PObs}[ \phi(\underline{w}, \bar{w}, \lambda, \bar{Q}; Y) ]$ for the upper bound implied by a putative quantile function $\bar{Q}$, write $\Psi^{-}( \bar{Q} )$ for its lower-bound counterpart, and let $\hat{\Psi}^{\pm} \equiv \hat{E}[ \phi^{\pm}(\underline{w}, \bar{w}, \lambda, \hat{Q}^{\pm}; Y) ]$ denote the two plug-in estimators.

For $b = 1, \ldots, B$, sample $n$ observations from the data with replacement, holding $\hat{Q} \equiv ( \hat{Q}^{+}, \hat{Q}^{-} )$ fixed at its full-sample value, and calculate $t^{\pm(b)} \equiv \sqrt{n} ( \hat{\Psi}^{\pm(b)} - \hat{\Psi}^{\pm} )$, where $\hat{\Psi}^{\pm(b)}$ uses the bootstrapped data.
Let $q^{\pm}_{\alpha}$ denote the $\alpha$ quantile of the bootstrap distribution of $t^{\pm(b)}$ and let $\hat{\sigma}^{\pm}$ denote the bootstrap standard deviation of $\hat{\Psi}^{\pm(b)}$.

We can construct two different forms of confidence sets, depending on whether we want to cover the identified set or we want to cover the true parameter. If we are interested in the identified set as a whole, then
\begin{align}
    CI(\alpha) & \equiv \left[ \hat{\Psi}^{-} - q^{-}_{1-\alpha} / \sqrt{n}, \ \hat{\Psi}^{+} - q^{+}_{\alpha} / \sqrt{n} \right], \label{eq:DirectCISet}
\end{align}
and is the union of two one-sided intervals, each at level $\alpha$.
If we are interested in the true parameter, then
\begin{align}
    CI^{IM}(\alpha) & \equiv \left[ \hat{\Psi}^{-} - r_\alpha( \hat{\rho} ) \, \hat{\sigma}^{-}, \ \hat{\Psi}^{+} + r_\alpha( \hat{\rho} ) \, \hat{\sigma}^{+} \right], & \hat{\rho} & \equiv \frac{ \hat{\Psi}^{+} - \hat{\Psi}^{-} }{ \max\{ \hat{\sigma}^{-}, \hat{\sigma}^{+} \} }, \label{eq:DirectCIIM}
\end{align}
where $r_\alpha(\rho)$ is the \citet{imbens2004confidence} critical value, defined as the unique solution of $\Phi( r + \rho ) - \Phi( -r ) = 1 - \alpha$.
It interpolates between $\Phi^{-1}(1-\alpha)$ when the identified set is wide relative to sampling uncertainty and the two-sided $\Phi^{-1}(1-\alpha/2)$ when the set collapses to a point.
%The distinction matters because the two objects are not nested guarantees of the same kind: $CI(\alpha)$ covers the whole interval at $1-2\alpha$, while $CI^{IM}(\alpha)$ covers each of its points at $1-\alpha$ and the interval itself only at $1-2\alpha$.

\begin{manualassumption}{Direct Reweighting}\label{assum:DirectReweighting}
The functions $\underline{w}$, $\bar{w}$ and $\lambda$ are known, with $\bar{w}(R) \leq \bar{W} < \infty$ for a fixed $\bar{W}$ and $\sup_r |\lambda(r)| < \infty$; $E_{\PObs}[Y^2] < \infty$; and the bootstrap holds $\hat{Q}$ fixed at its full-sample value across draws $b$, so that only $(R,Y)$ is resampled.
There are a class $\mathcal{Q}$ of quantile functions and fixed $\bar{Q}^{+}, \bar{Q}^{-} \in \mathcal{Q}$ such that
\begin{enumerate}
    \item[(i)] \emph{(Stabilization.)} $\| \hat{Q}^{\pm} - \bar{Q}^{\pm} \|_{L_2(\PObs)} = \op(1)$, and $\PObs( \hat{Q}^{\pm} \in \mathcal{Q} ) \to 1$;
    \item[(ii)] \emph{(Complexity.)} the class $\{ \phi(\underline{w}, \bar{w}, \lambda, Q; \cdot) : Q \in \mathcal{Q} \}$ is $\PObs$-Donsker, and likewise for $\phi^{-}$;
    \item[(iii)] \emph{(Nondegeneracy.)} $\Var_{\PObs}( \phi^{\pm}(\underline{w}, \bar{w}, \lambda, \bar{Q}^{\pm}; Y) ) > 0$;
    \item[(iv)] \emph{(Orthogonal.)} $\| \hat{Q}^{\pm} - Q^{\pm} \|_{L_2(\PObs)} = \op( n^{-1/4} )$ and the conditional density of $\lambda(R) Y$ given $R$ is bounded on a neighborhood of $Q^{\pm}(R)$.
\end{enumerate}
\end{manualassumption}
We show (i)--(iii) imply at least nominal coverage, and the addition of part (iv) implies asymptotically exact coverage. Parts (i)--(iii) restrict how $\hat{Q}$ is obtained but is agnostic about what it converges to; (iv) adds that it converges to the truth sufficiently quickly, which by (i) forces $\bar{Q}^{\pm} = Q^{\pm}$ to give exact coverage.
Without (iv) the bootstrap is centered at the bound implied by the realized $\hat{Q}$, which \Cref{cor:Robustness} places weakly above $\psi^{+}$, so coverage is at least nominal and tends to one whenever $\Psi^{+}(\bar{Q}^{+}) > \psi^{+}$.
The Donsker condition (ii) is weaker than \citet{dorn2022sharp}'s parametric linear quantile model, and can be bypassed under sample splitting.

%The Donsker condition \Cref{assum:DirectReweighting}(ii) is what the fixed quantile costs.
%It nests the parametric quantile model that \citet{dorn2022sharp} use for their inference result in that linear quantile regression as a special case.
%, and more generally any class of uniformly bounded functions of finite VC dimension or with square-integrable bracketing entropy, satisfies it.
%Quantile forests and nearest-neighbour estimators, which \citet{dorn2022sharp} admit for their consistency result, generally do not.
%Since $\phi$ is Lipschitz in its quantile argument with constant $\max\{ 1 - \underline{w}(R), \bar{w}(R) - 1 \} \leq \bar{W}$, \Cref{assum:DirectReweighting}(ii) follows from the Donsker property of $\mathcal{Q}$ itself.
%Cross-fitting is a standard procedure to allow for placing higher-level black-box conditions on nuisance estimation. The Donsker condition could be bypassed under cross-fitting. 

\begin{theorem}\label{prop:DirectBootstrap}
Suppose \Cref{assum:ObservableDistProperties}, \Cref{assum:QuantileMoments} and \Cref{assum:DirectReweighting}(i)--(iii) hold. Then
\begin{align*}
    \liminf_{n \to \infty} \PP\left( \psi^{+}( \underline{w}, \bar{w} ) \leq \hat{\Psi}^{+} - q^{+}_{\alpha} / \sqrt{n} \right) & \geq 1 - \alpha,
\end{align*}
with asymptotic equality under \Cref{assum:DirectReweighting}(iv).
The statement for $\psi^{-}$ is symmetric, so $\liminf_{n} \PP( [\psi^{-}, \psi^{+}] \subseteq CI(\alpha) ) \geq 1 - 2\alpha$.
If in addition \Cref{assum:DirectReweighting}(iv) holds and $\psi^{-}( \underline{w}, \bar{w} ) < \psi^{+}( \underline{w}, \bar{w} )$, then
\begin{align*}
    \liminf_{n \to \infty} \ \inf_{ \psi \in [\psi^{-}, \psi^{+}] } \PP\left( \psi \in CI^{IM}(\alpha) \right) & \geq 1 - \alpha.
\end{align*}
\end{theorem}
%Fixing $\hat{Q}$ across bootstrap draws is a computational convenience, and does not affect validity because the true quantile $Q^+$ is also a minimizer of the implied upper bound $\Psi^{+}(\bar{Q})$.
Fixing $\hat{Q}$ at the original estimated quantile speeds up computation substantially, but it only targets the sampling noise at the estimated quantile while leaving the systematic term $\Psi^{+}(\hat{Q}^{+}) - \psi^{+}$ in the limit, which \Cref{cor:Robustness} shows is conservative. Previous work has shown such validity under parametric \citep{dorn2022sharp} or orthogonality \citep{dvds,tan2024model} conditions. 
%This is the sense in which \citet{dorn2022sharp} observe that quantiles need not be re-estimated across bootstrap draws.
%The true quantile $Q^{+}$ is also the minimizer of the implied upper bound $\Psi^{+}(\cdot)$.
%, so quantile error is second order around the true quantile (\Cref{rem:SecondOrder}).
Neither approach allows an unbounded $\bar{w}$, which we allow here.

To show correct coverage under the \citet{imbens2004confidence} procedure, we use \Cref{assum:DirectReweighting}(iv) and condition $\psi^{-} < \psi^{+}$.
The \citeauthor{imbens2004confidence} calibration is a statement about two boundary estimators that are asymptotically normal around deterministic centers.
Under (i)--(iii), the natural centers are $\Psi^{\pm}(\hat{Q}^{\pm})$, which are random, so we use both consistency of quantiles and sufficient separation of the upper and lower bounds to ensure that we can replace them with the true centers to apply the \citeauthor{imbens2004confidence} argument. 
%and no rate is available to replace them by their limits at order $n^{-1/2}$.
%The calibration is also known to be delicate when the width of the identified set is small relative to sampling error (e.g., \citet{imbens2004confidence, stoye2009more}).
%The condition $\psi^{-} < \psi^{+}$ addresses the issue by forcing $\hat{\rho} \xrightarrow{p} \infty$, so that $r_\alpha(\hat{\rho}) \xrightarrow{p} \Phi^{-1}(1-\alpha)$.
%; \Cref{lemma:IMCalibration} in \Cref{sec:EstimationAppendix} states the calibration in the generality that \Cref{subsec:BootstrapForRD} needs, where the width does drift to zero.

In the next subsection we consider a specific, and relatively difficult, case in which $\hat{\underline{w}}$ and $\hat{\bar{w}}$ are estimated.
In many contexts, estimation error can be incorporated either through a parametric bootstrap \citep{zhao2019sensitivity} or through Neyman orthogonalization.
We follow the parametric strategy in our next subsection.
A generic bootstrap procedure would have to allow extreme quantiles and kink points, as under \citet{masten2018identification}'s conditional c-dependence, and is outside this paper's scope.

\subsection{Bootstrap Inference for Regression Discontinuity under Manipulation}\label{subsec:BootstrapForRD}

The RD bounds are not a corollary of \Cref{prop:DirectBootstrap} as the weight constraints $\underline{w}$ and $\bar{w}$ are estimated.
The weights of \Cref{prop:RDEstimands} depend on the manipulation share $\ProbManip$, which is recovered from a jump in the density of the running variable \citep{mccrary2008manipulation}.
Since $\ProbManip$ is estimated, the premise of \Cref{assum:DirectReweighting} is unavailable.
%There are then two further complications. First, $\eta$ is estimated at a nonparametric rate $r_n \ll \sqrt{n}$, so there is no $\sqrt{n}$ expansion of the weights to feed into the argument behind \Cref{prop:DirectBootstrap}.
A further complication is that the parameter is non-regular at $\ProbManip = 0$, where the truncation $\hat{\ProbManip} = \max\{ \tilde{\ProbManip}, 0 \}$ has a kink, so the bootstrap distribution of the bound estimator is not consistent.
To address this issue, we incorporate slowly-estimated manipulation probabilities by considering a more conservative bootstrap, mimiacking the procedure of \citet{gerard2020bounds} that ``tilts'' the confidence set. 
%, and use the bootstrap only for a scale and a recentering, never for a quantile of the bound estimator itself.

Let $\kappa_n = \log(n)^{1/2}$, $\tilde{\ProbManip} = 1 - \hat{f}^- / \hat{f}^+$ and $\hat{\ProbManip} = \max\{ \tilde{\ProbManip}, 0 \}$, where $\hat{f}^-$ and $\hat{f}^+$ are the estimated densities of the running variable just below and above the threshold. We implement the following procedure:
\begin{enumerate}
    \item Generate bootstrap samples $\{ Y_{i,b}, T_{i,b}, X_{i,b} \}_{i=1}^n$, $b=1,\cdots,B$ by sampling with replacement from the original data.
    \item Calculate $\tilde{\ProbManip}^*_b = 1 - \hat{f}^-_b / \hat{f}^+_b$, and let $\hat{\sigma}_{\tilde{\ProbManip}}$ denote the sample standard deviation of $\{ \tilde{\ProbManip}^* _b\}_{b=1}^B$.
    \item Calculate $\tilde{\ProbManip}_b = \tilde{\ProbManip}^*_b - \tilde{\ProbManip} + \max \{ \tilde{\ProbManip}, \kappa_n \hat{\sigma}_{\tilde{\ProbManip}} \}$ and $\hat{\ProbManip}_b = \max \{ \tilde{\ProbManip}_b, 0 \}$.
    \item For $s \in \{ +, - \}$, calculate $\hat{\psi}^s(\hat{\ProbManip}_b)$ on bootstrap sample $b$ through appropriately smoothed local polynomial regression, holding the estimated threshold quantile function fixed at its full-sample value and re-evaluating it at the level $\tau(\hat{\ProbManip}_b)$ that the draw implies. Let $\hat{\sigma}^{s}$ denote the sample standard deviation of $\{ \hat{\psi}^s (\hat{\ProbManip}_b) \}_{b=1}^B$.
\end{enumerate}
%Step 4 fixes the quantile for the reason given in \Cref{subsec:BootstrapDirect}.
Define $\hat{\Gamma}^{L*}$ and $\hat{\Gamma}^{U*}$ exactly as $\hat{\psi}^{-}$ and $\hat{\psi}^{+}$ except that $\hat{\ProbManip}^{*} = \max \{ \tilde{\ProbManip} , \kappa_n \hat{\sigma}_{\tilde{\ProbManip}} \}$ replaces $\hat{\ProbManip}$, and write $\hat{\Delta} \equiv \hat{\Gamma}^{U*} - \hat{\Gamma}^{L*}$. The confidence set is
\begin{align}
\mathcal{C}^{SRD}_{1-\alpha} &= \left[ \hat{\Gamma}^{L*} - r_\alpha( \hat{\rho} ) \cdot \hat{\sigma}^L, \ \hat{\Gamma}^{U*} + r_\alpha( \hat{\rho} ) \cdot \hat{\sigma}^U \right], & \hat{\rho} & \equiv \hat{\Delta} \big/ \max \{ \hat{\sigma}^L , \hat{\sigma}^U \} , \label{eq:RDConfidenceSet}
\end{align}
with $r_\alpha(\cdot)$ the \citet{imbens2004confidence} critical value of \Cref{eq:DirectCIIM}.
The floor $\max\{\tilde{\ProbManip}, \kappa_n \hat{\sigma}_{\tilde{\ProbManip}}\}$ prevents the estimated width from shrinking faster than the sampling error of $\tilde{\ProbManip}$: because $\kappa_n \to \infty$ while the bootstrap dispersion of $\tilde{\ProbManip}^*_b - \tilde{\ProbManip}$ is of order $\hat{\sigma}_{\tilde{\ProbManip}}$, the truncation $\hat{\ProbManip}_b = \max\{\tilde{\ProbManip}_b, 0\}$ is inactive for every $b$ with probability tending to one, which restores the differentiability.
%whose absence would otherwise call for the corrections of \citet{FangAndSantos}.

For $\ProbManip' \in [0,1)$ write $\underline{w}_{\ProbManip', \Lambda}$ and $\bar{w}_{\ProbManip', \Lambda}$ for the weight bounds of \Cref{prop:RDEstimands} with $\ProbManip$ replaced by $\ProbManip'$, and abbreviate the sharp identified set they induce by $\mathcal{I}_{\ProbManip', \Lambda} \equiv \mathcal{I}( \underline{w}_{\ProbManip', \Lambda}, \bar{w}_{\ProbManip', \Lambda} )$, in the notation of \Cref{def:IdentifiedSet}.
The target is $\mathcal{I}_{\ProbManip, \Lambda}$, the identified set at the true manipulation share and the assumed $\Lambda$; by \Cref{thm:IDSetUB} and \Cref{prop:RDEstimands} it contains the structural estimand $\psi_0$.

\begin{theorem}\label{prop:RDBootstrap}
Suppose \Cref{assum:RDInference} in \Cref{sec:EstimationAppendix} holds. Then
\begin{enumerate}
    \item[(i)] \emph{(Points.)} $\displaystyle \liminf_{n \to \infty} \ \inf_{ \psi \in \mathcal{I}_{\ProbManip, \Lambda} } \PP\left( \psi \in \mathcal{C}^{SRD}_{1-\alpha} \right) \geq 1 - \alpha$, and in particular $\liminf_{n} \PP( \psi_0 \in \mathcal{C}^{SRD}_{1-\alpha} ) \geq 1-\alpha$.
    \item[(ii)] \emph{(Set.)} $\displaystyle \liminf_{n \to \infty} \PP\left( \mathcal{I}_{\ProbManip, \Lambda} \subseteq \mathcal{C}^{SRD}_{1-\alpha} \right) \geq 1 - 2\alpha$, and the bound improves to $1-\alpha$ if $r_\alpha(\hat{\rho})$ is replaced by $\Phi^{-1}(1-\alpha/2)$ in \Cref{eq:RDConfidenceSet}.
\end{enumerate}
\end{theorem}
Both conclusions hold regardless of whether the limit quantile function $\bar{Q}$ of \Cref{assum:Estimation}\,(i) equals the correct threshold quantile, and along sequences $\ProbManip_n \to 0$.
The reported interval can be conservative because the manipulation share is deliberately tilted above the truth, and a threshold quantile that need not be consistent.
Both errors, when incurred, are conservative. \Cref{assum:RDInference}(iii) attaches a central limit theorem to these conservative estimands. 
%so each carries $\mathcal{I}_{\ProbManip, \Lambda}$ into a strictly larger interval, and the largest of these is what \Cref{assum:RDInference}(iii) attaches a central limit theorem to.
%Neither error has to be estimated or even to be estimable, which is what makes the argument available when $\hat{\ProbManip}$ is non-regular.
The two parts of the theorem on covering the true parameter pointwise and on covering the identified set mimic the statements of \Cref{prop:DirectBootstrap}.
%are conservative for different reasons.
%Part (i) is calibrated, so its only slack is the slack in the ordering, and coverage tends to one as either distortion grows.
%Part (ii) is not calibrated at all: the \citet{imbens2004confidence} critical value is designed for points, while covering an interval requires protecting both endpoints at once.

Compared with \Cref{prop:DirectBootstrap}, we use the bootstrap only to obtain a scale $\hat{\sigma}^{j}$ and a recentering. We do not use the bootstrap to obtain a quantile of the bound estimator, so consistency of the bootstrap distribution is not required and the kink at $\ProbManip = 0$ is harmless.
%Further, in \Cref{prop:DirectBootstrap} a positive identified width forces $\hat{\rho} \to \infty$ and $r_\alpha(\hat{\rho}) \to \Phi^{-1}(1-\alpha)$, so the interpolation does nothing in the limit; here the width may drift to zero with $\ProbManip_n$, the interpolation is what returns the correct critical value, and \Cref{assum:RDInference}(iii) makes $\hat{\rho}$ converge.

\section{Monte Carlo Evidence on Inference}
\label{sec:simulation}

We assess the finite-sample behavior of the tilted-bootstrap confidence set of \Cref{sec:EstimationAndInference} for the CLATE in sharp RD with one-sided manipulation. The CLATE is the estimand that \citet{gerard2020bounds} (henceforth GRR) bound, so the $\Lambda = \infty$ endpoint of our procedure is exactly their worst-case bound and the comparison between a worst case and a finite sensitivity restriction holds the estimand fixed.
We draw a full running variable instead of the limiting cutoff distributions, so that $\ProbManip$ must be recovered from an estimated density jump and the cutoff outcome limits from a bandwidth.

%A condition used to establish validity is that the estimated manipulation share is asymptotically normal at some rate with a consistent bootstrap scale. 
 %The data still come from the boundary of the sensitivity model, and the parameter still sits at an endpoint of the identified set.

\paragraph{Design.}
Non-manipulators draw a running variable with density $(1 + \beta x)/2$ on $[-1, 1]$, with $\beta = 0.4$, so the density has nonzero slope at the cutoff and a local constant estimator would be biased there. Manipulators draw an exponential running variable with rate $6$ truncated to $(0, 1]$, so their mass piles up just above the cutoff and decays smoothly; their population share is set to deliver a target cutoff share $\ProbManip = 1 - f^-(c)/f^+(c)$. Outcomes are $Y(1) = 0.5 - 0.4 X + U_1$ and $Y(0) = 0.6 X + U_0$ with $U_0$ standard normal and independent of manipulation, so the two sides carry different slopes and the untreated cutoff limit is unaffected by selection.

Manipulation selects on $Y(1)$ at the boundary of the sensitivity model: $\PTrue(M = 1 \mid Y(1), X = c)$ takes two values, with odds $\Lambda^{\text{true}}$ times the baseline below a threshold and $(\Lambda^{\text{true}})^{-1}$ times the baseline above. The threshold that makes the two-point rule consistent with $\PTrue(M = 1 \mid X = c) = \ProbManip$ is exactly the threshold probability $\tau(R)$ of \Cref{thm:IDSetUB}, and the implied reweighting of the treated cutoff law is exactly the pair $(\underline{w}, \bar{w})$ of \Cref{prop:RDEstimands}. Consequently $\E_{\PTrue}[Y(1) \mid X = c, M = 0]$ equals the upper endpoint of the identified set, the identified set is $(\mu_+ - \mu_-) \pm \sigma_1 \varphi(\Phi^{-1}(\tau))(\bar{w} - \underline{w})$ in closed form, and pointwise coverage is evaluated at a boundary point. Negating $U_1$ mirrors the design and places the parameter at the lower endpoint instead.

\paragraph{Estimation.}
The manipulation share is estimated by the two-step procedure of \citet{mccrary2008manipulation}: fine bins that do not straddle the cutoff, then one-sided local linear regression of normalized bin heights on bin midpoints, giving $\hat{f}^{\pm}$ and $\tilde{\eta} = 1 - \hat{f}^-/\hat{f}^+$. The cutoff outcome limits are one-sided local linear regressions with a triangular kernel and bandwidth proportional to $n^{-1/4}$ for undersmoothing to obtain valid inference. The threshold quantile is the kernel-weighted quantile of treated-side local linear residuals, re-evaluated at the $\hat{\tau}$ each bootstrap draw implies but otherwise fixed at its full-sample value.

%Two implementation choices are dictated by the theory rather than by convenience. First, $\Lambda$ is held fixed across the tilt. For the CLATE the untreated branch of \Cref{prop:RDEstimands} has $\underline{w} = \bar{w} = 1$, so the pseudo-outcome there reduces to the untreated cutoff mean and $\lambda$ cancels from the estimator; the bound estimator is therefore invariant to $\hat{\eta}_b$ through $\lambda$. Second, the bootstrap is used only for $\hat{\sigma}_{\tilde{\eta}}$, $\hat{\sigma}^{L}$ and $\hat{\sigma}^{U}$, never for a quantile of the bound estimator itself.
While we set $n = 2000$, we vary $\ProbManip \in \{0,\, 0.05,\, 0.15\}$; and $\Lambda^{\text{true}} \in \{1,\, 3,\, \infty\}$, with $1000$ replications and $500$ bootstrap draws per cell. 
%When $\ProbManip = 0$, $\Lambda^{\text{true}}$ is not identified, so we vary the assumed $\Lambda$ there instead, which is the case of a researcher hedging against manipulation that is absent.

%\paragraph{What a cutoff-limit design cannot test.}
%Three things become live once the running variable is drawn rather than conditioned away. The estimated share converges at a nonparametric rate, so $r_n \ll \sqrt{n}$ in \Cref{assum:RDInference} and $\hat{\sigma}_{\tilde{\eta}}$ is an order of magnitude larger than a parametric standard error at the same $n$: in our design its Monte Carlo standard deviation is $0.17$ at $n = 2000$ and still $0.04$ at $n = 50{,}000$. The estimator also carries boundary bias, and the sign of that bias is what matters. A downward-biased $\hat{\ProbManip}$ shrinks the weight interval by \Cref{lemma:RDMonotone} and therefore narrows the bounds, which is the anti-conservative direction; we find a downward bias of roughly $0.03$ at $n = 2000$. And $\ProbManip = 0$, which a design drawing from the cutoff limits cannot represent without degenerating, is exactly the configuration in which $\tilde{\eta}$ is non-regular and the truncation in $\hat{\eta} = \max\{\tilde{\eta}, 0\}$ has a kink.

%\paragraph{Coverage notions.}
We implement the bootstrap from \Cref{subsec:BootstrapForRD}. 
We report coverage of the identified set, the probability that the confidence set contains all of $[\psi^-, \psi^+]$, and coverage of the parameter, the probability that it contains $\psi_0$. 
%The \citet{imbens2004confidence} critical value is calibrated for points of the identified set, interpolating between a two-sided margin when the set is short and a one-sided margin when it is wide. Our appendix guarantees the nominal level for points and only $1 - 2\alpha$ for the set, so the set-coverage column is reported as a descriptive summary rather than as a quantity the calibration targets. Since the design places $\psi_0$ at an endpoint, the parameter column tests the calibration exactly where it binds.

\paragraph{Results.}
\Cref{tab:mc_rd_coverage} reports coverage of the identified set and the parameter when $\Lambda$ is correctly specified in our method together with the GRR limit. 
%Three findings emerge.

First, coverage is at or above the nominal level in every cell up to Monte Carlo error, and it is close to nominal where the theory says it should be. The point-identified cells, $\Lambda = 1$ or $\ProbManip = 0$, attain coverage within Monte Carlo error of $0.95$; these are the cells in which the identified set is degenerate, the critical value reverts to its two-sided value, and none of the outward distortions operate. Where the identified set has positive width, coverage is conservative. 
The conservativeness is attributed to the floor from tilting $\ProbManip$: omitted simulations that fix $\ProbManip$ at the true value but estimate $\hat{Q}$ achieve near-exact coverage.

Second, the width of the confidence set is substantially smaller than that of GRR across the designs (with the exception of when $\Lambda = \infty$ where they are equal). This phenomenon occurs even when $\eta =0$ as the estimated $\hat \eta$ may be nonzero, and GRR adversarially accounts for the worst-case inference in light of sampling noise. 

%Second, the floor is \emph{necessary} rather than merely conservative. Setting $\kappa_n = 0$ removes the only protection against the downward bias in $\hat{\ProbManip}$, which by \Cref{lemma:RDMonotone} translates directly into bounds that are too narrow, and identified-set coverage falls below the nominal level. This is a stronger justification for the device than the one usually given, which is that it prevents the estimated width from shrinking faster than the sampling error of $\hat{\ProbManip}$.

%Third, the one-sided robustness of \Cref{cor:Robustness} is confirmed where it can be. Replacing the threshold quantile with a deliberately inconsistent one widens the estimated bounds and leaves coverage at or above nominal, so quantile misspecification costs precision and not validity. The mirrored design, in which $\psi_0$ sits at the lower endpoint, behaves the same way, so the guarantee is not an artifact of stressing only one side.

\begin{table}[ht]
\centering
\caption{Monte Carlo coverage and median confidence-set width for the tilted-bootstrap procedure of \Cref{sec:EstimationAndInference} applied to the CLATE in sharp RD with one-sided manipulation. CLATE coverage refers to coverage of the upper endpoint of the identified set, while set coverage refers to coverage of the full partially identified set. The ``GRR'' columns report the exact worst-case $\Lambda = \infty$ limit of the procedure. 
%When $\ProbManip = 0$ the weights equal one for every $\Lambda$, so the row varies the assumed $\Lambda$ instead. 
Based on $1000$ replications with $500$ bootstrap draws each.}
\label{tab:mc_rd_coverage}
\begin{tabular}{rrrrrrrr}
\toprule
& & \multicolumn{2}{c}{Point calibration} & \multicolumn{2}{c}{Set calibration} & \multicolumn{2}{c}{GRR, set calibration} \\
\cmidrule(lr){3-4} \cmidrule(lr){5-6} \cmidrule(lr){7-8}
$\eta$ & $\Lambda^{\text{true}}$ & CLATE cov. & Width & Set cov. & Width & Set cov. & Width \\
\midrule
0.00 & 1 & 0.963 & 0.987 & 0.963 & 0.987 & 0.999 & 4.540 \\
0.05 & 1 & 0.947 & 0.969 & 0.947 & 0.969 & 1.000 & 4.399 \\
0.05 & 3 & 1.000 & 2.436 & 1.000 & 2.686 & 1.000 & 4.429 \\
0.05 & $\infty$ & 1.000 & 3.969 & 1.000 & 4.333 & 1.000 & 4.333 \\
0.15 & 1 & 0.946 & 0.964 & 0.946 & 0.964 & 1.000 & 4.086 \\
0.15 & 3 & 1.000 & 2.117 & 1.000 & 2.341 & 0.999 & 4.090 \\
0.15 & $\infty$ & 1.000 & 3.689 & 1.000 & 4.011 & 1.000 & 4.011 \\
\bottomrule
\multicolumn{8}{p{0.92\textwidth}}{\footnotesize Nominal level 0.95; 1000 replications, 500 bootstrap draws each; Monte Carlo standard error at most 0.007. Both set-coverage columns are scored against the identified set at $\Lambda^{\text{true}}$.} \\
\end{tabular}

\end{table}

\section{Effects of Council Size}

\citet{egger2010government} study the effect of council size on total expenditure in Bavaria. 
They apply a sharp RD design: council size is a deterministic step-wise function of the population size of the municipality (see their Table 1).
In the main specification (their Table 3), they studied the effect of crossing the population threshold (T) on the log of total expenditure (Y), using the log population (X) as the running variable. 
Using their replication data, we estimate that crossing the threshold increases expenditure by 8\% for the 30\% window, which is smaller than the authors' estimates as we do not rely on a polynomial approximation.\footnote{
Our method compares the difference in means above and below the threshold so that it is immediately comparable to our benchmark in \citet{gerard2020bounds}.
}

Manipulation has been documented in this context by \citet{eggers2018regression}. Manipulation  can occur for several reasons in this context, including selective precision (e.g., ordering extra checks to ensure new arrivals are properly processed), and strategic recruitment (e.g., encouraging friends to change their address to the municipality).
The manipulation also occurs in one direction, as mayors have little incentive to have smaller council sizes, given their aforementioned reasons for manipulation. With our given bandwidth, we estimate $\hat{\eta} = 0.11$.
Inference follows \Cref{subsec:BootstrapForRD}.
We take $\kappa_n = \log(n)^{1/2}$, estimate $\hat{\sigma}_{\tilde{\ProbManip}}$ as the standard deviation of $\tilde{\ProbManip}$ across $500$ bootstrap draws, tilt the share to $\hat{\ProbManip}^{*} = \max\{ \tilde{\ProbManip}, \kappa_n \hat{\sigma}_{\tilde{\ProbManip}} \}$, and report the confidence set $\mathcal{C}^{SRD}_{0.95}$ of \Cref{eq:RDConfidenceSet}.
%The threshold quantile is a linear quantile regression of $\lambda(R) Y$ on the treatment indicator and the running variable, held fixed at its full-sample value across draws as \Cref{subsec:BootstrapDirect} licenses; by \Cref{cor:Robustness} an inconsistent quantile costs width rather than coverage here.

\begin{figure}
    \centering
    \includegraphics[width=.9\textwidth]{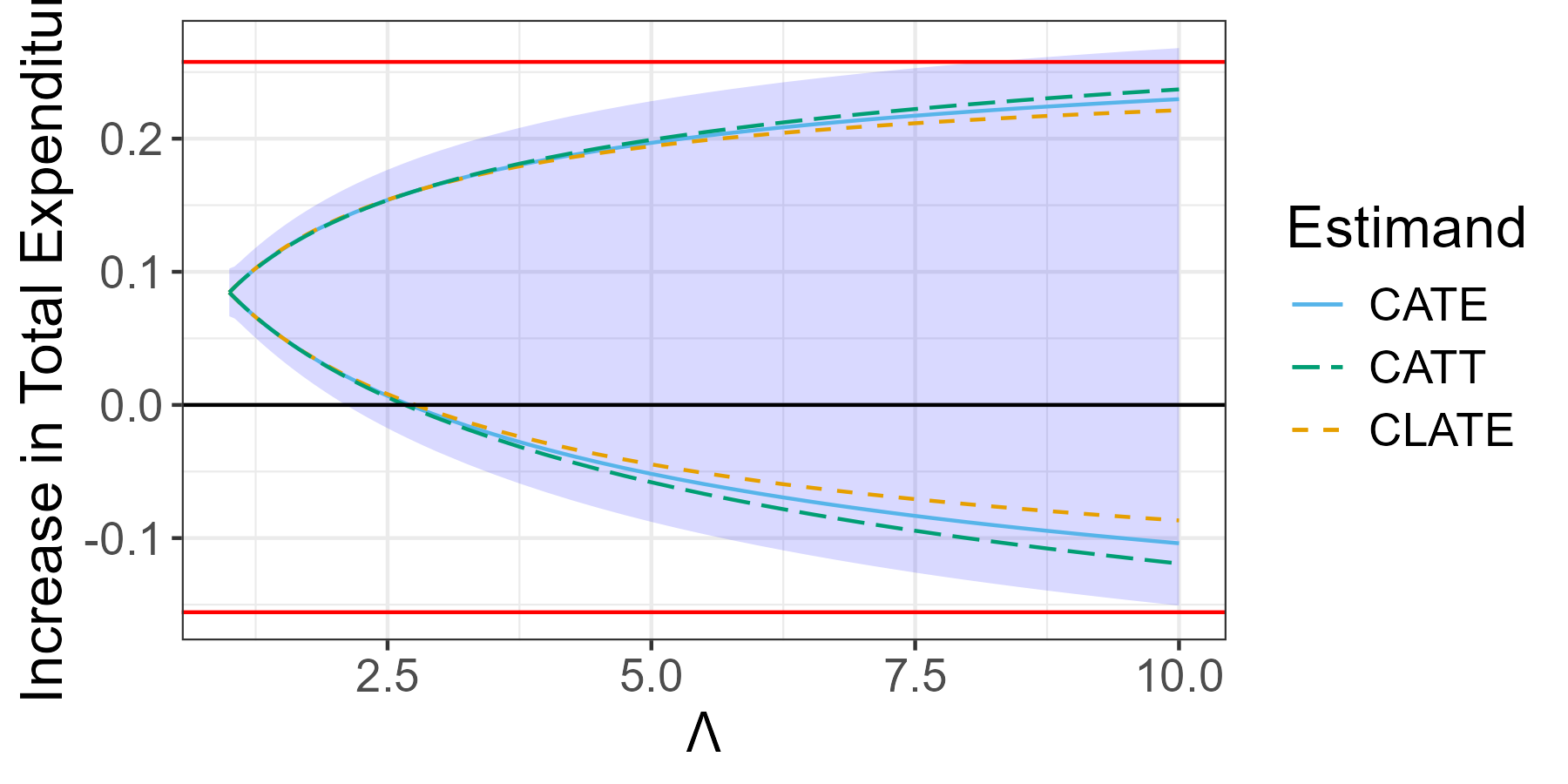}
    \caption{The upper and lower bounds of the shaded region above and below the CATE estimate indicates the bounds of the estimated 95\% confidence set. The horizontal red lines denote the worst-case GRR bounds for CLATE. $n =44,276$.}
    \label{fig:council_app}
\end{figure}

We find similar conclusions whether we study the CLATE (as in GRR) or the CATE or CATT (as our method allows). Worst-case bounds like those of GRR would automatically overturn the results of the study, because worst-case bounds are extremely uninformative. However, our sensitivity analysis provides a more nuanced perspective on the significance of the effects: at small $\Lambda$ values, the effect on total expenditure remains significant and positive: if we believe that the odds ratio of manipulation for any given expenditure after crossing the threshold is not more than 2.71, then the estimated lower bounds for CLATE and CATE would remain above 0. We estimate confidence intervals for the CATE: the lower bound would cross zero at $\Lambda = 2.1$.
In this context, small $\Lambda$ values may be defensible if the higher expenditures (Y) do not translate directly to, say, the mayors' salaries which would incentivize the manipulation.

\section{Conclusion}

This paper proposes a novel sensitivity analysis framework for identification failures for linear estimators. By placing bounds on the distributional distance between the observed distribution and a target distribution that identifies the structural parameter of interest, we obtain sharp and tractable analytic bounds. We provide new results on sharp and valid sensitivity analysis that allow even unbounded likelihood ratios. We illustrate how our framework and partial identification results contribute to important applications, including a new procedure for sensitivity analysis for the CATE under RD with manipulation.

\bibliographystyle{plainnat}
\bibliography{linear_est}

%%%%%%%%%%%%%%%%%%%%%%%%%%
%%%%%%%%%%%%%%%%%%%%%%%%%%
\newpage

\appendix

\section{Additional Assumptions and Results} \label{sec:additional-results}

\begin{manualassumption}{RD}\label{assum:RD} 
    We extend \citeauthor{gerard2020bounds}'s continuity condition and assume that there is a well-defined conditional distribution function $\PTrue( Y(1), Y(0) \mid X, M )$ such that  for each $y_1, y_0$, $\PTrue( Y(1) \leq y_1, Y(0) \leq y_0 \mid X = x, M = 0 )$ can be defined as continuous in $x$ and that the derivative of $\PTrue(X \leq x \mid M = 0)$ is continuous at $x = c$. We further assume that $\PTrue( Y(1) \leq y_1, Y(0) \leq y_0 \mid X = x, M = 1 )$ is right-continuous $x$, assume that $\PTrue( M = 1 \mid X = x)$ is right-continuous at $x = c$, assume that $X$ admits a density $f$ with respect to Lebesgue measure on a neighborhood of $c$ whose one-sided limits $f(c^-)$ and $f(c^+)$ exist with $f(c^+) > 0$, assume that marginals $\PTrue( Y(t) \leq y_t \mid X=x, M=m)$ are right-continuous for $m=1$ (continuous for $m=0$) at $X=c$ and define conditional distributions at $X = c$, $X = c^+$, and $X = c^-$ using the appropriate limits, such as $\PObs(X = c^+ \mid X = c) = \lim_{\varepsilon \to 0^+} \PObs( T = 1 \mid |X - c| \leq \varepsilon)$.
    %$\PTrue( M = 1 \mid X = c) = \PTrue( M = 1 \mid X = c^+) \PObs(X = c^+ \mid X = c)$ is well-defined.
    %in \Cref{def:FormalLinfCondProbs}.
\end{manualassumption}

The following claim shows that the conditional c-dependence identified set is finite and bounded.

\begin{proposition}\label{cor:CDependenceIDdSet}
Suppose $Y \mid X, Z \sim \mathcal{N}(\mu(X, Z), \sigma(X, Z)^2)$, the support of the observed propensity function $e(X)$ is the closed interval $[\eta_1, 1-\eta_2] \subset (0, 1)$ (where $\eta_1, \eta_2$ are fixed finite constants), the conditional outcome variance $\sigma(X, Z)$ is positive and bounded, and $e(X)$ admits a density $f_{e(X)}$ bounded above by $\bar{f}_{e(X)} < \infty$. Then there is a finite $B > 0$ such that the ATE identified set is a subset of $[ E[\mu(X, 1) - \mu(X, 0)] - B, E[\mu(X, 1) - \mu(X, 0)] + B]$ for all $c \leq \min\{ \eta_1, \eta_2 \}$, and is $(-\infty, \infty)$ for all $c > \min\{ \eta_1, \eta_2 \}$.
\end{proposition}

For exclusion failure, suppose we have a structural model that implies there are functions $\ell_{Nt}(X), u_{Nt}(X)$, $\ell_{At}(X), u_{At}(X)$, $\ell_{Co}^{1}(X), u_{Co}^{1}(X)$, $\ell_{Co}^0(X), u_{Co}^0(X)$ such that:
\begin{align*}
    \ell_{Nt}(X) & \leq \frac{d \PTrue( Y(0, 0) \mid Nt, X)}{d \PTrue( Y(0, 1) \mid Nt, X)} \leq u_{Nt}(X), \qquad \ell_{At}(X) \leq \frac{d \PTrue( Y(1, 1) \mid At, X)}{d \PTrue( Y(1, 0) \mid At, X)} \leq u_{At}(X), \\
    \ell_{Co}^1(X) & \leq \frac{d \PTrue( Y(1, 0) \mid Co, X)}{d \PTrue( Y(1, 1) \mid Co, X)} \leq u_{Co}^1(X), \qquad \ell_{Co}^0(X) \leq \frac{d \PTrue( Y(0, 1) \mid Co, X)}{d \PTrue( Y(0, 0) \mid Co, X)} \leq u_{Co}^0(X).
\end{align*}
The assumption stated in the main text is a special case of this setting where $\ell_{Nt}(X) =\ell_{At}(X)=\ell_{Co}^{1}(X)=\ell_{Co}^0(X) = \ell(X)$ and $u_{Nt}(X) = u_{At}(X)= u_{Co}^{1}(X) = u_{Co}^0(X) = u(X)$.
The following lemma shows how $\frac{d \PTarget( Y \mid X, T=t, Z=z)}{d \PObs( Y \mid X, T=t, Z=z)} $ can be written as objects that we bound above.

\begin{lemma} \label{lem:excl_bound_transform}
In the Instrumental Variables application in \Cref{sec:IVExclusion}, the following decompositions hold:
\footnotesize 
\begin{align*}
     & \frac{d\PTarget(Y|X,T=0,Z=0)}{d\PObs(Y|X,T=0,Z=0)}  =\omega(0 \mid X)+\omega(1 \mid X)  \frac{\PObs(Nt \mid X)}{\PObs(Co \mid X) + \PObs(Nt \mid X)} \frac{d\PObs(Y\mid X,T=0,Z=1)}{d\PObs(Y\mid X, T=0,Z=0)}\\
    & \quad+ \bigg\{ \omega(1 \mid X)\frac{d\PTrue(Y(0,1)|X,Co)}{d\PTrue(Y(0,0)|X,Co)} \\
    &\quad \times \left(1-
    \frac{\PObs(Nt \mid X)}{\PObs(Co \mid X) + \PObs(Nt \mid X)} \frac{d\PTrue(Y(0,0)|X,Nt)}{d\PTrue(Y(0,1)|X,Nt)}\frac{d\PObs(Y\mid X,T=0,Z=1)}{d\PObs(Y\mid X,T=0,Z=0)}\right) \bigg\} \\
     & \frac{d\PTarget(Y|X,T=1,Z=1)}{d\PObs(Y|X,T=1,Z=1)}  =\omega(1 \mid X)+\omega(0 \mid X)  \frac{\PObs(At \mid X)}{\PObs(Co \mid X) + \PObs(At \mid X)} \frac{d\PObs(Y\mid X,T=1,Z=0)}{d\PObs(Y\mid X,T=1,Z=1)}\\
    & \quad+ \bigg\{ \omega(0 \mid X)\frac{d\PTrue(Y(1,0)|X,Co)}{d\PTrue(Y(1,1)|X,Co)}\\
    &\quad \times \left(1 - \frac{\PObs(At \mid X)}{\PObs(Co \mid X) + \PObs(At \mid X)}  \frac{d\PTrue(Y(1,1) \mid X, At)}{d\PTrue(Y(1,0) \mid X, At)}\frac{d\PObs(Y\mid X,T=1,Z=0)}{d\PObs(Y\mid X,T=1,Z=1)}\right) \bigg \} \\
    & \frac{d \PTarget(Y \mid X, T = 0, Z = 1)}{d \PObs(Y \mid X, T = 0, Z = 1)} = \omega(1 \mid X) + \omega(0 \mid X) \frac{d \PTrue( Y(0, 0) \mid X, Nt )}{d \PTrue( Y(0, 1) \mid X, Nt)} \\
    & \frac{d \PTarget(Y \mid X, T = 1, Z = 0)}{d \PObs(Y \mid X, T = 1, Z = 0)} = \omega(0 \mid X) + \omega(1 \mid X) \frac{d \PTrue( Y(1, 1) \mid X, At )}{d \PTrue( Y(1, 0) \mid X, At)},
\end{align*}
\normalsize 
where for values of $X$ with $\PObs(T, Z \mid X) = 0$, we write $\frac{d \PTarget(Y \mid X, T, Z)}{d \PObs(Y \mid X, T, Z)} = 1$. 
\end{lemma}

%\subsection{Some Other Applications}\label{sec:OtherApplications}

%This subsection briefly mentions several other potential applications that we do not work out in detail.

%Our framework also applies to instrumental variables that are not valid due to the failure of randomization in the binary instrument $Z$. If there is selection in $Z$, then the analysis proceeds very similarly to \Cref{sec:ipw}. 

\section{Additional Applications of Framework} \label{sec:AdditionalApplications}
\subsection{Differences in Differences}\label{sec:DD}

We explain how our framework applies to difference in differences (DD) designs with selected treatments, focusing on a 2-group-2-period structure for simplicity. In this context, the estimand is the average treatment effect on the treated (ATT). With $Y_\tau(t)$ denoting the potential outcome at time $\tau = 0,1$ with treatment status $t=0,1$, $T$ the treatment status, $\tau$ the time period, and $X$ covariates, there is a full distribution $\PTrue$ over $(T, \tau, Y_0(0), Y_0(1), Y_1(0), Y_1(1),X)$, but we only observe the coarsening $\PObs$ over $R=(T,\tau,Y_0, Y_1,X)$, and $\xi$ is empty.

The object of interest is $E_{\PTrue} \left[  Y_1(1) - Y_1(0) \mid T=1 \right]$, and the primary estimand is $E_{\PObs} \left[  Y_1 - Y_0 \mid T=1 \right] - E_{\PObs} \left[  Y_1 - Y_0 \mid T=0 \right]$. Two prominent identification assumptions have been used in the literature: randomization and parallel trends. Our approach is more directly applicable to sensitivity analysis to violations of randomization. %We discuss our method's application to parallel trends later. 

When randomization \citep{athey2022design} holds, both the structural estimand and the primary estimand are also the ATE. Randomization is the assumption that $T \indep Y_\tau(t) \mid X,\tau$. This problem is analogous to the IPW environment. 
If randomization holds, \\
$E_{\PTrue} \left[ Y_1(1) - Y_1(0) \mid T=1 \right] = E_{\PTrue} \left[ Y_1(1) - Y_1(0) \right] $, so our structural estimand is the ATE in the post period:
\begin{align*}
    ATE &= E_{\PTrue} \left[ Y_1(1) - Y_1(0) \right].
\end{align*}

We let $Y = Y_1 - Y_0$ denote the difference in observed outcomes in the two periods. The DD estimand is $E_{\PObs} \left[ \lambda(R) Y \right]$, for $\lambda(R) = \frac{T}{\PObs(T=1 \mid X)} - \frac{1-T}{\PObs(T=0 \mid X)}$. If randomization holds, then $E_{\PObs} \left[ Y_1(0) - Y_0(0) \mid T=0 \right] = E_{\PObs} \left[ Y_1(0) - Y_0(0) \mid T=1 \right]$, which can be interpreted as parallel trends. We allow for violations of randomization in the form of: 
\begin{align*}
\ell_t(X) \leq \frac{e(X, U) / (1-e(X, U))}{e(X) / (1-e(X))} \leq u_t(X)
\end{align*}
where $e(X) = \PObs(T=1 \mid X)$ and $e(X,U) =\PTrue(T=1 \mid X,U)$, where $U = (Y_1(1) - Y_0(0), Y_1(0) - Y_0(0))$.
Then, our method can be implemented with:
\begin{align*}
\lambda(R) &= \frac{T}{e(X)}-\frac{1-T}{1-e(X)},  \\
\underline{w}(R) & = T \left( e(X) +  \frac{1-e(X)}{u_1(X)} \right) + (1-T) \left( 1 - e(X) + e(X) l_0(X) \right) \\
\bar{w}(R) & = T \left( e(X) + \frac{1-e(X)}{l_1(X)} \right) + (1-T) \left( 1 - e(X) + e(X) u_0(X) \right) \\
\frac{d\PTarget(Y \mid R )}{d\PObs(Y \mid R )} &= \frac{T e(X)}{e(X,U)} + \frac{(1-T) (1-e(X))}{1-e(X,U)}
\end{align*}
With this reweighting, the DD estimand recovers $E_{\PTrue}\left[ Y_1(1) - Y_0(0) \right] - E_{\PTrue}\left[ Y_1(0) - Y_0(0) \right] = E_{\PTrue}\left[ Y_1(1) - Y_1(0) \right]$.
Due to their analogy to IPW, these bounds are sharp due to the argument in \citet{dorn2022sharp}.

\begin{remark}
    The analysis here is applicable to selection on parallel trends. The typical parallel trends assumption is some statement like,  $E_{\PObs} \left[ Y_{1}(0) - Y_{0}(0) \mid T=1 \right] = E_{\PObs} \left[ Y_{1}(0) - Y_{0}(0) \mid T=0 \right]$, i.e., the difference in the untreated potential outcome is the same for both treated and untreated groups (e.g., \citet{callaway2021difference} Assumption 4). \citet{roth2023parallel} show that this statement is functional form-independent only under strong randomization-type assumption. Our approach can therefore be interpreted as a functional-form-independent sensitivity analysis for parallel trends violation. 
    %Alternatively, when the representation of the outcome variable is important to identification, our analysis also carries through for the ATT under a bound on the degree of selection on the trends $Y_{1}(0) - Y_{0}(0)$. 
    %Other approaches to sensitivity analysis for difference-in-difference designs include bounding the difference in trends \citep{rambachan2023more} or distributional distances \citep{bertsimas2022distributionally} that can be interpreted as $f$-divergence restriction on selection in our framework \citep{jin2022sensitivity}. Parallel trends can also be microfounded \citep{marx2024parallel} and the bias of the DD estimand can also be characterized when parallel trends fails \citep{ghanem2022selection}. These constructions can inform the sensitivity parameter used, but the exact mapping is beyond the scope of this paper. 
\end{remark}

\subsection{Instrumental Variables}\label{sec:IVExclusion}

In this subsection, we consider the canonical context of instrumental variables to qualify the advantages and limitations of our framework. The change-of-measure analysis is immediate, so we focus on partial identification bounds in terms of underlying potential outcomes. Our framework has a straightforward application to endogenous instrumental variables. Our analysis extends to exclusion failure, but is not sharp when our bound is derived from exclusion failure restrictions. Our framework is not applicable to monotonicity failure. 
%In this subsection, we first explain the more involved case of exclusion failure, then briefly comment on how exogeneity failure reduces to the IPW problem, and on the non-applicability of our framework to monotonicity failure.

%Consider a canonical IV setup. 
We assume there is a distribution $\PTrue$ over $\left(Z,\{T\left(z\right)\}_{z \in \{0,1 \}},T,\left\{ Y\left(t,z\right)\right\}_{t,z \in \{0,1 \}} ,X\right)$, where $Z$ is a binary instrument, $Y(t, z)$ is the potential outcome with binary treatment status $t$ and instrument status $z$, and $T(z)$ is the potential treatment given the instrument $z$ is $T(z)$. However, we only observe the coarsening $\Pobs$ over $\left(Z,T=T\left(Z\right),Y=Y\left(T,Z\right),X\right)$. We will maintain that monotonicity holds ($T(1) \geq T(0)$). Under monotonicity, there are three treatment response groups: always-takers (At, $T(1) = T(0) = 1$), never-takers (Nt, $T(1) = T(0) = 0$), and compliers (Co, $1 = T(1) > T(0) = 0$). The regressors are $R = (T, Z, X)$ and $\xi = \left(\{T\left(z\right)\}_{z \in \{0,1 \}}\right)$.

%While the focus of this section is on the failure of exclusion in IV settings, another assumption used in IV is that of exogeneity. We assumed exogeneity in this section (i.e., that $Y(t,z) \indep Z \mid X$), but our framework can also allow us to conduct sensitivity analysis to failure of exogeneity. For exposition, 
We first consider sensitivity analysis to failure of exogeneity.
Suppose exclusion holds, so $Y(T,z)=Y(T)$ for all $z$, and we target the treatment effect on compliers. Then, $\psi = E_{\PTrue} \left[ \frac{1\{ Co \}}{\PObs(Co)}  (Y(1) - Y(0)) \right]$. With slight abuse of notation, we let $e(X)= \PObs(Z=1\mid X), e(X,U) = \PTrue(Z=1\mid X,U)$, where $U = ( Y(0), Y(1), T(0), T(1))$ denotes the individual unobservable as before. Suppose we allow exogeneity to fail for compliers, but require exogeneity to nonetheless hold for always takers and never takers, i.e., $(Y(1),Y(0)) \indep Z \mid X , T(1), T(0) = T(1)$.\footnote{
This assumption is reasonable in most empirical settings. Suppose we are interested in the impact of participation in 401(k) tax-deferred savings plans (T) on additional savings (Y), and we instrument for T with eligibility for 401(k) (Z) based on employment. Never-takers have no incentive to manipulate their eligibility because they would not take the plan up even if they were offered. However, compliers who really want the plan may manipulate their eligibility by choosing an occupation appropriately.
}
Then, it can be shown that:\footnote{We use the results that:
\begin{align*}
E_{\PTrue} \left[ \frac{Z Y}{e(X,U)}\right] &= E_{\PTrue} \left[ \frac{Z 1\{ Co \} Y}{e(X,U)}\right] + E_{\PTrue} \left[ \frac{(1-Z)T Y}{e(X)}\right] + E_{\PTrue} \left[ \frac{Z (1-T)Y}{e(X)}\right] \\
E_{\PTrue} \left[ \frac{(1-Z) Y}{1-e(X,U)}\right] &= E_{\PTrue} \left[ \frac{(1-Z) 1\{ Co \} Y}{1-e(X,U)}\right]  + E_{\PTrue} \left[ \frac{(1-Z)T Y}{e(X)}\right] + E_{\PTrue} \left[ \frac{Z (1-T)Y}{e(X)}\right] 
\end{align*}}
\begin{align*}
    \psi &= E_{\PTrue} \left[ \frac{Z1 \{ Co\} Y}{e(X,U)} - \frac{(1-Z) 1 \{ Co\} Y}{1- e(X,U)}\right] \\
    &= E_{\PTrue} \left[ \frac{Z Y}{e(X,U)} - \frac{(1-Z) Y}{1- e(X,U)}\right] 
    = E_{\PTarget} \left[ \lambda(R) Y \right]
\end{align*}
for $\lambda(R) =  \frac{Z - \PObs(Z = 1 \mid X)}{\PObs(Z = 1 \mid X) \PObs(Z = 0 \mid X)}$. Then, the problem becomes entirely analogous to IPW, so the form of $\underline{w},\bar{w}$ is identical, and we have sharp bounds.

%We consider bounds on local average treatment effects (LATEs) in instrumental variable models under restrictions on violations of the exogeneity and exclusion restrictions in more detail. Our previous examples correspond to forms of selection on unobservables, so the extension to violations of exogeneity in IV is fairly analogous. Hence, we first discuss how our framework applies to bounds on LATEs in instrumental variable models under restrictions on violations of the exclusion restriction, then briefly mention how the application to exogeneity proceeds.

Turning to sensitivity analysis to failure of exclusion, now suppose that the instrument is randomly assigned ($Y(t, z) \indep Z \mid X$). The exclusion restriction is $Y\left(t,1\right)=Y\left(t,0\right)$. If exclusion fails, then the standard conditional IV estimand, $(E[Y \mid Z=1, X] - E[Y \mid Z=0, X]) / (E[T \mid Z=1, X] - E[T \mid Z=0, X])$, will incorrectly assign any direct effect of the instrument on outcomes to treatment effects. 
We target an instrument-weighted average LATE given by: 
\begin{align*}
\psi & : =E_{\PTrue} \left[ \eta(X) 1\{Co\} \sum_z \omega(z \mid X) \left(Y(1,z)-Y(0,z)\right)  \right] 
\end{align*}
where $\eta(X)$ reflects covariate weighting and $\omega(1 \mid X) = 1 - \omega(0 \mid X)$ is a researcher-estimated instrument weighting function in $[0, 1]$. When exclusion holds, this class includes the average complier treatment effect for $\eta(X) = 1 / \PTrue(Co)$. When exclusion fails, this specification also allows the researcher to target a particular weighted average of treatment effects across instrument statuses. For example, $\omega(Z \mid X) = 1$ targets the treatment effect at $Z = 1$ and $\omega(Z \mid X) = E[Z \mid X]$ targets the treatment effect at the average observed instrument status.

We rewrite the structural estimand $\psi$ in terms of conditional expected outcomes as follows:
\begin{align*}
    %\psi = E_{\PTrue}\left[ \lambda(R) E_{\PTrue}\left[ Y(T, Z) \mid X, Co \right] \right],
    \psi = E_{\PTrue}\left[ \lambda(R) E_{\PTrue}\left[ \sum_z  \omega(z \mid X) Y(T, z) \mid X, T(1), T(0) \right] \right],
\end{align*}
where $\lambda(Z, X, T) = \eta(X) \left( \frac{Z - \PObs(Z = 1 \mid X)}{\PObs(Z = 1 \mid X) \PObs(Z = 0 \mid X)} \right)$. Always-takers and never-takers are included in this statement of $\psi$ for convenience, but their potential outcomes cancel out for $\lambda(R)$ because $E_{\PTrue}[\lambda(R) \mid T(1), T(0), X] = 0$. %$\lambda(Z, X, T) = \frac{\eta(X)}{\PObs(Co \mid X)} \left( \frac{T(Z - \PObs(Z = 1 \mid X)) - (1-T)( (1-Z) - \PObs(Z = 0 \mid X))}{\PObs(Z = 1 \mid X) \PObs(Z = 0 \mid X)} \right)$. 

The target distribution $\PTarget$ is constructed as a marginal distribution. For $v \in \{0, 1\}$, define the distribution $\PTarget$, which re-draws a value $v$ from a Bernoulli($\omega(1 \mid X)$) distribution in order to achieve appropriate outcome weighting, as follows: $$\PTarget(X, T, Z, Y) \equiv \sum_{v, t_1, t_0 \in \{0, 1\}} \PTrue(X, t, Z, T(1) = t_1, T(0) = t_0) \omega(v \mid X) \PTrue( Y(T, v) \mid X, T(1), T(0)).$$ 
%follows:  
%\begin{enumerate}
%    \item Draw $(X, T, Z, T(1), T(0)) \sim \PTrue$
%    \item Draw a reweighted instrument $V \sim Bern( %\omega(1 \mid X) )$
%    \item Draw $Y$ from the distribution of $Y(T, V) %\mid X, T(1), T(0)$ under $\PTrue$
%    \item Return $(X, T, Z, Y)$
%\end{enumerate}
The distribution of $R$ is the same under $\PTarget$ and $\PObs$. Further, the distribution satisfies $E_{\PTarget}[\lambda(R) Y] = E_{\PTrue}[ \lambda(R) E_{\PTrue}[ Y(T, Z) \mid X, Co] ] = \psi$, so that $E_{\PTarget}[\lambda(R) Y]$ is the structural estimand.

Suppose we have a structural model, where, for $z \in \{ 0,1 \}$,
\begin{equation}
\ell(X) \leq \frac{d\PTrue(Y(T(z),1-z) \mid \{ T(z) \}, X)}{d\PTrue(Y(T(z),z) \mid \{ T(z) \}, X)} \leq u(X),
\end{equation}
%\jd{Can we find a more interpretable set of assumptions? Maybe odds ratios that are the same for everyone?}
and suppose further that all associated Radon-Nikodym derivatives are finite and strictly positive. The exclusion restriction $Y(t,1)=Y(t,0)$ implies $\ell(x) = u(x) = 1$. Worst-case bounds that only restrict the support of the potential outcomes correspond to $\ell(x) = 0$, $u(x) = \infty$. With a binary instrument, the sensitivity assumption places a bound on the ratio of the distribution of $Y(t,1)$ to that of $Y(t,0)$ for a given treatment response group.  When $Y$ is binary, \citet{ramsahai2012causal} proposes a sensitivity analysis that places a bound on $\PTrue(Y=1\mid X, T, Z=1, U) - \PTrue(Y=1\mid X, T, Z=0, U)$ for some unobserved $U$, which immediately translates to bounds on always- and never-taker likelihood ratios by taking $U = \{ T(1), T(0) \}$ and implies bounds on complier likelihood ratios. This object can be decomposed as a convolution of several of the probability objects above, so its interpretation is less transparent in a causal framework with potential treatments. Alternatively, $\Lambda$ bounds on the effect of $Z$ on the odds of a binary potential outcome given $X$ and potential treatments imply $\Lambda$ bounds on the odds ratios here, though narrower partially identified regions could be obtained by leveraging estimated regression functions.

As we show in Appendix \Cref{lem:excl_bound_transform}, these objects further imply bounds on $\frac{d \PTarget( Y \mid X, T=t, Z=z)}{d \PObs( Y \mid X, T=t, Z=z)} $ for each value of $T$ and $Z$.

\begin{proposition}\label{prop:ExclusionBounds}
    Suppose we write the implied bounds from the worst-case $\ell$ and $u$ applied to the formulas from Appendix \Cref{lem:excl_bound_transform} as 
    \begin{align*}
    \underline{w} (t,z|X) &\leq \frac{d \PTarget( Y \mid X, T=t, Z=z)}{d \PObs( Y \mid X, T=t, Z=z)} \leq \bar{w}( t,z|X),
    \end{align*}
    where we write $\underline{w}(t,z|X) = \bar{w}(t,z|X) = 1$ for values such that $\PObs(T=t, Z=z \mid X) = 0$. Then our method can be implemented on $E_{\PTrue}[ \eta(X)  1\{Co\} \sum_z \omega(z \mid X) ( Y(1, z) - Y(0, z) )]$
    with the following values: 
    \begin{align*}
        \lambda(R) & = \eta(X) \frac{Z - \PObs(Z = 1 \mid X)}{\PObs(Z = 1 \mid X) \PObs(Z = 0 \mid X)}, \quad
        %\lambda(R) & = \frac{\eta(X)}{\PObs(Co \mid X)} \left( \frac{T(Z - \PObs(Z = 1 \mid X)) - (1-T)( (1-Z) - \PObs(Z = 0 \mid X))}{\PObs(Z = 1 \mid X) \PObs(Z = 0 \mid X)} \right) \\
        \underline{w}(R) = \underline{w}(T, Z \mid X), \quad \bar{w}(R) = \bar{w}(T, Z \mid X). 
    \end{align*}
\end{proposition}

The resulting measure of exclusion failure inherits the invariance property in \Cref{rem:Invariance}, unlike bounds stated in outcome units.
Unlike the previous examples, the characterization in \Cref{prop:ExclusionBounds} is conservative. 
\begin{proposition}\label{prop:ConservativeBounds}
    Suppose the observed distribution follows $X = 1$; $Z \mid X \sim Bern(0.5)$; $T \mid Z, X \sim Bern( Z / 2 )$; and $Y \mid X, Z, T \sim Unif(-1, 1)$. Suppose we are interested in the average complier treatment effect at $z=1$, i.e. $\eta(x) = 2$ and $\omega(z \mid x) = z$. Suppose a structural model implies lower bounds of $\ell(x) = 1$ and $u(x) = \infty$. Then the structural model implies the sharp bounds are the singleton $\{0\}$, but the bounds from \Cref{prop:ExclusionBounds} are $[-1, 1]$. 
\end{proposition}
Intuitively, the structural model implicitly includes cross-restrictions between the $Co$, $At$, and $Nt$ Radon-Nikodym derivatives. In this case, the $\frac{d \PTarget(Y \mid T=0, Z=0)}{d \PObs( Y \mid T=0, Z=0)}$ includes a product of Radon-Nikodym derivatives for compliers and never-takers, which must satisfy further constraints that our approach omits for the purpose of ease of use.\footnote{Non-sharpness is attributed to how the sensitivity parameter maps to $\frac{d \PTarget(Y\mid X,T,Z)}{d \PObs(Y\mid X,T,Z)}$. If we use $\frac{d \PTarget(Y\mid X,T,Z)}{d \PObs(Y\mid X,T,Z)}$ as the sensitivity parameter, for instance, we will automatically get sharp bounds, though $\frac{d \PTarget(Y\mid X,T,Z)}{d \PObs(Y\mid X,T,Z)}$ is not interpretable.}

\begin{remark}
Our procedure can be applied somewhat trivially to linear projections using Ordinary Least Squares (OLS). Consider a hypothetical linear model $Y=X\beta+u$. Suppose we are interested in a linear combination of coefficients $\delta' \beta$, where $X$ includes an intercept but $\delta$ puts no weight on the intercept term, so that without loss of generality we can assume $E_{\PTrue}[u] = 0$. However, suppose there may be endogeneity in the sense that $E[X u] \neq 0$. If we targeted a distribution $\PTarget$ that first sampled $X \sim \PObs$, then drew $u \mid X$ from the distribution of $u$ under $\PTrue$, and then returned $Y = X \beta + u$, then the target distribution would  obtain the correct coefficients. The sensitivity assumption is then on $\frac{d\PTarget}{d\PObs}$. For instance, we may have $\underline{w} \leq \frac{d\PTarget}{d\PObs}\leq \bar{w}$. Using the notation of our general framework, $\lambda(R) = \delta' E[X' X]^{-1} X$. An open question is whether there is a microfounded interpretation of such an assumption.
%While the bounds are sharp, the sensitivity assumption need not be interpretable in general. 
\end{remark}

\begin{remark}
Due to our interpretation of the sensitivity parameter as the Radon-Nikodym derivative of the conditional distribution of $Y$, one case that is not easily covered by our framework is when the sensitivity assumption does not concern $Y$ directly. In the context of instrumental variables, some sensitivity analysis to violations of the monotonicity assumption that $T(1) \geq T(0)$ makes assumptions on the mass of defiers with $T(1) < T(0)$ that do not condition on potential outcomes (e.g., \citet{yap2025sensitivity}), which do not map to $Y$ directly.
\end{remark}

\section{Supporting Results for Estimation and Inference}\label{sec:EstimationAppendix}

This appendix collects the results used in \Cref{sec:EstimationAndInference}.
\Cref{rem:SecondOrder} is the second-order expansion behind the exact regime of \Cref{prop:DirectBootstrap}, and \Cref{lemma:IMCalibration} states the \citet{imbens2004confidence} calibration used in both \Cref{prop:DirectBootstrap} and \Cref{prop:RDBootstrap}.
\Cref{rem:Nesting} and \Cref{lemma:RDMonotone} then establish that the RD model family is nested in the manipulation share and in the sensitivity parameter, which lets the tilted share of \Cref{subsec:BootstrapForRD} be replaced by a deterministic dominating sequence, and \Cref{assum:RDInference} states the conditions for inference.

%\begin{remark}[Why stabilization is needed]\label{rem:StabilizationNeeded}
%The moments $E_{\PObs}[|Y|] < \infty$ and $\sup_{r}|\lambda(r)| < \infty$ in \Cref{assum:Estimation} keep a law of large numbers available inside every truncated model, so the divergence in \Cref{prop:consistency}(ii) is driven by $\bar{w}$ rather than by $Y$ failing to have a first moment, and \Cref{assum:QuantileMoments} plays the analogous role for $Q^{+}$ on the event that $\bar{w}(R)$ is finite. The stabilization requirement on $\hat{Q}$ is not innocuous. Suppose $\hat{Q}$ were free to track the sample: an interpolating quantile estimator with $\hat{Q}(R_{i}) = \hat{\lambda}(R_{i})Y_{i}$ sets every hinge $\{\hat{\lambda}(R_{i})Y_{i} - \hat{Q}(R_{i})\}_{+}$ to zero, so that $\hat{\Psi} = \hat{E}[\hat{\lambda}(R)Y] \xrightarrow{p} E_{\PObs}[\lambda(R)Y]$, a finite limit, and $\hat{\bar{w}}$ never enters the estimator at all. No moment or tail condition on $(Y,R)$ rules this out; only $\|\hat{Q}-\bar{Q}\|_\infty = o_P(1)$ for a fixed integrable $\bar{Q}$ does, which is what makes the truncation argument legitimate. The related scenario in which $\|\hat{Q}\|_\infty$ diverges slowly enough that no sample point exceeds it is excluded by the same requirement, though there the conclusion holds a fortiori whenever $\underline{w}$ is bounded away from one. What \Cref{prop:consistency}(ii) does not require is that $\bar{Q}$ be correct.
%\end{remark}

\begin{remark}[Nested families give nested bounds]\label{rem:Nesting}
Let $(\underline{w}, \bar{w})$ and $(\underline{w}', \bar{w}')$ be weight bounds with $[\underline{w}(R), \bar{w}(R)] \subseteq [\underline{w}'(R), \bar{w}'(R)]$ almost surely.
Then $\mathcal{M}(\underline{w}, \bar{w}) \subseteq \mathcal{M}(\underline{w}', \bar{w}')$ by \Cref{assum:SensitivityModel}, so taking the supremum and infimum over the larger family gives
\begin{align*}
    \psi^{+}( \underline{w}, \bar{w} ) & \leq \psi^{+}( \underline{w}', \bar{w}' ), & \psi^{-}( \underline{w}, \bar{w} ) & \geq \psi^{-}( \underline{w}', \bar{w}' ).
\end{align*}
The same ordering holds for the plug-in bound at an arbitrary putative quantile function $\bar{Q}$, and not only at $Q^{\pm}$.
Writing $\phi$ for the pseudo-outcome of \Cref{eq:UBFormula} evaluated at $\bar{Q}$, its partial derivative in $\bar{w}(R)$ is $\{ \lambda(R) Y - \bar{Q}(R) \}_{+} \geq 0$ and its partial derivative in $\underline{w}(R)$ is $-\{ \bar{Q}(R) - \lambda(R) Y \}_{+} \leq 0$, so widening the weight interval weakly raises the upper bound and weakly lowers the lower bound for every fixed $\bar{Q}$.
A sensitivity analysis conducted at weights that are too permissive is therefore conservative.
\end{remark}

Write $\Gamma^{U}( \ProbManip', \Lambda \mid \bar{Q} )$ and $\Gamma^{L}( \ProbManip', \Lambda \mid \bar{Q} )$ for the bounds of \Cref{thm:IDSetUB} at the weights of \Cref{prop:RDEstimands} with $\ProbManip$ replaced by $\ProbManip' \in [0,1)$, at sensitivity parameters $\Lambda = (\Lambda_0, \Lambda_1) \in [1,\infty]^2$, and with a putative quantile function $\bar{Q}$ in place of $Q^{\pm}$; abbreviate the sharp bounds themselves by $\Gamma^{j} \equiv \Gamma^{j}(\ProbManip, \Lambda \mid Q^{\pm})$.
The first ingredient of the argument is that the RD model family is nested in both arguments, so that \Cref{rem:Nesting} applies and overstating either can only widen the identified set.

\begin{lemma}\label{lemma:RDMonotone}
Let $\mathcal{W}_{\ProbManip', \Lambda}(R)$ denote the weight interval implied by \Cref{prop:RDEstimands} with $\ProbManip$ replaced by $\ProbManip' \in [0,1)$.
For each of the CATE, CLATE and CATT, $\mathcal{W}_{\ProbManip', \Lambda}(R) \subseteq \mathcal{W}_{\ProbManip', \Lambda'}(R)$ almost surely whenever $\Lambda \leq \Lambda'$ componentwise.
For the CLATE and the CATT, $\mathcal{W}_{\ProbManip', \Lambda}(R) \subseteq \mathcal{W}_{\ProbManip'', \Lambda}(R)$ almost surely whenever $\ProbManip' \leq \ProbManip''$; for the CATE the same holds for $\ProbManip' \leq \ProbManip''$ in $[0, \bar{\ProbManip}]$ whenever $\bar{\ProbManip} \leq 1/2$ and $4 \bar{\ProbManip} + \{ \Lambda_1 - 2 \}_+ \bar{\ProbManip}^2 \leq 2$.
Consequently, by \Cref{rem:Nesting}, the family $\mathcal{M}( \underline{w}_{\ProbManip', \Lambda}, \bar{w}_{\ProbManip', \Lambda})$ is nested and both $\Gamma^{U}(\ProbManip', \Lambda \mid \bar{Q})$ and $-\Gamma^{L}(\ProbManip', \Lambda \mid \bar{Q})$ are nondecreasing in each argument on these ranges, for every fixed $\bar{Q}$ and for $\bar{Q} = Q^{\pm}$.
\end{lemma}

\begin{proof}[Proof of \Cref{lemma:RDMonotone}]
\emph{Monotonicity in $\Lambda$.} By inspection of \Cref{prop:RDEstimands}, for the CLATE, $\underline{w}_{\ProbManip', \Lambda}(R) = (1-\ProbManip'+\ProbManip'\Lambda_1)^{-1}$ and $\bar{w}_{\ProbManip', \Lambda}(R) = (1-\ProbManip'+\ProbManip'\Lambda_1^{-1})^{-1}$ on $\{T=1\}$ and both equal one on $\{T=0\}$; the first is nonincreasing and the second nondecreasing in $\Lambda_1$.
For the CATT, $\underline{w}_{\ProbManip', \Lambda}(R) = (1-\ProbManip')+\ProbManip'\Lambda_0^{-1}$ and $\bar{w}_{\ProbManip', \Lambda}(R) = (1-\ProbManip')+\ProbManip'\Lambda_0$ on $\{T=0\}$, with the same monotonicity in $\Lambda_0$, and both equal one on $\{T=1\}$.
For the CATE, $\Lambda_0$ enters the $\{T=0\}$ branch only through $-\frac{\ProbManip'}{2-\ProbManip'}(1-\Lambda_0^{-1})$ in $\underline{w}$ and $+\frac{\ProbManip'}{2-\ProbManip'}(\Lambda_0-1)$ in $\bar{w}$, and $\Lambda_1$ enters the $\{T=1\}$ branch only through $\frac{1-\ProbManip'}{2-\ProbManip'}(1-\ProbManip'+\ProbManip'\Lambda_1)^{-1}$ in $\underline{w}$ and $\frac{1-\ProbManip'}{2-\ProbManip'}(1-\ProbManip'+\ProbManip'\Lambda_1^{-1})^{-1}$ in $\bar{w}$.

\emph{Monotonicity in $\ProbManip'$ for the CLATE and the CATT.} For the CLATE, $\underline{w}_{\ProbManip', \Lambda}(R) = (1 + \ProbManip'(\Lambda_1-1))^{-1}$ and $\bar{w}_{\ProbManip', \Lambda}(R) = (1 - \ProbManip'(1-\Lambda_1^{-1}))^{-1}$ on $\{T=1\}$, which are respectively nonincreasing and nondecreasing in $\ProbManip'$ on $[0,1)$ because $\Lambda_1 \geq 1$.
For the CATT, $\underline{w}_{\ProbManip', \Lambda}(R) = 1 - \ProbManip'(1-\Lambda_0^{-1})$ and $\bar{w}_{\ProbManip', \Lambda}(R) = 1 + \ProbManip'(\Lambda_0 - 1)$ on $\{T=0\}$, which are affine in $\ProbManip'$ with the required signs.
Both hold for every $\ProbManip' \in [0,1)$, which is the sense in which these two estimands are unrestricted.

\emph{Monotonicity in $\ProbManip'$ for the CATE.} On $\{T = 0\}$ the CATE weights are $1 - \frac{\ProbManip'}{2-\ProbManip'}(1-\Lambda_0^{-1})$ and $1 + \frac{\ProbManip'}{2-\ProbManip'}(\Lambda_0-1)$, and $\ProbManip' \mapsto \frac{\ProbManip'}{2-\ProbManip'}$ is increasing on $[0,1)$, which gives the claim on that branch for every $\ProbManip'$ and without any side condition.
%The restriction is therefore entirely a $\{T=1\}$ phenomenon.
On $\{T=1\}$, write $u = \Lambda_1 - 1 \geq 0$ and $v = 1 - \Lambda_1^{-1} \in [0,1]$ and collect terms over a common denominator:
\begin{align*}
    \underline{w}_{\ProbManip', \Lambda}(R) & = \frac{1}{2-\ProbManip'}\left( 1 + \frac{1-\ProbManip'}{1+u\ProbManip'} \right)
    = \frac{2 + (u-1)\ProbManip'}{(2-\ProbManip')(1+u\ProbManip')}, \\
    \bar{w}_{\ProbManip', \Lambda}(R) & = \frac{1}{2-\ProbManip'}\left( 1 + \frac{1-\ProbManip'}{1-v\ProbManip'} \right)
    = \frac{2 - (1+v)\ProbManip'}{(2-\ProbManip')(1-v\ProbManip')} .
\end{align*}
Differentiating in $\ProbManip'$ and discarding the squared denominators, the numerators of the two derivatives are
\begin{align*}
    u \left[ (u-1)\ProbManip'^2 + 4\ProbManip' - 2 \right] \qquad \text{and} \qquad v \left[ 2 - 4\ProbManip' + (1+v)\ProbManip'^2 \right] .
\end{align*}
Hence $\underline{w}_{\ProbManip', \Lambda}$ is nonincreasing in $\ProbManip'$ exactly when $(\Lambda_1 - 2)\ProbManip'^2 + 4\ProbManip' \leq 2$, and $\bar{w}_{\ProbManip', \Lambda}$ is nondecreasing in $\ProbManip'$ exactly when $(2 - \Lambda_1^{-1})\ProbManip'^2 - 4\ProbManip' + 2 \geq 0$.
The first left-hand side is increasing in $\ProbManip'$ on $[0,1)$, since its derivative $2(\Lambda_1-2)\ProbManip' + 4 \geq 4 - 4\ProbManip' > 0$, so it suffices to check it at $\bar{\ProbManip}$, where $4\bar{\ProbManip} + \{\Lambda_1 - 2\}_+ \bar{\ProbManip}^2 \leq 2$ gives it.
The second left-hand side is decreasing in $\ProbManip'$ on $[0, 1/2]$, since its derivative $2(2-\Lambda_1^{-1})\ProbManip' - 4 \leq -2$, and at $\ProbManip' = \bar{\ProbManip} \leq 1/2$ it is at least $2 - 4\bar{\ProbManip} \geq 0$.
Both conditions therefore hold on all of $[0, \bar{\ProbManip}]$.

\emph{Nesting and the bounds.} On the ranges just established the weight intervals are nested, so \Cref{rem:Nesting} applies verbatim: $\mathcal{M}( \underline{w}_{\ProbManip', \Lambda}, \bar{w}_{\ProbManip', \Lambda})$ is nested, and $\Gamma^{U}(\ProbManip', \Lambda \mid \bar{Q})$ and $-\Gamma^{L}(\ProbManip', \Lambda \mid \bar{Q})$ are nondecreasing in each argument for every fixed $\bar{Q}$.
What remains is the claim at $\bar{Q} = Q^{\pm}$, which does not follow from the fixed-$\bar{Q}$ statement because the minimizing quantile itself moves with $(\ProbManip', \Lambda)$.
By \Cref{thm:IDSetUB} and \Cref{cor:Robustness}, $\Gamma^{U}(\ProbManip', \Lambda \mid Q^{+}) = \inf_{\bar{Q}} \Gamma^{U}(\ProbManip', \Lambda \mid \bar{Q})$, a pointwise infimum of a family of functions each of which is nondecreasing in $(\ProbManip', \Lambda)$, and a pointwise infimum of nondecreasing functions is nondecreasing.
The argument for $\Gamma^{L}$ is symmetric.
\end{proof}

At $\Lambda_1 = \infty$ the CATE treated-side lower weight is $\underline{w}_{\ProbManip', \Lambda}(R) = (2-\ProbManip')^{-1}$ on $\{T = 1\}$, which increases in $\ProbManip'$.
The condition for the CATE is nonetheless slack at the magnitudes we encounter: our application's estimated manipulation share of $\hat{\ProbManip} = 0.11$ permits every $\Lambda_1 \leq 130.9$, where $\bar{\ProbManip}$ must be taken to dominate the tilted share $\hat{\ProbManip}^{*}$.
The CLATE and CATT are unrestricted.

\begin{remark}[Quantile error is second order at the truth]\label{rem:SecondOrder}
Perturbing in a bounded direction $h$ with $E_{\PObs}[\bar{w}(R)|h(R)|] < \infty$,
\[
    \frac{d}{dt}\,\Psi(Q^{+} + t h)\bigg|_{t=0}
    = E_{\PObs}\!\left[ \left\{ (1-\underline{w}(R)) - (\bar{w}(R)-\underline{w}(R))\,\PObs\!\left(\lambda(R)Y > Q^{+}(R) \mid R\right) \right\} h(R) \right] = 0,
\]
because $\PObs(\lambda(R)Y > Q^{+}(R) \mid R) = 1-\tau(R) = \frac{1-\underline{w}(R)}{\bar{w}(R)-\underline{w}(R)}$ whenever the conditional law of $\lambda(R)Y$ places no atom at $Q^{+}(R)$. Consequently, $\Psi(\cdot)$ is Neyman orthogonal in the quantile at $\bar{Q} = Q^{+}$. Since $q \mapsto (1-\underline{w}(R))q + (\bar{w}(R)-\underline{w}(R))E_{\PObs}[\{\lambda(R)Y-q\}_{+} \mid R]$ has second derivative $(\bar{w}(R)-\underline{w}(R))f_{\lambda(R)Y \mid R}(q)$, a conditional density bounded by $\bar{f}$ on the relevant neighborhood yields
\[
    0 \ \leq \ \Psi(\bar{Q}) - \Psi \ \leq \ \frac{\bar{f}}{2}\, E_{\PObs}\!\left[ (\bar{w}(R)-\underline{w}(R)) \left(\bar{Q}(R)-Q^{+}(R)\right)^{2} \right].
\]
The left inequality is \Cref{cor:Robustness} and needs only convexity; the right one makes quantile estimation error second order, as in \citet{dorn2022sharp}, \citet{dvds}, and \citet{tan2024model}. 
%Boundedness of $\bar{w}$ is doing work in the upper bound: when $\bar{w}$ is unbounded the quadratic remainder need not be finite, which is one reason \Cref{prop:consistency}(ii) is a divergence statement rather than a rate statement.
\end{remark}

We now state the conditions under which \Cref{prop:RDBootstrap} holds.
The notation $\Gamma^{j}( \ProbManip', \Lambda \mid \bar{Q} )$ is that of \Cref{lemma:RDMonotone} above; $\hat{\Gamma}^{L*}$, $\hat{\Gamma}^{U*}$ and $\hat{\Delta}$ are defined in \Cref{subsec:BootstrapForRD}.
Write $\Rightarrow$ for convergence in distribution.

\begin{manualassumption}{RD Inference}\label{assum:RDInference}
Fix $\alpha \in (0, 1/2)$ and an assumed $\Lambda$, and let $\bar{Q}$ denote the limit quantile function of \Cref{assum:Estimation}.
\begin{enumerate}
    \item[(i)] \emph{(Sampling and smoothing.)} The data are i.i.d.\ draws from $\PObs$, and the objects estimated at the cutoff are the limits that \Cref{assum:RD} defines, computed by local polynomial methods at bandwidths tending to zero slowly enough that the corresponding effective sample sizes diverge. Writing $\sigma^{\max}_n = \max\{ \sigma^{L}_n, \sigma^{U}_n \}$ for the scales in (iii), the smoothing bias of the outcome-side estimators is $o(\sigma^{\max}_n)$, and the smoothing bias $b_n$ of $\tilde{\eta}$ satisfies $\min\{ b_n, 0 \} = o(\varsigma / r_n)$.
    \item[(ii)] \emph{(Regular manipulation share.)} There are $r_n \to \infty$ and $\varsigma \in (0, \infty)$ with $r_n ( \tilde{\eta} - \ProbManip ) \Rightarrow \mathcal{N}(0, \varsigma^2)$ and $r_n \hat{\sigma}_{\tilde{\eta}} \xrightarrow{p} \varsigma$; and $\kappa_n \to \infty$ with $\kappa_n / r_n \to 0$.
    \item[(iii)] \emph{(Reference interval.)} There is a deterministic sequence $\ProbManip^*_n \geq \ProbManip$ with $\ProbManip^*_n \to \ProbManip$, lying in a range on which \Cref{lemma:RDMonotone} applies, such that, writing $\Gamma^{j \ddagger}_n \equiv \Gamma^{j}( \ProbManip^*_n, \Lambda \mid \bar{Q})$ and $\Delta^{\ddagger}_n \equiv \Gamma^{U \ddagger}_n - \Gamma^{L \ddagger}_n$, there are $\sigma^{L}_n, \sigma^{U}_n > 0$ with
    \begin{align*}
        \left( \left. \left( \hat{\Gamma}^{L*} - \Gamma^{L \ddagger}_n \right) \right/ \sigma^{L}_n , \ \left. \left( \hat{\Gamma}^{U*} - \Gamma^{U \ddagger}_n \right) \right/ \sigma^{U}_n  \right) & \Rightarrow (Z^{L}, Z^{U}),
    \end{align*}
    where $Z^{L}$ and $Z^{U}$ are standard normal; $\hat{\sigma}^{j} / \sigma^{j}_n \xrightarrow{p} c^{j}$ for constants $c^{j} \geq 1$; and $( \hat{\Delta} - \Delta^{\ddagger}_n ) / \max\{ \hat{\sigma}^{L}, \hat{\sigma}^{U} \} \xrightarrow{p} 0$.
    \item[(iv)] \emph{(Model.)} Assumptions \ref{assum:ObservableDistProperties}, \ref{assum:QuantileMoments}, \ref{assum:RD} and \ref{assum:Estimation} hold, $\Gamma^{L \ddagger}_n$ and $\Gamma^{U \ddagger}_n$ are finite, and the assumed $\Lambda$ weakly exceeds the $\Lambda$ of \Cref{eq:RDManipulationBounds} componentwise.
\end{enumerate}
\end{manualassumption}
Part (i) states the sampling frame the identification results already use. 
%The estimand is the cutoff limit rather than an average over a fixed window, so no local randomization assumption is introduced and the continuity conditions of \Cref{assum:RD} continue to do the identification work; what is added is that the bias of estimating those limits must be negligible on the scale of the sampling error. 
On the outcome side, smoothing bias enters the bound additively and shifts its location, so it must be removed in the usual way, by undersmoothing or by bias correction. On the density side, only downward bias is problematic, because by \Cref{lemma:RDMonotone} an understated manipulation share narrows the identified set. The assumption is also substantive only where the floor is slack: when $\kappa_n \hat{\sigma}_{\tilde{\eta}}$ exceeds $\tilde{\eta}$, the tilt already dominates any bias of the order of the standard error of $\tilde{\eta}$, so the condition has bite only in the regime $\ProbManip \gg \kappa_n \varsigma / r_n$.
Part (iii) ensures that a central limit theorem applies. 
%When the floor is asymptotically slack, $\liminf_n r_n \ProbManip / \kappa_n > \varsigma$, we have $\hat{\eta}^* = \tilde{\eta}$ with probability tending to one, and (iii) holds with $\ProbManip^*_n = \ProbManip$ by the delta method applied to the closed forms of \Cref{prop:RDEstimands}, which are continuously differentiable in $\ProbManip^\prime$ away from $\ProbManip^\prime = 1$; the sampling variability of $\tilde{\eta}$ is then part of $\sigma^{j}_n$ and is reproduced by $\hat{\sigma}^{j}$, so $c^{j} = 1$. When the floor binds, including at $\ProbManip = 0$, we have $\hat{\eta}^{*} = \kappa_n \hat{\sigma}_{\tilde{\eta}}$ with probability tending to one and (iii) holds with $\ProbManip^*_n = \kappa_n \varsigma / r_n$, which exceeds $\ProbManip$ for $n$ large; the bootstrap dispersion of $\hat{\psi}^{j}(\hat{\eta}_b)$ then still carries the variability of $\tilde{\eta}$, which the estimator itself no longer has, so $c^{j} \geq 1$ rather than $c^{j} = 1$. Allowing $c^{j} \geq 1$ is what lets one assumption cover both regimes.

\section{Proofs}\label{sec:Proofs}

\subsection{Proofs of Theorem}
The following lemma is used in the proof of \Cref{thm:IDSetUB}.

\begin{lemma}\label{lemma:RewriteA}
    Recall $\tau(R) = \frac{\bar{w}(R) - 1}{\bar{w}(R) - \underline{w}(R)}$ and $a(\bar{w}, \underline{w}, s) = (\bar{w} - \underline{w}) 1\{ s > 0 \} - (1-\underline{w})$. On the event $\{ \underline{w}(R) < 1, \ \bar{w}(R) < \infty \}$, equivalently $\{ 0 < \tau(R) \cdot 1 < 1 \}$ up to the degenerate case below,
    \begin{align*}
        a(\bar{w}, \underline{w}, s) & = (1-\underline{w}) \left( \frac{1\{ s > 0 \}}{1-\tau} - 1 \right).
    \end{align*}
\end{lemma}

\begin{proof}[Proof of \Cref{thm:IDSetUB}]
We prove the statement for $\psi^{+}$; the argument for $\psi^{-}$ is symmetric, replacing $Q^{+}$ by $Q^{-}$ and reversing the inequality in the indicator. Convexity of the identified set is established in Step 4.

\emph{Step 1: reduction to a pointwise problem.} By \Cref{assum:ObservableDistProperties} and \Cref{lemma:ReweightingHolds}, a distribution $\Q$ lies in $\mathcal{M}(\underline{w}, \bar{w})$ if and only if $W \equiv \frac{d\Q(R,Y)}{d\PObs(R,Y)}$ satisfies $W \in [\underline{w}(R), \bar{w}(R)]$ almost surely and $E_{\PObs}[W \mid R] = 1$ almost surely, the latter because $\Q$ and $\PObs$ share an $R$-marginal. Define the pointwise value
\begin{align}
    \psi^+(R) & \equiv \sup_{W} \ E_{\PObs}\left[ W \lambda(R) Y \mid R \right] \quad \text{s.t.} \quad W \in [\underline{w}(R), \bar{w}(R)] \ \text{ and } \ E_{\PObs}[W \mid R] = 1. \label{eq:PointwiseUB}
\end{align}
For any feasible $W$, the law of iterated expectations gives $E_{\PObs}[W \lambda(R) Y] = E_{\PObs}[ E_{\PObs}[W \lambda(R) Y \mid R] ] \leq E_{\PObs}[\psi^+(R)]$, so $\psi^{+} \leq E_{\PObs}[\psi^+(R)]$. For the reverse inequality, first suppose $\bar{w}$ is bounded. \Cref{cor:WorstCaseWeights} exhibits a feasible $W^{*}_{sup}$, measurable in $(R,Y)$, attaining the conditional supremum in \Cref{eq:PointwiseUB} almost surely, so the two sides agree. If $\bar{w}$ is unbounded, apply this to $\bar{w} \wedge K$ for $K < \infty$. Since $\mathcal{M}(\underline{w}, \bar{w} \wedge K) \subseteq \mathcal{M}(\underline{w}, \bar{w})$, we obtain $\psi^{+} \geq \psi^{+}(\underline{w}, \bar{w} \wedge K)$ for every $K$, and Step 3 below shows $\psi^{+}(\underline{w}, \bar{w} \wedge K) \uparrow E_{\PObs}[\psi^{+}(R)]$ as $K \to \infty$. Hence $\psi^{+} = E_{\PObs}[\psi^+(R)]$ in either case.

\emph{Step 2: the pointwise value is a conditional value at risk.} Reparametrize $\Cref{eq:PointwiseUB}$ by $V = (W - \underline{w}(R))/(1-\underline{w}(R))$ on the event $\underline{w}(R) < 1$. The feasible set becomes $0 \leq V \leq \frac{\bar{w}(R)-\underline{w}(R)}{1-\underline{w}(R)} = \frac{1}{1-\tau(R)}$ with $E_{\PObs}[V \mid R] = 1$, which is exactly the dual constraint set of the level-$\tau(R)$ conditional value at risk. By the dual representation, for any level $\alpha \in [0,1)$ and any $\lambda(R)Y$ whose conditional mean exists,
\[
    CVaR^{+}_{\alpha}(\lambda(R) Y \mid R) = \sup_{V} \ E_{\PObs}[ V \lambda(R) Y \mid R] \quad \text{s.t.} \quad 0 \leq V \leq \tfrac{1}{1-\alpha} \ \text{ and } \ E_{\PObs}[V \mid R] = 1,
\]
so that, as in \citet{dvds},
\begin{align}
    \psi^+(R) & = \underline{w}(R)\, E_{\PObs}[\lambda(R) Y \mid R] + (1-\underline{w}(R))\, CVaR^{+}_{\tau(R)}(\lambda(R) Y \mid R) \quad \text{a.s.} \label{eq:PointwiseCVaR}
\end{align}
When $\tau(R) = 1$ the constraint $V \leq 1/(1-\tau(R))$ is vacuous, the supremum places all mass at the essential supremum of $\lambda(R) Y$ given $R$, and $CVaR^{+}_{1}(\lambda(R)Y \mid R) = Q^{+}(R)$ by the convention for $t(R) \in \{0,1\}$ in the notation. When $\underline{w}(R) = \bar{w}(R) = 1$, feasibility forces $W = 1$ almost surely, $\psi^{+}(R) = E_{\PObs}[\lambda(R)Y \mid R]$, and \Cref{eq:PointwiseCVaR} holds under any assignment of $\tau(R)$.

\emph{Step 3: the conditional value at risk equals the pseudo-outcome.} Write $\phi^{+}(R) \equiv E_{\PObs}[ \phi^+(\underline{w}, \bar{w}, R, Y \mid Q^+) \mid R]$. On the event $\{\tau(R) < 1\}$, \Cref{lemma:RewriteA} gives
\begin{align*}
    \phi^{+}(R) & = E_{\PObs}\left[ \lambda(R) Y + (1-\underline{w}(R)) (\lambda(R) Y - Q^+(R) ) \left( \frac{1\{ \lambda(R) Y > Q^+(R) \}}{1-\tau(R)} - 1 \right) \Bigm| R \right] \\
    & = \underline{w}(R)\, E_{\PObs}[\lambda(R) Y \mid R] + (1-\underline{w}(R)) \left( Q^+(R) + \frac{E_{\PObs}[\{ \lambda(R) Y - Q^+(R) \}_+ \mid R]}{1-\tau(R)}\right) \\
    & = \underline{w}(R)\, E_{\PObs}[\lambda(R) Y \mid R] + (1-\underline{w}(R))\, CVaR^{+}_{\tau(R)}(\lambda(R) Y \mid R),
\end{align*}
using the representation $CVaR^{+}_{\alpha}(\lambda(R) Y \mid R) = Q^{+}(R) + \frac{1}{1-\alpha} E_{\PObs}[\{\lambda(R) Y - Q^{+}(R)\}_{+} \mid R]$ at $\alpha = \tau(R)$. On the complementary event $\{\tau(R) = 1\}$, which by the definition of $\tau$ arises when $\underline{w}(R) = 1$ or $\bar{w}(R) = \infty$, $Q^{+}(R)$ is the essential supremum of $\lambda(R)Y$ given $R$, so $1\{\lambda(R) Y > Q^{+}(R)\} = 0$ almost surely and $a( \bar{w}(R), \underline{w}(R), \lambda(R)Y - Q^{+}(R)) = -(1-\underline{w}(R))$ almost surely. Hence
\[
    \phi^{+}(R) = \underline{w}(R)\, E_{\PObs}[\lambda(R) Y \mid R] + (1-\underline{w}(R)) Q^{+}(R) = \underline{w}(R)\, E_{\PObs}[\lambda(R) Y \mid R] + (1-\underline{w}(R))\, CVaR^{+}_{1}(\lambda(R) Y \mid R),
\]
which agrees with \Cref{eq:PointwiseCVaR}. Combining the two events, $\psi^{+}(R) = \phi^{+}(R)$ almost surely.

It remains to record the monotone approximation used in Step 1. Writing $\tau_K(R)$ for the threshold probability built from $\bar{w}(R) \wedge K$, the map $w \mapsto \frac{w-1}{w-\underline{w}}$ has derivative $\frac{1-\underline{w}}{(w-\underline{w})^{2}} \geq 0$, so $\tau_K \uparrow \tau$; and $\alpha \mapsto CVaR^{+}_{\alpha}$ is nondecreasing with $CVaR^{+}_{\alpha} \uparrow CVaR^{+}_{1} = Q^{+}$ as $\alpha \uparrow 1$. Monotone convergence applied to \Cref{eq:PointwiseCVaR} then gives $\psi^{+}(\underline{w}, \bar{w} \wedge K) \uparrow E_{\PObs}[\psi^{+}(R)]$, which is finite or $+\infty$ under \Cref{assum:QuantileMoments}. Taking expectations of $\psi^{+}(R) = \phi^{+}(R)$ and applying Step 1 yields \Cref{eq:UBFormula}.

\emph{Step 4: convexity of the identified set.} The set $\mathcal{M}(\underline{w}, \bar{w})$ is convex: if $\Q_0, \Q_1 \in \mathcal{M}(\underline{w}, \bar{w})$ and $t \in [0,1]$, then $\Q_t \equiv t \Q_1 + (1-t) \Q_0$ is a distribution on the same space with $\frac{d\Q_t}{d\PObs} = t \frac{d\Q_1}{d\PObs} + (1-t)\frac{d\Q_0}{d\PObs} \in [\underline{w}(R), \bar{w}(R)]$ almost surely. Let $\psi_0 < \psi < \psi_1$ with $\psi_0, \psi_1 \in \mathcal{I}(\underline{w}, \bar{w})$; it suffices to show $\psi \in \mathcal{I}(\underline{w}, \bar{w})$. If $\psi_0$ and $\psi_1$ are finite, with witnesses $\Q_0$ and $\Q_1$, then $t \mapsto E_{\Q_t}[\lambda(R) Y] = t \psi_1 + (1-t)\psi_0$ is continuous and attains $\psi$ at $t = (\psi - \psi_0)/(\psi_1 - \psi_0)$. If $\psi_1 = +\infty$, then by the monotone approximation of Step 3 there is $K < \infty$ with $\psi^{+}(\underline{w}, \bar{w} \wedge K) > \psi$; the corresponding worst-case distribution of \Cref{cor:WorstCaseWeights} lies in $\mathcal{M}(\underline{w}, \bar{w} \wedge K) \subseteq \mathcal{M}(\underline{w}, \bar{w})$ and has a finite value exceeding $\psi$, so the finite case applies with $\psi_1$ replaced by that value. The case $\psi_0 = -\infty$ is symmetric. Hence $\mathcal{I}(\underline{w}, \bar{w})$ is an interval in $\bar{\R}$.
\end{proof}

The following lemma states the \citet{imbens2004confidence} calibration in the generality that \Cref{prop:DirectBootstrap} and \Cref{prop:RDBootstrap} require, where the width of the reference interval is allowed to drift to zero.
\begin{lemma}[Calibration at a reference interval]\label{lemma:IMCalibration}
Let $L_n \leq U_n$ be deterministic sequences with $\Delta_n \equiv U_n - L_n$, let $\hat{L}$ and $\hat{U}$ be estimators with $\hat{\Delta} \equiv \hat{U} - \hat{L}$, and let $\sigma^{L}_n, \sigma^{U}_n > 0$, writing $\sigma^{\max}_n \equiv \max\{ \sigma^{L}_n, \sigma^{U}_n \}$ and $\hat{\sigma}_{\max} \equiv \max\{ \hat{\sigma}^{L}, \hat{\sigma}^{U} \}$. Suppose
\begin{enumerate}
    \item[(i)] $\left( ( \hat{L} - L_n ) / \sigma^{L}_n, \ ( \hat{U} - U_n ) / \sigma^{U}_n \right) \Rightarrow ( Z^{L}, Z^{U} )$ with standard normal marginals;
    \item[(ii)] $\hat{\sigma}^{j} / \sigma^{j}_n \xrightarrow{p} c^{j}$ for constants $c^{j} \geq 1$;
    \item[(iii)] either $( \hat{\Delta} - \Delta_n ) / \hat{\sigma}_{\max} \xrightarrow{p} 0$, or $\Delta_n / \sigma^{\max}_n \to \infty$ and $\hat{\Delta} / \Delta_n \xrightarrow{p} 1$.
\end{enumerate}
Let $\hat{\rho} \equiv \hat{\Delta} / \hat{\sigma}_{\max}$ and $\mathcal{C} \equiv [ \hat{L} - r_\alpha(\hat{\rho}) \hat{\sigma}^{L}, \ \hat{U} + r_\alpha(\hat{\rho}) \hat{\sigma}^{U} ]$. Then
\begin{align*}
    \liminf_{n \to \infty} \ \inf_{ \psi \in [L_n, U_n] } \PP( \psi \in \mathcal{C} ) & \geq 1 - \alpha .
\end{align*}
Replacing $r_\alpha(\hat{\rho})$ by $\Phi^{-1}(1-\alpha/2)$ gives $\liminf_n \PP( [L_n, U_n] \subseteq \mathcal{C} ) \geq 1-\alpha$, while $r_\alpha(\hat{\rho})$ itself gives only $1 - 2\alpha$ for the interval.
\end{lemma}

\begin{proof}[Proof of \Cref{prop:DirectBootstrap}]
Write $\Psi^{+}(\hat{Q}^{+})$ for the population upper bound at the realized $\hat{Q}^{+}$ and let $\sigma^{2}_{+} \equiv \Var_{\PObs}( \phi(\underline{w}, \bar{w}, \lambda, \bar{Q}^{+}; Y) )$, which lies in $(0, \infty)$ by \Cref{assum:DirectReweighting}(iii) together with $\bar{w} \leq \bar{W}$, $\sup_r |\lambda(r)| < \infty$ and $E_{\PObs}[Y^2] < \infty$.

\emph{Step 1: expansion.} The map $\bar{Q} \mapsto \phi(\underline{w}, \bar{w}, \lambda, \bar{Q}; Y)$ is Lipschitz with constant $\max\{ 1 - \underline{w}(R), \bar{w}(R) - 1 \} \leq \bar{W}$, so \Cref{assum:DirectReweighting}(i) gives $\| \phi(\cdot, \hat{Q}^{+}) - \phi(\cdot, \bar{Q}^{+}) \|_{L_2(\PObs)} \leq \bar{W} \| \hat{Q}^{+} - \bar{Q}^{+} \|_{L_2(\PObs)} = \op(1)$. Since $\hat{Q}^{+}$ lies in $\mathcal{Q}$ with probability tending to one, asymptotic equicontinuity of the empirical process indexed by the Donsker class of \Cref{assum:DirectReweighting}(ii) gives $\sqrt{n} ( \hat{E} - E_{\PObs} ) [ \phi(\cdot, \hat{Q}^{+}) - \phi(\cdot, \bar{Q}^{+}) ] = \op(1)$, whence
\begin{align*}
    \sqrt{n} \left( \hat{\Psi}^{+} - \Psi^{+}(\hat{Q}^{+}) \right) & = \sqrt{n} ( \hat{E} - E_{\PObs} ) [ \phi(\cdot, \bar{Q}^{+}) ] + \op(1) \Rightarrow \mathcal{N}(0, \sigma^{2}_{+}) .
\end{align*}

\emph{Step 2: bootstrap.} Because $\hat{Q}$ is held fixed across draws, $t^{+(b)} = \sqrt{n} ( \hat{E}^{(b)} - \hat{E} )[ \phi(\cdot, \hat{Q}^{+}) ]$, a bootstrapped sample mean of a fixed function. Donsker classes are bootstrap-Donsker \citep{gine1990bootstrapping}, so the argument of Step 1 applies to the bootstrap empirical process and $t^{+(b)} = \sqrt{n} ( \hat{E}^{(b)} - \hat{E} )[ \phi(\cdot, \bar{Q}^{+}) ] + \op(1)$, whose conditional law converges weakly in probability to $\mathcal{N}(0, \sigma^2_{+})$. Hence $q^{+}_{\alpha} \xrightarrow{p} \sigma_{+} \Phi^{-1}(\alpha)$ and, by Step 1, $\PP( \Psi^{+}(\hat{Q}^{+}) \leq \hat{\Psi}^{+} - q^{+}_{\alpha}/\sqrt{n} ) \to 1-\alpha$.

\emph{Step 3: sign.} By \Cref{cor:Robustness}, $\Psi^{+}(\hat{Q}^{+}) \geq \psi^{+}(\underline{w}, \bar{w})$ for every realization of $\hat{Q}^{+}$, so the event of Step 2 implies the event in the display, which gives the stated inequality. It is strict in the limit whenever $\Psi^{+}(\bar{Q}^{+}) > \psi^{+}$, in which case coverage tends to one.

\emph{Step 4: exactness.} Under \Cref{assum:DirectReweighting}(iv), part (i) forces $\bar{Q}^{+} = Q^{+}$ in $L_2(\PObs)$, and \Cref{rem:SecondOrder} gives $0 \leq \Psi^{+}(\hat{Q}^{+}) - \psi^{+} = \Op( \| \hat{Q}^{+} - Q^{+} \|^{2}_{L_2(\PObs)} ) = \op(n^{-1/2})$. The centering in Step 2 may therefore be replaced by $\psi^{+}$ at cost $\op(n^{-1/2})$, and coverage is asymptotically exact.

\emph{Step 5: the identified set.} The same argument applied to $\hat{\Psi}^{-}$, with the lower-bound form of \Cref{cor:Robustness}, gives $\liminf_n \PP( \psi^{-} \geq \hat{\Psi}^{-} - q^{-}_{1-\alpha}/\sqrt{n} ) \geq 1-\alpha$. A union bound over the two one-sided events gives $\liminf_n \PP( [\psi^{-}, \psi^{+}] \subseteq CI(\alpha) ) \geq 1-2\alpha$.

\emph{Step 6: pointwise coverage.} Assume \Cref{assum:DirectReweighting}(iv) and $\psi^{-} < \psi^{+}$. Steps 1--4 applied to both endpoints give $\sqrt{n}( \hat{\Psi}^{\pm} - \psi^{\pm} ) \Rightarrow \mathcal{N}(0, \sigma^{2}_{\pm})$ with $\sigma_{\pm} > 0$, and Step 2 gives $\sqrt{n}\,\hat{\sigma}^{\pm} \xrightarrow{p} \sigma_{\pm}$. Apply \Cref{lemma:IMCalibration} with $L_n = \psi^{-}$, $U_n = \psi^{+}$, $\sigma^{L}_n = \sigma_{-}/\sqrt{n}$, $\sigma^{U}_n = \sigma_{+}/\sqrt{n}$ and $c^{j} = 1$. Since $\Delta_n = \psi^{+} - \psi^{-}$ is a positive constant while $\sigma^{\max}_n \to 0$, we have $\Delta_n / \sigma^{\max}_n \to \infty$, and $\hat{\Delta} \xrightarrow{p} \Delta_n$ gives $\hat{\Delta}/\Delta_n \xrightarrow{p} 1$, so the second alternative in hypothesis (iii) of that lemma holds.
\end{proof}

\begin{proof}[Proof of \Cref{prop:RDBootstrap}]
\emph{Step 1: ordering.} We show
\begin{align}
    \mathcal{I}_{\ProbManip, \Lambda} \ \subseteq \ \mathcal{I}_{\ProbManip^{*}_n, \Lambda} \ \subseteq \ \left[ \Gamma^{L \ddagger}_n, \Gamma^{U \ddagger}_n \right] \equiv \left[ \Gamma^{L}( \ProbManip^{*}_n, \Lambda \mid \bar{Q} ), \ \Gamma^{U}( \ProbManip^{*}_n, \Lambda \mid \bar{Q} ) \right] . \label{eq:RDInclusionChain}
\end{align}
By \Cref{assum:RDInference}(iii) the deterministic sequence satisfies $\ProbManip^{*}_n \geq \ProbManip$ and lies in a range on which \Cref{lemma:RDMonotone} applies, so $\mathcal{W}_{\ProbManip, \Lambda}(R) \subseteq \mathcal{W}_{\ProbManip^{*}_n, \Lambda}(R)$ almost surely and \Cref{rem:Nesting} gives the first inclusion. By \Cref{cor:Robustness}, $\Gamma^{U}( \ProbManip^{*}_n, \Lambda \mid Q^{+} ) \leq \Gamma^{U \ddagger}_n$ and symmetrically $\Gamma^{L \ddagger}_n \leq \Gamma^{L}( \ProbManip^{*}_n, \Lambda \mid Q^{-} )$, which is the second.

%That $\psi_0 \in \mathcal{I}_{\ProbManip, \Lambda}$ is where the assumed and true sensitivity parameters are distinguished. 
By \Cref{assum:RDInference}(iv) the assumed $\Lambda$ weakly exceeds the true $\Lambda^{\text{true}}$ of \Cref{eq:RDManipulationBounds} componentwise, so \Cref{lemma:RDMonotone} and \Cref{rem:Nesting} give $\mathcal{I}_{\ProbManip, \Lambda^{\text{true}}} \subseteq \mathcal{I}_{\ProbManip, \Lambda}$, and $\psi_0 \in \mathcal{I}_{\ProbManip, \Lambda^{\text{true}}}$ by \Cref{thm:IDSetUB} and \Cref{prop:RDEstimands}. The remainder of the argument never refers to $\Lambda^{\text{true}}$.

\emph{Step 2: calibration at the reference interval.} \Cref{assum:RDInference}(iii) supplies hypotheses (i), (ii) and the first alternative of (iii) of \Cref{lemma:IMCalibration}, with $L_n = \Gamma^{L\ddagger}_n$, $U_n = \Gamma^{U\ddagger}_n$, $\hat{L} = \hat{\Gamma}^{L*}$ and $\hat{U} = \hat{\Gamma}^{U*}$. That lemma gives
\begin{align*}
    \liminf_{n \to \infty} \ \inf_{ \psi \in [\Gamma^{L\ddagger}_n, \Gamma^{U\ddagger}_n] } \PP( \psi \in \mathcal{C}^{SRD}_{1-\alpha} ) & \geq 1-\alpha ,
\end{align*}
and an infimum over a superset is weakly smaller, so \Cref{eq:RDInclusionChain} gives part (i) for $\mathcal{I}_{\ProbManip, \Lambda}$ and hence for $\psi_0$.

\emph{Step 3: the set.} By \Cref{eq:RDInclusionChain}, $\{ \mathcal{I}_{\ProbManip, \Lambda} \not\subseteq \mathcal{C}^{SRD}_{1-\alpha} \} \subseteq \{ [\Gamma^{L\ddagger}_n, \Gamma^{U\ddagger}_n] \not\subseteq \mathcal{C}^{SRD}_{1-\alpha} \}$, and Step 4 of the proof of \Cref{lemma:IMCalibration} bounds the probability of the latter by $2\alpha$ under $r_\alpha(\hat{\rho})$ and by $\alpha$ under $\Phi^{-1}(1-\alpha/2)$.

%Consistency of $\hat{Q}$ is used nowhere: $\bar{Q}$ enters only through $\Gamma^{j \ddagger}_n$ and through \Cref{assum:RDInference}(iii), which is stated at $\Gamma^{j \ddagger}_n$ rather than at the sharp bounds, and what is asked of $\hat{Q}$ is the stabilization of \Cref{assum:Estimation}(i) together with the sign from \Cref{cor:Robustness}. Nor is any rate for $\tilde{\ProbManip}$ used beyond \Cref{assum:RDInference}(ii): the tilt replaces an estimate of the share by a deterministic sequence dominating it, and \Cref{lemma:RDMonotone} converts domination into inclusion.
\end{proof}

\subsection{Proofs of Propositions}

\begin{proof}[Proof of \Cref{prop:IPW_Y1}]
    This is immediate by the discussion in \Cref{sec:ipw}. 
\end{proof}

\begin{proof}[Proof of \Cref{prop:RDEstimands}]
Recall by \Cref{lemma:ImpliedLikelihoodRatios} that selection bounds on $q_t(Y(t)) = \Q( M = 1 \mid X = c, Y(t) )$ from Equation (\ref{eq:RDManipulationBounds}) correspond to likelihood ratios for $\Q \in \mathcal{M}( \underline{w}, \bar{w} )$ as:
    \begin{align*}
        \frac{d \Q(Y(1) \mid X = c, M = 0)}{d \PObs(Y \mid X = c^+)} & = \frac{1}{1-\ProbManip} \left( 1 + 2 \frac{q_1(Y(1))}{1 - q_1(Y(1))} \right)^{-1} \\
        & \in \left[ \frac{1}{1-\ProbManip} \left( 1 + 2 \frac{\ProbManip}{2(1-\ProbManip)} \Lambda_1 \right)^{-1}, \frac{1}{1-\ProbManip} \left( 1 + 2 \frac{\ProbManip}{2(1-\ProbManip)} \Lambda_1^{-1} \right)^{-1} \right] \\
        \frac{d \Q(Y(0) \mid X = c, M = 1)}{d \PObs(Y \mid X = c^-)} & =  \frac{2(1-\ProbManip)}{\ProbManip} \frac{q_0(Y(0))}{1-q_0(Y(0))}
    \end{align*}

    As a result, due to \Cref{eq:M1M0ratio},
    \begin{align*}
        \frac{d \Q(Y(1) \mid X = c, M = 0)}{d \PObs(Y \mid X = c^+)} & \in \left[ \frac{1}{1 - \ProbManip + \ProbManip \Lambda_1}, \frac{1}{1 - \ProbManip + \ProbManip \Lambda_1^{-1}} \right] \\
        \frac{d \Q(Y(0) \mid X = c, M = 1)}{d \PObs(Y \mid X = c^-)} & \in \left[ \Lambda_0^{-1}, \Lambda_0 \right]. 
    \end{align*}

Next, we verify that the target distribution's estimand $E_{\PTarget}[ \lambda(R) Y ]$ achieves the relevant structural estimands. Suppose we call the target estimator $\bar{\psi}$. We proceed as follows:
\begin{align*}
    \bar{\psi}_{CLATE} & = E_{\PTarget}\left[ \lambda(R) Y \mid X = c \right] \\
     & = \PObs( X = c^+ \mid X = c) E_{\PTarget}\left[ \lambda(R) Y \mid X = c^+ \right] \\
     &- \PObs( X = c^- \mid X = c) E_{\PTarget}\left[ \lambda(R) Y \mid X = c^- \right]   \\
     & =  E_{\PTarget}\left[ Y \mid X = c^+ \right] - E_{\PTarget}\left[ Y \mid X = c^- \right] \\
     & = E_{\PTrue}\left[ Y(1) \mid X = c, M = 0 \right] - E_{\PObs}[Y(0) \mid X = c, M = 0] = \psi_{CLATE}\\
    \bar{\psi}_{CATT} & = E_{\PObs}[Y \mid X = c^+] - E_{\PTrue}[Y(0) \mid X = c^+] = \psi_{CATT} \\
     \bar{\psi}_{CATE} & = E_{\PTrue}[ Y(1) \mid X = c] - E_{\PTrue}[ Y(0) \mid X = c] = \psi_{CATE}. 
\end{align*}

Finally, we verify that the Radon–Nikodym bounds have the appropriate forms by using \Cref{eq:RD_estimands}. The formula for the CLATE is immediate. The formula for the CATT follows by observing that:
\begin{align*}
    &\frac{d \PTrue( Y(0) \mid X = c^+ )}{d \PObs( Y \mid X = c^-)} = \PObs( M = 0 \mid X = c^+) * \frac{d \PTrue( Y(0) \mid M = 0, X = c^+)}{d \PObs( Y \mid X = c^-)} \\
    & \quad + \PObs( M = 1 \mid X = c^+) * \frac{d \PTrue( Y(0) \mid M = 1, X = c^+)}{d \PObs( Y \mid X = c^-)} \\
    & = \PObs( M = 0 \mid X = c^+ ) * 1 + \PObs( M = 1 \mid X = c^+ ) * \frac{d \PTrue( Y(0) \mid X = c, M = 1)}{d \PObs( Y \mid X = c^-)} \\
    & = (1-\ProbManip) + \ProbManip \frac{d \PTrue( Y(0) \mid X = c, M = 1)}{d \PObs(Y \mid X = c^-)} \\ 
   & \in \left[ (1-\ProbManip) + \ProbManip \Lambda_0^{-1}, (1-\ProbManip) + \ProbManip \Lambda_0 \right],
\end{align*}
verifying the form. 

The formula for the CATE similarly follows by the argument from \Cref{lemma:EstimandExpressions}. In particular:
\begin{align*}
    &\frac{d \PTrue( Y(1) \mid X = c)}{d \PObs( Y \mid X = c^+)} = \PObs(X = c^+ \mid X=c) * \frac{d \PTrue( Y(1) \mid X = c^+)}{d \PObs( Y \mid X = c^+)} \\
    & \quad + \PObs(X = c^- \mid X=c) * \frac{d \PTrue( Y(1) \mid X = c^-)}{d \PObs( Y \mid X = c^+)}  \\
    & = \PObs(X = c^+ \mid X=c) * 1 + \PObs(X = c^- \mid X = c) * \frac{d \PTrue( Y(1) \mid X = c, M = 0)}{d \PObs( Y \mid X = c^+)}.
\end{align*}

Similarly:
\begin{align*}
    &\frac{d \PTrue( Y(0) \mid X = c)}{d \PObs( Y \mid X = c^-)} = \PObs( X = c^+ \mid X = c) * \frac{d \PTrue( Y(0) \mid X = c^+)}{d \PObs( Y \mid X = c^-)} + \PObs( X = c^- \mid X = c) * 1 \\
    & = \ProbManip \PObs(X = c^+ \mid X=c) \frac{d \PTrue( Y(0) \mid X = c, M = 1)}{d \PObs( Y \mid X = c^-)} + (1 - \ProbManip \PObs(X = c^+ \mid X=c)). 
\end{align*}
Recall from the proof of \Cref{lemma:ImpliedLikelihoodRatios} that $\PObs(X = c^+ \mid X=c) = 1 / (2-\ProbManip)$, so that $\PObs(X = c^- \mid X=c) = (1-\ProbManip) / (2-\ProbManip)$.
Therefore:
\begin{align*}
    1\{T=1\} \underline{w}(R) & = 1\{T=1\} \left( \frac{1}{2-\ProbManip} + \frac{1-\ProbManip}{2-\ProbManip} \left( 1-\ProbManip + \ProbManip \Lambda_1 \right)^{-1} \right) \\
    1\{T=1\} \bar{w}(R) & = 1\{T=1\} \left( \frac{1}{2-\ProbManip} + \frac{1-\ProbManip}{2-\ProbManip} \left( 1-\ProbManip + \ProbManip \Lambda_1^{-1} \right)^{-1} \right) \\
    1\{T=0\} \underline{w}(R) & = 1\{T=0\} \left( 1 - \frac{\ProbManip}{2-\ProbManip} + \frac{\ProbManip}{2-\ProbManip} \Lambda_0^{-1} \right) \\
    1\{T=0\} \bar{w}(R) & = 1\{T=0\} \left( 1 - \frac{\ProbManip}{2-\ProbManip} + \frac{\ProbManip}{2-\ProbManip} \Lambda_0 \right),
\end{align*}
completing the proof.
\end{proof}

\begin{proof}[Proof of \Cref{prop:RDSharpness}]

This proof is extensive, so we split up the argument into separate subsubsections. 
To prove that the bounds are sharp, we adopt the following strategy. First, we show that our bounds are feasible, by showing that the likelihood constraint implies the allowable distribution and that any $\psi_1,\psi_0$ that solves our problem are jointly feasible. Second and more extensively, we explicitly describe a worst-case distribution, then show that it is in the model family and achieves the extremal estimands, and then show that a convex combination of distributions enables achieving any pair of interior estimands as well.

\subsubsection{Likelihood Constraint Implies Allowable Distribution}

    First, we show that if $\Q \in \mathcal{M}_{\Lambda}'(\Lambda)$, then $\Q \in \mathcal{M}_{\Lambda}(\Lambda, \Lambda)$. It is clear that it only remains to show that for $t=1,0$, we have $\frac{\Q(M = 1 \mid Y(t), X = c)}{\Q( M = 0 \mid Y(t), X = c)} / \frac{\PTrue( M = 1 \mid X = c)}{\PTrue(M = 0 \mid X = c)} \in [\Lambda^{-1}, \Lambda]$ almost surely. This holds by iterated expectations over $\int \Q(M = 1 \mid Y(1), Y(0), X = c) d \Q (Y(1-t) \mid X = c, Y(1-t))$ and the restriction on the domain of $\Q(M = 1 \mid Y(1), Y(0), X = c)$ under Equation (\ref{eq:JointPOSelection}). Therefore $\Q \in \mathcal{M}_{\Lambda}(\Lambda, \Lambda)$. This completes the first direction of the proof. 

\subsubsection{Any Two Feasible Potential Outcomes are Jointly Feasible}

    We now proceed to construct a $\Q'$ corresponding to $\psi_1$ and $\psi_0$.

    For simplicity, we proceed assuming that $Y$ is continuously distributed. (If the conditional distribution of $Y$ contains mass points at the referenced quantiles, add a uniform random variable to provide a strict ordering on observations of $Y \mid X$ and then apply the construction below to the quantiles of the tie-broken $Y$.)
    
    Recall that all of our structural estimands of interest can be written as observable linear transformations of the partially identified average potential outcomes $E[ Y(t) \mid X = c, M=1-t]$ for $t=1,0$. 

    We define a distribution $\Q_{+, +}$ which we show is in the model family $\mathcal{M}_{\Lambda}'(\Lambda) \subset \mathcal{M}_{\Lambda}(\Lambda, \Lambda) \subset \mathcal{M}_{\Lambda}'(\infty)$ and achieves the upper bounds on both average potential outcomes within $\mathcal{M}_{\Lambda}(\Lambda, \Lambda)$. We then argue that any $\Q_{\pm, \pm}$ combination of upper and lower bounds is feasible. Finally, we argue by mixture that any pair of structural estimands of interest between the two lower and upper bounds under $\mathcal{M}_{\Lambda}(\Lambda, \Lambda)$ are feasible for some distribution $\Q \in \mathcal{M}_{\Lambda}'(\Lambda)$.

\subsubsection{Construction of Worst-Case Distribution}

    We will define a distribution $\Q_{+,+}$ with an unobserved confounder ``$V$." $V = 1$ will correspond to a high manipulation probability, small treated potential outcomes (to maximize $E[Y(1) \mid X = c^+, M = 0]$), and large untreated outcomes (to maximize $E[Y(0) \mid X = c^+, M = 1]$). At the threshold, a fraction $\tau_1$ of treated and $1-\tau_0$ of untreated observations will have $V=1$, where:
    \begin{align*}
        \tau_1 & = \frac{1 + \eta (\Lambda - 1)}{\Lambda + 1} \\
        \tau_0 & = \frac{\Lambda}{\Lambda + 1}. 
    \end{align*}
    In order to facilitate this drawing, we define the observable $V$ function: 
    \begin{align*}
        V(x, t, y) & = \left\{\begin{array}{rl}
        1\{ y \leq Q_{\tau_1}(Y \mid X=c^+) \}     &  \mbox{if } t = 1 \\
        1\{ y > Q_{\tau_0}(Y \mid X=c^-) \}     & \mbox{if } t = 0. 
        \end{array} \right. 
    \end{align*}
    Finally, in order to match worst-case selection probabilities, $M \mid X = c, V=v$ will be drawn iid from a $Bern( T q_{++}(v) )$ distribution, where $q_{++}(v)$ is defined as:
    \begin{align*}
        q_{++}(v) & = \left\{
        \begin{array}{rl}
        \frac{\eta}{(\Lambda + 1) (1-\tau_1)}     & \mbox{if } v = 0  \\
        \frac{\eta \Lambda}{ (\Lambda + 1) \tau_1 }   & \mbox{if } v = 1.
        \end{array} \right. 
    \end{align*}
    Finally, we will draw unobserved potential outcomes from the distribution with the same value of $M$ and $V$.

    Formally, the distribution $\Q_{+,+}$ over $(M, X(1), X(0), Y(1), Y(0), T, T(0))$ is defined as follows:
    \begin{enumerate}
        \item Draw $(X, T, Y) \sim \PObs$
        \item Define the random variable $V = V(X, T, Y)$
        \item Draw $M \sim Bern( T q_{++}(V) )$ 
        \item Set the lower-case variable realizations for later use $x = X$, and $y = Y$, $t=T$, $v=V$, and $m=M$
        \item Set $X(M) = X$ and $Y(t) = y$
        \item Draw $X(1-M)$ from the distribution of $X \mid M = 1-m$ under $\Q_{+,+}$ defined so far 
        \item If $m = 0$, draw $Y(1-t)$ from the distribution of $Y \mid X = 2c - x, M = 0, V = v$ under the construction of $\Q_{+,+}$ so far 
        \item If $m=1$, draw $Y(0)$ from the distribution of $Y \mid X = 2 c - x, V = V(t, y)$ under the construction of $\Q_{+,+}$ so far 
        \item Set $T(0)=1 \{ X(0)>c \}$
        \item Return data $(M, X(1), X(0), Y(1), Y(0), T, T(0))$
    \end{enumerate}

    \subsubsection{Showing Constructed Distribution is in Model Family}

    We wish to show that $\Q_{+,+} \in \mathcal{M}_{\Lambda}^\prime(\Lambda)$, i.e.:
    \begin{enumerate}[label=(\alph*)]
        \item The distribution of $(X = X(M), Y = Y(T), T)$ under $\Q_{+,+}$ marginalizes to the distribution of $(X, T, Y)$ under $\PObs$\label{req:MarginalizesToObserved}
        \item $\Q_{+,+}(T = 0, M = 1) = 0$\label{req:ManipulatorsAreTreated}
        \item $\Q_{+,+}(M = 1 \mid X = c^+) = \ProbManip$ \label{req:ManipulationProbRight}
        \item $\Q_{+,+}( Y(t) \leq y \mid X = x, M = 0)$ is continuous at $c$ and $\Q_{+,+}( Y(1) \leq y \mid x = c, M = 1)$ is right-continuous at $x = c$. \label{req:ContinuityOfCDFs}
        \item $\Q_{+,+}( Y(1), Y(0), T(0), M \mid X = c )$ is a conditional distribution that is well-defined as the appropriate continuous limit. \label{req:ContinuousDistribution} 
        \item  $\Q_{+,+}$ satisfies Equation (\ref{eq:JointPOSelection}): \label{req:JointPOSelection}
        \begin{align*}
            \left. \frac{\Q_{+,+}( M = 1 \mid Y(1), Y(0), X = c)}{\Q_{+,+}( M = 0 \mid Y(1), Y(0), X = c)} \right/ \frac{\PTrue( M = 1 \mid X = c)}{\PTrue( M = 0 \mid X = c)} \in [\Lambda^{-1}, \Lambda]. 
        \end{align*}
    \end{enumerate}

    Requirements \ref{req:MarginalizesToObserved} and \ref{req:ManipulatorsAreTreated} are immediate.

    Requirement \ref{req:ManipulationProbRight} follows by inspection:
    \begin{align*}
        \Q_{+,+}(M = 1 \mid X = c^+) & = E_{\PObs}[ q_{++}(V) \mid X = c^+] \\
        & = \sum_v \PObs(V = v \mid X = c^+) q_{++}(v) \\
        & = \frac{\ProbManip}{\Lambda + 1} + \frac{\ProbManip \Lambda}{\Lambda + 1} = \ProbManip = \PTrue(M = 1 \mid X = c^+).
    \end{align*}

    The requirement \ref{req:ContinuityOfCDFs} holds as follows. $\PTrue( Y \leq y \mid X = x)$ is left- and right-continuous at $x=c$ as follows. For left-continuity, for $x < c$, $\PTrue( Y \leq y \mid X = x) = \PTrue( Y \leq y \mid M=0, X = x)$ is left-continuous at $x=c$ by assumption. For right-continuity, for $x > c$, 
    \begin{align*}
        \PTrue( Y \leq y \mid X = x) &= \PTrue( M = 1 \mid X=x) \PTrue( Y \leq y \mid M = 1, X = x) \\
        &\quad + \PTrue( M = 0 \mid X=x) \PTrue( Y \leq y \mid M = 0, X = x), 
    \end{align*}
    which by assumption is right-continuous at $x=c$. As a result, $\Q_{+,+}(Y(t) \leq y \mid X = x, M = 0)$ is continuous in $x$ at $x = c$. Note also that for $X > c$ and $y \leq Q_{\tau_1}(Y \mid X = c^+)$, we have:
    \small 
    \begin{align*}
        \Q_{+,+}&( Y(t) \leq y \mid X, M = 1) \\
        & = \frac{\PObs( Y(t) \leq y \mid X) q_{++}(0)}{\PObs( Y(t) \leq Q_{\tau_1}( Y \mid X = c^+) \mid X) q_{++}(0) + \PObs( Y(t) > Q_{\tau_1}( Y \mid X = c^+) \mid X)},
    \end{align*} 
    \normalsize 
    and similarly, for $Y > Q_{\tau_1}(Y \mid X = c^+)$:
    \small 
    \begin{align*}
        \Q_{+,+}&( Y(t) \leq y \mid X, M = 1) \\
        & = \frac{\PObs( Y(t) \leq Q_{\tau_1}(Y \mid X = c^+) \mid X) q_{++}(0) + \PObs( Y(t) \in [Q_{\tau_1}(Y \mid X = c^+), y] \mid X) q_{++}(0)}{\PObs( Y(t) \leq Q_{\tau_1}( Y \mid X = c^+) \mid X) q_{++}(0) + \PObs( Y(t) > Q_{\tau_1}( Y \mid X = c^+) \mid X)},
    \end{align*} 
    \normalsize 
    both of which are continuous in $X$.
    
    \ref{req:ContinuousDistribution} holds as follows, where we write $T(x) = 1\{ x > c \}$. For $m=1$ and some $A \subseteq \R^2$:
    \footnotesize 
    \begin{align*}
        & \lim_{\varepsilon \to 0^+} \Q_{+,+}( (Y(1), Y(0)) \in A, T(0) = t, M = 1, V = v \mid |X - \varepsilon| \leq c ) \\
        = & \int 1\{X > c, (Y(1), Y(0)) \in A\} q_{++}( v ) \frac{1}{2} d \Q_{+,+}( V, Y(0) \mid X, Y, M = 1) d \PObs( Y, X \mid |X - c| \leq \varepsilon ) \\
        = & \int q_{++}( v ) \frac{1\{ X > c, (Y(1), Y(0)) \in A, V(X, 1, Y) = v \}}{2} d \Q_{+,+}( Y \mid |X - c| \leq c, T=0, V=v) d \PObs( Y, X \mid |X - c| \leq \varepsilon ),
    \end{align*}
    \normalsize 
    which has a well-defined limit by one-sided continuity of the CDF of $Y \mid X$ under our assumptions. The remaining claim holds by taking the sum over $v$. Similarly, for $m = 0$:
    \begin{align*}
        & \lim_{\varepsilon \to 0^+} \Q_{+,+}( (Y(1), Y(0)) \in A, T(0) = t, M = 1, V = v \mid |X - \varepsilon| \leq c ) \\
        = & 1\{ X \leq c \} \Q_{+,+}( (Y(1), Y(0)) \in A, T(0) = t, M = 0, V = v \mid |X - \varepsilon| \leq c ) \\
        + & 1\{ X > c \} \Q_{+,+}( (Y(1), Y(0)) \in A, T(0) = t, M = 0, V = v \mid |X - \varepsilon| \leq c ), 
    \end{align*}
    which is continuous by analogous arguments.

    Requirement \ref{req:JointPOSelection} involves a longer argument, so we show it in a separate section.

    \subsubsection{$\Q_{+,+}$ Satisfies the Odds Ratio Bound}

    We wish to show:
    \begin{align*}
        \left. \frac{\Q_{+,+}( M = 1 \mid Y(1), Y(0), X = c)}{\Q_{+,+}( M = 0 \mid Y(1), Y(0), X = c)} \right/ \frac{\PTrue( M = 1 \mid X = c)}{\PTrue( M = 0 \mid X = c)} \in [\Lambda^{-1}, \Lambda]. 
    \end{align*}
    Recall that $\frac{\PTrue(M = 1 \mid X = c)}{\PTrue(M = 0 \mid X = c)} = \frac{\eta}{2 (1-\eta)}$.

    Notice that $Y(1), Y(0)$ are iid given $M, X, V$ under $\Q_{+,+}$. As a result, we can equivalently show that for $v = 0, 1$:
    \begin{align*}
        \frac{\Q_{+,+}( M = 1 \mid V=v, X = c)}{\Q_{+,+}( M = 0 \mid V=v, X = c)} & = \frac{\sum_t \Q_{+,+}( T=t, V=v, M = 1  \mid X = c)}{\sum_t \Q_{+,+}( T=t, V=v, M = 0 \mid X = c)} \\
        & \in \left[ \Lambda^{-1} \frac{2 (1-\eta)}{\eta}, \Lambda  \frac{2 (1-\eta)}{\eta} \right]. 
    \end{align*}

    We do so case-by-case for values of $v$. In the $v = 0$ case:
    \begin{align*}
        \sum_t \Q_{+,+}( T=t, V=0, M = 1  \mid X = c) & = \frac{1}{2-\eta} (1-\tau_1) q_{++}(0) = \frac{\eta}{(\Lambda + 1)(2-\eta)} \\
        \sum_t \Q_{+,+}( T=t, V=0, M = 0 \mid X = c) & = \frac{1}{2-\eta} (1-\tau_1) (1 - q_{++}(0)) + \frac{1-\eta}{2-\eta} \tau_0 \\
        & = \frac{1}{2-\eta} \left( 1-\tau_1 - \frac{\eta}{\Lambda +1} \right) + \frac{\Lambda (1-\eta)}{(2-\eta)(\Lambda + 1)}  \\
        & = \frac{1}{(2-\eta)(\Lambda + 1)} \left( \Lambda (1-\eta) \right) + \frac{\Lambda (1-\eta)}{(2-\eta)(\Lambda + 1)} \\
        & = \frac{2 \Lambda (1-\eta)}{(\Lambda+1)(2-\eta)} \\
       \frac{\Q_{+,+}( M = 1 \mid V=0, X = c)}{\Q_{+,+}( M = 0 \mid V=0, X = c)} & = \Lambda^{-1} \frac{\eta}{2 (1-\eta)}.
    \end{align*}
    Similarly, in the $v=1$ case:
    \begin{align*}
        \sum_t \Q_{+,+}( T=t, V=1, M = 1  \mid X = c) & = \frac{1}{2-\eta} \tau_1 q_{++}(1) = \frac{\Lambda \eta}{(\Lambda + 1)(2-\eta)} \\
        \sum_t \Q_{+,+}( T=t, V=1, M = 0 \mid X = c) & = \frac{1}{2-\eta} \tau_1 (1 - q_{++}(1)) + \frac{1-\eta}{2-\eta} (1-\tau_0) \\
        & = \frac{1}{(\Lambda+1)(2-\eta)} \left( 1 - \eta \right) + \frac{1-\eta}{(\Lambda+1)(2-\eta)} \\
       \frac{\Q_{+,+}( M = 1 \mid V=1, X = c)}{\Q_{+,+}( M = 0 \mid V=1, X = c)} & = \Lambda \frac{\eta}{2 (1-\eta)}.
    \end{align*}
    Therefore, $\Q_{+,+} \in \mathcal{M}_{\Lambda}'(\Lambda)$.

    \subsubsection{Showing Constructed Distribution is Worst-Case for Both Estimands}

    We wish to show that $\E_{\Q_{+,+}}[ Y(1) \mid X = c, M = 0 ]$ and $\E_{\Q_{+,+}}[ Y(0) \mid X = c, M = 1 ]$ are maximal within $\mathcal{M}_{\Lambda}(\Lambda, \Lambda)$. By the proof of \Cref{prop:RDEstimands}, the maximal conditional expectations are achieved if:
    \begin{align*}
        \frac{d \Q_{+,+}(Y(1) \mid X = c, M = 0)}{d \PObs( Y(1) \mid X = c^+)} & = \left\{
        \begin{array}{rl}
        \frac{\Lambda}{\Lambda + \eta (1 - \Lambda)}     &  \mbox{if } Y(1) > Q_{\tau_1}(Y \mid X = c^+) \\
        \frac{1}{1 + \eta (\Lambda - 1)}     & \mbox{if } Y(1) \leq Q_{\tau_1}(Y \mid X = c^+)
        \end{array}
        \right. \\
        \frac{d \Q_{+,+}(Y(0) \mid X = c, M = 1)}{d \PObs( Y(0) \mid X = c^-)} & = \left\{
        \begin{array}{rl}
        \Lambda     &  \mbox{if } Y(0) > Q_{\tau_0}(Y \mid X = c^+) \\
        \Lambda^{-1}     & \mbox{if } Y(0) \leq Q_{\tau_0}(Y \mid X = c^+). 
        \end{array}
        \right. 
    \end{align*}

    We begin with the $Y(1)$ distribution:
    \begin{align*}
        \frac{d \Q_{+,+}( Y(1) \mid X = c, M = 0)}{d \PObs( Y(1) \mid  X = c^+ )} & = \frac{d \Q_{+,+}( Y(1) \mid X = c^+, M = 0)}{d \PObs( Y(1) \mid  X = c^+ )} \\
         & = \frac{d \Q_{+,+}( Y(1), M = 0 \mid X = c^+)}{(1-\eta) d \PObs( Y(1) \mid  X = c^+ )} \\
         & = \frac{\Q_{+,+}( M = 0 \mid X = c^+, Y(1))}{1 - \eta} \\
         & = \frac{1 - \Q_{+,+}(M = 1 \mid X = c^+, Y(1))}{1 - \eta} \\
         & = \frac{1 - q_{++}( V(X, 1, Y(1))}{1 - \eta} \\
         & = \left\{
         \begin{array}{rl}
         \frac{1 - \frac{\eta}{\Lambda + \eta (1 - \Lambda)}}{1 - \eta}     & \mbox{if } Y(1) < Q_{\tau_1}(Y \mid X = c^+) \\ 
         \frac{1 - \frac{\eta \Lambda}{1 + \eta (\Lambda - 1)}}{1 - \eta}     &  \mbox{if } Y(1) \leq Q_{\tau_1}(Y \mid X = c^+) 
         \end{array}
         \right.  \\
         & = \left\{
         \begin{array}{rl}
         \frac{\Lambda }{\Lambda + \eta (1 - \Lambda)}  & \mbox{if } Y(1) < Q_{\tau_1}(Y \mid X = c^+) \\
         \frac{1}{1 + \eta (\Lambda - 1)}     &  \mbox{if } Y(1) \leq Q_{\tau_1}(Y \mid X = c^+) 
         \end{array}
         \right. ,
    \end{align*}
    which are the desired likelihood ratios.

    We continue with the $Y(0)$ distribution:
    \begin{align*}
        \frac{d \Q_{+,+}( Y(0) \mid X = c, M = 1)}{d \PObs( Y(0) \mid  X = c^- )} & = \frac{d \Q_{+,+}( Y(0) \mid X = c^+, M = 1)}{d \PObs( Y(0) \mid  X = c^- )}  \\ 
        & = \frac{\Q_{+,+}( V = V(X, 0, Y(0)) \mid X = c^+, M = 1)}{d \PObs( V = V(X, 0, Y(0)) \mid  X = c^- )}\\ 
        & = \frac{\Q_{+,+}( V = V(X, 0, Y(0)), M = 1 \mid X = c^+)}{\eta d \PObs( V = V(X, 0, Y(0)) \mid  X = c^- )}.
    \end{align*}
    By inspection of the two cases, this is $\Lambda V(X, 0, Y(0)) + \Lambda^{-1} (1-V(X, 0, Y(0)))$, i.e. the desired likelihood ratio.

    \subsubsection{Construction of All Mixtures of Estimands}

    By symmetric arguments, we can find a $\Q_{+,-}$ that achieves the maximal value of $E_{\Q}[Y(1) \mid X = c, M=0]$ and minimal value of $\E_{\Q}[ Y(0) \mid X = c, M=1]$ over $\Q \in \mathcal{M}_{\Lambda}(\Lambda, \Lambda)$, as well as analogous $\Q_{-,+}$ and $\Q_{-,-}$ for all of the other extreme combinations. 

    Now suppose we are achieving some $(\psi_1, \psi_0)$ within both pointwise bounds, i.e. there are some $\alpha_1, \alpha_0 \in [0, 1]$ such that:
    \begin{align*}
        \alpha_1 \left( E_{\Q_{+,+}}[Y(1) \mid X=c, M=0] - E_{\Q_{-,-}}[Y(1) \mid X=c, M=0] \right) & = \psi_1 - E_{\Q_{-,-}}[Y(1) \mid X=c, M=0] \\
        \alpha_0 \left( E_{\Q_{+,+}}[Y(0) \mid X=c, M=1] - E_{\Q_{-,-}}[Y(0) \mid X=c, M=1] \right) & = \psi_0 - E_{\Q_{-,-}}[Y(0) \mid X=c, M=1].
    \end{align*}
    Define $\Q^*$ as follows:
    \begin{itemize}
        \item Draw $V_1 \sim Bern(\alpha_1)$ and $V_0 \sim Bern(\alpha_0)$
        \item Write $S_d$ to be $+$ if $V_t = 1$ and $S_t$ to be $-$ if $V_t = 0$
        \item Draw $(X, M, D, Y(1), Y(0)) \sim \Q_{S_1, S_0}$
    \end{itemize}
    By inspection, $\Q^* \in \mathcal{M}_{\Lambda}'(\Lambda)$. It also has:
    \begin{align*}
        \E_{\Q^*}[Y(t) \mid X=c, M=1-t] & = \alpha_t \psi_t^+(\Lambda) + (1-\alpha_t) \psi_t^-(\Lambda) \\
        & = \frac{\psi_t^+(\Lambda) \left( \psi_t - \psi_t^-(\Lambda) \right) + \left( \psi_t^+(\Lambda) - \psi_t \right) \psi_t^-(\Lambda)}{\psi_t^+(\Lambda) - \psi_t^-(\Lambda)} \\
        & = \psi_t 
    \end{align*} 
    Demonstrating the claim.

\end{proof}

\begin{proof}[Proof of \Cref{prop:consistency}]
Write $\Delta_\lambda \equiv \|\hat{\lambda}-\lambda\|_\infty$, $\Delta_{\underline{w}} \equiv \|\hat{\underline{w}}-\underline{w}\|_\infty$ and $\Delta_{Q} \equiv \|\hat{Q}-\bar{Q}\|_\infty$, each $\op(1)$ by \Cref{assum:Estimation}\,(i). Recall $M(R,Y) \equiv 1 + |Y| + |\bar{Q}(R)|$, and abbreviate $\phi_0 \equiv \phi(\underline{w},\bar{w},\lambda,\bar{Q};Y)$ and $\hat{\phi} \equiv \phi(\hat{\underline{w}},\hat{\bar{w}},\hat{\lambda},\hat{Q};Y)$.

\emph{Part (i).} We first record that $\Psi(\bar{Q})$ is finite. Since $0 \leq \underline{w}(R) \leq 1 \leq \bar{w}(R)$ and $\{ \lambda(R) Y - \bar{Q}(R) \}_{+} \leq |\lambda(R)Y| + |\bar{Q}(R)|$,
\begin{align}
    \left| \phi_0 \right| & \leq \left( 1 + \bar{w}(R) \right) \left( |\lambda(R) Y| + |\bar{Q}(R)| \right) \leq 2 \max\{ 1, \textstyle\sup_r |\lambda(r)| \} \, \bar{w}(R) M(R,Y), \label{eq:PhiEnvelope}
\end{align}
whose right side is integrable by \Cref{assum:Estimation}\,(ii)--(iii). So $\Psi(\bar{Q})$ is finite, and the weak law gives $\hat{E}[ \phi_0 ] \xrightarrow{p} \Psi(\bar{Q})$. It therefore suffices that $\hat{\Psi} - \hat{E}[ \phi_0 ] = \op(1)$.

Bounding the integrand difference term by term, and using that $q \mapsto \{ x - q \}_{+}$ is $1$-Lipschitz,
\begin{align*}
    \left| \hat{\phi} - \phi_0 \right|
    & \leq \left( \Delta_{\underline{w}} + \Delta_\lambda + \Delta_{Q} \right) C \left( 1 + \bar{w}(R) \right) M(R,Y) \\
    & \quad + \left| \hat{\bar{w}}(R) - \bar{w}(R) \right| C \, M(R,Y),
\end{align*}
for a constant $C$ depending only on $\sup_r |\lambda(r)|$ and, for $n$ large, on the fact that $\Delta_\lambda, \Delta_{\underline{w}}, \Delta_{Q}$ are eventually below one. Averaging, the first term is $\op(1) \cdot C \, \hat{E}[ (1+\bar{w}(R)) M(R,Y) ]$, and that sample mean converges in probability to a finite limit by the weak law, since $E_{\PObs}[(1+\bar{w}(R))M(R,Y)] < \infty$ by \Cref{assum:Estimation}\,(ii)--(iii). The second term is $\op(1)$ by the last condition of \Cref{assum:Estimation}\,(iii). No property of $\bar{Q}$ beyond integrability is used, and in particular $\bar{Q}$ need not equal $Q^{+}$.

Finally, $\Psi(\bar{Q}) \geq \Psi$ is \Cref{cor:Robustness}, whose integrability hypothesis holds by the finiteness just established, and $\Psi(Q^{+}) = \Psi$ by \Cref{thm:IDSetUB}.

\emph{Part (ii).} For $K \in [1, \infty)$ let $\hat{\Psi}_K$ denote the plug-in estimator with $\hat{\bar{w}}$ replaced by $\hat{\bar{w}} \wedge K$, and let $\Psi_K(\bar{Q}) \equiv E_{\PObs}[ \phi(\underline{w}, \bar{w} \wedge K, \lambda, \bar{Q}; Y) ]$. We argue in three steps.

First, $\phi$ is nondecreasing in its $\bar{w}$ argument, since the coefficient of $\bar{w}$ is $\{ \lambda(R) Y - \bar{Q}(R) \}_{+} \geq 0$. Hence $\hat{\Psi} \geq \hat{\Psi}_K$ deterministically for every $K$.

Second, the truncated model satisfies \Cref{assum:Estimation}\,(iii): $\bar{w} \wedge K$ is bounded by $K$, so $E_{\PObs}[ (\bar{w}(R) \wedge K) M(R,Y) ] \leq K E_{\PObs}[M(R,Y)] < \infty$ by \Cref{assum:Estimation}\,(ii), and
\[
    \hat{E}\left[ \left| \hat{\bar{w}}(R) \wedge K - \bar{w}(R) \wedge K \right| M(R,Y) \right] \leq \left\| \hat{\bar{w}} \wedge K - \bar{w} \wedge K \right\|_\infty \hat{E}\left[ M(R,Y) \right] = \op(1) \Op(1),
\]
using \Cref{assum:Estimation}\,(i) for the first factor and the weak law for the second. Part (i) applied within the truncated model therefore gives $\hat{\Psi}_K \xrightarrow{p} \Psi_K(\bar{Q})$.

Third, $\Psi_K(\bar{Q}) \uparrow \Psi(\bar{Q})$. Decompose
\begin{align*}
    \phi(\underline{w}, \bar{w} \wedge K, \lambda, \bar{Q}; Y) & = \underbrace{\underline{w}(R) \lambda(R) Y + (1-\underline{w}(R)) \bar{Q}(R)}_{\text{free of } K} \\
    & \quad + \left( \bar{w}(R) \wedge K - \underline{w}(R) \right) \{ \lambda(R) Y - \bar{Q}(R) \}_{+} .
\end{align*}
The first term is integrable by \Cref{assum:Estimation}\,(ii). For $K \geq 1$ we have $\bar{w}(R) \wedge K \geq 1 \geq \underline{w}(R)$, so the second term is nonnegative and nondecreasing in $K$, and it increases pointwise to $(\bar{w}(R) - \underline{w}(R)) \{ \lambda(R) Y - \bar{Q}(R) \}_{+}$. Monotone convergence applied to that term gives the claim.

Given $M < \infty$, the hypothesis $\Psi(\bar{Q}) = +\infty$ lets us choose $K$ with $\Psi_K(\bar{Q}) > M + 1$, and then $\PObs(\hat{\Psi} > M) \geq \PObs(\hat{\Psi}_K > M) \to 1$ by the first two steps. Since $\Psi(\bar{Q}) \geq \Psi$ whenever both are defined, the hypothesis holds in particular when $\Psi = +\infty$.
\end{proof}

The following lemma is used in the proof of \Cref{cor:CDependenceIDdSet}.

\begin{lemma}\label{prop:IPWIdentifiedSetBound}
    \emph{Bounds on identified set}. Consider the general identification setting and suppose $Y \mid R \sim \mathcal{N}( \mu(R), \sigma(R)^2 )$. Then for all $\epsilon \in (0, 1)$ and $e=e(X)$ denoting the propensity score, the identified set is a subset of $$\left[ E[\lambda(R) Y] \pm E\left[ \sigma(R) \lambda(R) (1-\underline{w}(R))  \left( \sqrt{ 2 \log( \bar{w}(R) ) } +  \sqrt{2/\pi} + (1 - \underline{w}(R))^{\epsilon} \sqrt{ 1 / (e * \epsilon) } \right)  \right] \right],$$ where $[a \pm b]$ denotes the closed interval $[a - b, a + b]$ and $\log$ is the natural logarithm. 
\end{lemma}

\begin{proof}[Proof of \Cref{cor:CDependenceIDdSet}]
We first show that the identified set of $E[Y(1)]$ is unbounded if $c > \eta_1$; the argument is symmetric for the lower bound of $E[Y(0)]$ if $c > \eta_2$. Throughout, write $\underline{e}(X) = \max\{0, e(X)-c\}$ and $\bar{e}(X) = \min\{1, e(X)+c\}$, so that \Cref{prop:IPW_Y1} gives $\lambda(R) = Z/e(X)$, $\underline{w}(R) = Z \frac{e(X)}{\bar{e}(X)}$ and $\bar{w}(R) = Z \frac{e(X)}{\underline{e}(X)} + (1-Z)\frac{1-e(X)}{1-\bar{e}(X)}$, and write $\mu(R) = E[Y \mid R]$.

By the decomposition in the proof of \Cref{prop:IPWIdentifiedSetBound}, the upper bound of the identified set of $E[Y(1)]$ is
\begin{align*}
    \psi^+ & = E\left[ \lambda(R) \mu(R) + (1 - \underline{w}(R)) CVaR_{\tau(R)}^+(R) \right].
\end{align*}
Suppose $c > \eta_1$. Since $\eta_1$ is the infimum of the support of $e(X)$, the event $\mathcal{E} = \{ e(X) \in [\eta_1, c] \}$ has positive probability. On $\mathcal{E}$ we have $\underline{e}(X) = 0$, hence $\bar{w}(R) = \infty$ and $\tau(R) = 1$, so that $CVaR^+_{\tau(R)}(R)$ is the essential supremum of $\lambda(R) Y$ given $R$. Because $Y \mid X, Z$ is Gaussian with $\sigma(X,1) > 0$ almost surely, this essential supremum is $+\infty$. Its coefficient is strictly positive: if $e(X) + c \leq 1$ then $\underline{w}(R) = e(X)/(e(X)+c) < 1$ because $c > 0$, and otherwise $\underline{w}(R) = e(X) \leq 1 - \eta_2 < 1$. The remaining term $E[\lambda(R)\mu(R)]$ is finite because $\mu$ is bounded and $e(X) \geq \eta_1 > 0$. Hence
\begin{align*}
    \psi^+ & \geq E\left[ \lambda(R) \mu(R) \right]
    + E\left[ 1\{ \mathcal{E} \} (1 - \underline{w}(R)) CVaR_{\tau(R)}^+(R) \right]
    = +\infty ,
\end{align*}
so the identified set is unbounded above; the symmetric argument gives an unbounded lower bound, and the identified set is $(-\infty, \infty)$ by the convexity established in \Cref{thm:IDSetUB}.

Now suppose $c \leq \eta_1$ and we wish to show that the identified set is uniformly bounded. The boundary case $c = \eta_1$ is included because the propensity score admits a bounded density. Since $e(X)$ admits a density, $\PObs(e(X) = \eta_1) = 0$, so $\underline{e}(X) > 0$ and hence $\bar{w}(R) < \infty$ almost surely, even though $\bar{w}$ is not bounded; equivalently $\tau(R) < 1$ almost surely and the indicator $F$ of \Cref{assum:QuantileMoments} equals one almost surely. \Cref{thm:IDSetUB} and \Cref{prop:IPWIdentifiedSetBound} therefore apply unchanged. Write $D$ for the upper bound of the support of $\frac{\sigma(X, Z)}{e(X)} + \frac{\sigma(X, Z)}{1-e(X)}$, which is finite by assumption, and fix any $\epsilon \in (0,1)$. By \Cref{prop:IPWIdentifiedSetBound},
\begin{align*}
    \psi^+ & \leq E[\lambda(R) Y] + E\left[ \sigma(R) \lambda(R) (1-\underline{w}(R))
    \left( \sqrt{ 2 \log( \bar{w}(R) ) } + \sqrt{2/\pi}
    + (1 - \underline{w}(R))^{\epsilon} \sqrt{ 1 / (e(X) \epsilon) } \right) \right] \\
    & \leq E[ \lambda(R) Y ] + D \left( E\left[ Z \sqrt{ 2 \log( \bar{w}(R) ) } \right]
    + \sqrt{2 / \pi} + \sqrt{ 1 / (\eta_1 \epsilon) } \right),
\end{align*}
where $\log$ is the natural logarithm, and where we used $(1-\underline{w}(R))^{\epsilon} \leq 1$ and $e(X) \geq \eta_1$. It remains to bound $E[ Z \sqrt{ 2 \log( \bar{w}(R) ) } ]$. Let $f_{e(X)}$ denote the density of $e(X)$ and $\bar{f}_{e(X)}$ its upper bound. Since $E[Z \mid X] = e(X)$ and $Z \log \bar{w}(R) = Z \log\left( e(X)/\underline{e}(X) \right)$,
\begin{align*}
    E[ Z \log( \bar{w}(R) ) ]
    & = \int_{\eta_1}^{1-\eta_2} p \log\left( p / (p-c) \right) f_{e(X)}(p) \, dp
    \leq \bar{f}_{e(X)} \int_{\eta_1}^{1-\eta_2} \log\left( p / (p - \eta_1) \right) dp \\
    & \leq - \bar{f}_{e(X)} \int_{\eta_1}^{1-\eta_2} \log( p-\eta_1 ) \, dp
    \leq - \bar{f}_{e(X)} \int_0^1 \log( t ) \, dt = \bar{f}_{e(X)} .
\end{align*}
The first inequality uses $p \leq 1$, $f_{e(X)} \leq \bar{f}_{e(X)}$, and $c \leq \eta_1$, the last of which gives $p - c \geq p - \eta_1$ and hence $\log(p/(p-c)) \leq \log(p/(p-\eta_1))$ for every $p$ in the support; the resulting bound is therefore uniform in $c \leq \eta_1$ and in particular covers the boundary case. The second inequality uses $\log p \leq 0$ on the support, and the third substitutes $t = p - \eta_1$ and enlarges the domain of integration, which is valid because $-\log t \geq 0$ on $(0,1)$.

Because $Z$ is binary and $\bar{w}(R) \geq 1$, Cauchy--Schwarz applied to $Z \sqrt{2 \log \bar{w}(R)} = \sqrt{Z} \cdot \sqrt{2 Z \log \bar{w}(R)}$ gives
\begin{align*}
    E\left[ Z \sqrt{ 2 \log( \bar{w}(R) ) } \right]
    & \leq \sqrt{ E[Z] } \sqrt{ 2 E[ Z \log( \bar{w}(R) ) ] }
    \leq \sqrt{ 2 \bar{f}_{e(X)} } .
\end{align*}
Therefore an APO upper bound is
\begin{align*}
    \psi^+ & \leq E[ \lambda(R) Y ] + D \left( \sqrt{ 2 \bar{f}_{e(X)} }
    + \sqrt{2 / \pi} + \sqrt{ 1 / (\eta_1 \epsilon) } \right).
\end{align*}
Since this bound holds symmetrically for the lower bound of $E[Y(0)]$, with $\eta_2$ replacing $\eta_1$, writing $B = 2 D \left( \sqrt{ 2 \bar{f}_{e(X)} } + \sqrt{2 / \pi} + \sqrt{ 1 / (\min\{\eta_1, \eta_2\} \epsilon) } \right)$ completes the proof.
\end{proof}

\begin{proof} [Proof of \Cref{prop:ExclusionBounds}]
With the $\lambda(R)$ as stated, we have:
\begin{align*}
    E_{\PTarget}&[ \lambda(R) Y \mid X] = \eta(X) \PObs(Z = 1 \mid X) \frac{1 - \PObs(Z = 1 \mid X)}{\PObs(Z = 1 \mid X) \PObs(Z = 0 \mid X)} E_{\PTarget}[ Y \mid X, Z=1 ] \\
    &\quad + \eta(X) \PObs(Z = 0 \mid X) \frac{0 - \PObs(Z = 1 \mid X)}{\PObs(Z = 1 \mid X) \PObs(Z = 0 \mid X)} E_{\PTarget}[ Y \mid X, Z=0 ]  \\
    & = \eta(X) \left( E_{\PTarget}[ Y \mid X, Z = 1 ] - E_{\PTarget}[ Y \mid X, Z = 0 ] \right) \\
    & = E_{\PTrue}\left[ \eta(X) E_{\PTrue} \left[ \sum_z \omega(z \mid X) \{ Y(T(1), z) - Y(T(0), z) \} \mid X, T(1), T(0) \right] \mid X \right] \\
    & = E_{\PTrue}\left[ \eta(X) 1\{Co\} \sum_z \omega(z \mid X) ( Y(1, z) - Y(0, z) ) \mid X \right].
\end{align*}
By iterated expectations, $E_{\PTarget}[\lambda(R) Y] = \psi$. 

The constraints on $\frac{d \PTarget(Y \mid X, T, Z)}{d \PObs(Y \mid X, T, Z)}$ follow by any pointwise bounds on those likelihood ratios that are derived as implications of \Cref{lem:excl_bound_transform}.
\end{proof}

\begin{proof}[Proof of \Cref{prop:ConservativeBounds}]
    We suppress the dependence on $X$ because it is constant. 
    
    Notice that in this example, $\PObs(Co) = \PObs(T = 1 \mid Z=1) - \PObs(T = 1 \mid Z = 0) = 1/2$; $\PObs(Nt) = 1 - \PObs(T = 1 \mid Z = 1) = 1/2$; and $\PObs(Z = 1) = 1/2$. As a result, we can write $\omega(z \mid x) = 2 z$. 

    Because $\PObs(At) = 0$, we directly observe $\PTrue( Y(1, 1) \mid Co ) = \PObs( Y \mid Z = 1, T = 1)$ and $E[ Y(1, 1) \mid Co] = 0$. The partial identification problem is to bound $\PTrue( Y(0, 1) \mid Co)$.

    Under the structural model, we have $1 \leq \frac{d \PTrue(Y(0, 0) \mid Nt)}{d \PTrue( Y(0, 1) \mid Nt)} = \frac{d \PTrue(Y(0, 0) \mid Nt)}{d \PObs( Y \mid Z = 1, T = 0)} < \infty$, which implies that the distribution of $Y(0, 0) \mid Nt$ is the distribution of $Y \mid Z = 1, T=0$. To see this, if $\frac{d \PTrue(Y(0, 0) \mid Nt)}{d \PTrue( Y(0, 1) \mid Nt)} > 1$ with positive probability, then $\frac{d \PTrue(Y(0, 0) \mid Nt)}{d \PTrue( Y(0, 1) \mid Nt)} < 1$ with some other positive probability, which contradicts the lower bound, so $\frac{d \PTrue(Y(0, 0) \mid Nt)}{d \PTrue( Y(0, 1) \mid Nt)} =1$ almost everywhere. Notice analogously that under the structural model, $0 \leq \frac{d \PTrue( Y(0, 1) \mid Co)}{d \PTrue( Y(0, 0) \mid Co)} \leq 1$, which implies that the distribution of $Y(0, 1) \mid Co$ is the same as the distribution of $Y(0, 0) \mid Co$. As a result, $E[Y(0, 1) \mid Co]$ is point-identified:
    \begin{align*}
        E[ Y \mid Z = 0, T = 0] & = \PObs(Co \mid Z = 0, T = 0) E[ Y(0, 0) \mid Co ] + \PObs( Nt ) E [ Y(0, 0) \mid Nt] \\
       0 & = 0.5 E[ Y(0, 1) \mid Co ] + 0.5 E[Y \mid Z = 1, T = 0] = 0.5 E[ Y(0, 1) \mid Co ].
    \end{align*}
    The sharp bounds are the singleton $\{0\}$.

    Now consider the statistical bounds of \Cref{prop:ExclusionBounds}. The likelihood ratios of interest from \Cref{lem:excl_bound_transform} are:
    \begin{align*}
        & \frac{d \PTarget( Y \mid T=1, Z=1)}{d \PObs( Y \mid T=1, Z=1)} = \overbrace{\omega(1)}^{=1} + \overbrace{\omega(0)}^{=0} \overbrace{\frac{d \PTrue(Y(1, 0) \mid Co)}{d \PTrue(Y(1, 1) \mid Co)}}^{\in [1, \infty)}   \\ 
        & \frac{d \PTarget(Y \mid T=1, Z=0)}{d \PObs(Y \mid T=1, Z=0)} = 1 \\
        & \frac{d \PTarget(Y \mid T=0, Z=0)}{d \PObs(Y \mid T=0, Z=0)} = \overbrace{\omega(0)}^{=0} + \overbrace{\omega(1)}^{=1} \overbrace{\frac{d \PTrue( Y(0, 1) \mid Co)}{d \PTrue( Y(0, 0) \mid Co)}}^{\in [1, \infty)} \\
         & \ \ + \underbrace{\omega(1)}_{=1} \underbrace{\frac{\PObs(Nt)}{1-\PObs(At)}}_{=1/2} \underbrace{\frac{d \PObs(Y \mid T=0, Z=1)}{d \PObs(Y \mid T=0, Z=0)}}_{=1} \bigg\{ 1  - \underbrace{\frac{d \PTrue(Y(0, 1) \mid Co)}{d \PTrue( Y(0, 0) \mid Co)}}_{\in [1, \infty)} \underbrace{\frac{d \PTrue( Y(0, 0) \mid Nt)}{d \PTrue( Y(0, 1) \mid Nt)}}_{\in [1, \infty)} \bigg\} \\
        & \frac{d \PTarget(Y \mid T=0, Z=1)}{d \PObs(Y \mid T=0, Z=1)} = \underbrace{\omega(1)}_{=1} + \underbrace{\omega(0)}_{=0} \underbrace{\frac{d \PTrue( Y(0, 0) \mid Nt)}{d \PTrue( Y(0, 1) \mid Nt)}}_{ \in [1/2, \infty)}.
    \end{align*}
    Therefore we obtain the likelihood ratio equalities $\frac{d \PTarget( Y \mid T=1, Z=1)}{d \PObs( Y \mid T=1, Z=1)} = \frac{d \PTarget( Y \mid T=1, Z=0)}{d \PObs( Y \mid T=1, Z=0)} = \frac{d \PTarget( Y \mid T=0, Z=1)}{d \PObs( Y \mid T=0, Z=1)} = 1$, while under the pointwise approach, $ \frac{d \PTarget(Y \mid T=0, Z=0)}{d \PObs(Y \mid T=0, Z=0)}$ could seemingly be as large as $\infty$ and as small as $0$. Therefore our approach's partial identification bounds on $E[ Y \mid T, Z]$ are equal to the singleton $\{0\}$ for $(1-T) (1-Z) = 0$ and are equal to the domain of $Y$, $[-1, 1]$, for $T=0, Z=0$. 

    The pointwise bounds are $$\left[ - 2 \PObs(Co) \PTrue(T = 0 \mid Z = 0), 2 \PObs(Co) \PTrue(T = 0 \mid Z = 0) \right].$$Note that $\PObs(Co) = 1/2$ and $\PTrue(T = 0 \mid Z=0) = 1$, so that the partially identified set is $[-1, 1]$.
\end{proof}

It is convenient to define the true and estimated residuals from the quantile $S \equiv \lambda(R) Y - Q_{\tau(R)}( \lambda(R) Y \mid R)$ and $\hat{S} \equiv \hat{\lambda}(R)Y - \hat{Q}_{\hat{\tau}(R)}(\hat{\lambda}(R)Y\mid R)$.

\subsection{Proofs of Lemmas}

\begin{proof}[Proof of \Cref{lemma:ReweightingHolds}]
    This is a standard result for importance sampling. 
\end{proof}

\begin{proof}[Proof of \Cref{cor:WorstCaseWeights}]
    Let $\underline{\tau}(R)$ and $\bar{\tau}(R)$ be random variables satisfying $\underline{\tau}(R) = \PObs( \lambda(R) Y < Q+(R) \mid R)$ and $\bar{\tau}(R) = \PObs( \lambda(R) Y \leq Q+(R) \mid R)$ almost surely. Note that $\PObs( \lambda(R) Y = Q^+(R) \mid R) = \bar{\tau}(R) - \underline{\tau}(R)$ almost surely.

    Define the function $\gamma: \R^d \to \R^1$ as follows:
    \begin{align*}
        \gamma(r) & \equiv \begin{cases}
            \frac{1 - \underline{\tau}(r) \underline{w}(r) - (1-\bar{\tau}(r)) \bar{w}(r)}{\bar{\tau}(r) - \underline{\tau}(r)} & \text{if} \quad \bar{\tau}(r) > \underline{\tau}(r) \\
            \underline{w}(r) & \text{else}. 
        \end{cases}
    \end{align*}
    It is clear that $\gamma(r)$ is well-defined. The result that $\gamma(R) \in [\underline{w}(R), \bar{w}(R)]$ almost surely and that the resulting $W^{*}_{sup}$ solves the upper bound optimization problem follows by an adapted version of \cite{dorn2022sharp}'s Proposition 2 argument and the observation that $E[W^{*}_{sup} \mid R] = 1$ almost surely. 
    In particular, because $Q^+(R)$ is the $\tau(R)$-quantile of $\lambda(R) Y \mid R$, the threshold probability obeys $\underline{\tau}(R) \le \tau(R) \le \bar{\tau}(R)$ almost surely. On the event $\bar{\tau}(R) > \underline{\tau}(R)$ this sandwich, applied to the definition of $\gamma(r)$, gives $\gamma(R) \in [\underline{w}(R), \bar{w}(R)]$ (with $\gamma(R) = \bar{w}(R)$ when $\underline{\tau}(R) = \tau(R)$ and $\gamma(R) = \underline{w}(R)$ when $\bar{\tau}(R) = \tau(R)$); on the complementary event $\gamma(R) = \underline{w}(R)$ trivially lies in the interval. That $W^{*}_{sup}$ attains the supremum is a Neyman--Pearson argument: subject to $W \in [\underline{w}(R), \bar{w}(R)]$ and $E_{\PObs}[W \mid R] = 1$, the conditional objective $E_{\PObs}[W \lambda(R) Y \mid R]$ is maximized by loading the largest admissible weight $\bar{w}(R)$ on the largest realizations of $\lambda(R) Y$ and the smallest weight $\underline{w}(R)$ on the smallest, with $Q^+(R)$ as the pivot and $E_{\PObs}[W^{*}_{sup} \mid R] = 1$ fixing the tie-set weight $\gamma(R)$. This is the conditional analogue of \citet{dorn2022sharp}'s Proposition 2.
\end{proof}

\begin{proof}[Proof of \Cref{cor:Robustness}]
    Recall from the proof of \Cref{thm:IDSetUB} that we may write:
    \begin{align*}
        \psi^+ & = \underline{w}(R) \lambda(R) Y + (1-\underline{w}(R)) \left( Q^+(R) + \frac{\{ \lambda(R) Y - Q^+(R) \}_+}{1-\tau(R)}\right),
    \end{align*}
    where $\{ x \}_+ = \max\{ x, 0 \}$ and where $Q^+$ is the quantile regression function.

    We make the stronger claim that:
    \begin{align*}
        Q^+(R) & \in \argmin_{\bar{Q}} \E\left[ \underline{w}(R) \lambda(R) Y + (1-\underline{w}(R)) \left( \bar{Q}(R) + \frac{\{ \lambda(R) Y - \bar{Q}(R) \}_+}{1-\tau(R)}\right) \mid R \right] \text{ a.s.}
    \end{align*}

    If $H \equiv  1\{ \tau(R) <1 \} = 0$ and hence $\tau(R) = 1$, then the right hand side should be understood to be infinite if $\lambda(R) Y > Q^+(R)$. Therefore $Q^+(R) = F_{1}(\lambda(R) Y \mid R)$ is the right-hand side minimizer on the event $H = 0$. Therefore:
    \begin{align*}
        Q^+(R) & \in \argmin_{\bar{Q}} (1-H) \E\left[ \underline{w}(R) \lambda(R) Y + (1-\underline{w}(R)) \left( \bar{Q}(R) + \frac{\{ \lambda(R) Y - \bar{Q}(R) \}_+}{1-\tau(R)}\right) \mid R \right] \text{ a.s.}
    \end{align*}
    
    On the event $H = 1$, we can write the quantile regression function $Q^+$ as some function satisfying the weighted quantile regression definition:
    \begin{align*}
        Q^+(R) & \in \argmin_{\bar{Q}} E_{\PObs}\left[ (1-\tau(R))^{-1} \left( \tau(R) \{\lambda(R) Y - \bar{Q}(R)\}_+ + (1-\tau(R)) \{ \bar{Q}(R) - \lambda(R) Y \}_+  \right) \right] \\
        \Leftrightarrow & \in \argmin_{\bar{Q}} E_{\PObs}\left[ (1-\tau(R))^{-1} \left( \{\lambda(R) Y - \bar{Q}(R)\}_+ + (1-\tau(R)) ( \bar{Q}(R) - \lambda(R) Y )   \right)  \right] \\
        & = E_{\PObs}\left[  \bar{Q}(R) - \lambda(R) Y + \frac{\lambda(R) Y - \bar{Q}(R)\}_+}{1-\tau(R)}  \right].
    \end{align*}
    The $\lambda(R) Y$ term does not affect the minimizer, so that:
    \begin{align*}
        Q^+(R) & \in \argmin_{\bar{Q}} H \E\left[ \underline{w}(R) \lambda(R) Y + (1-\underline{w}(R)) \left( \bar{Q}(R) + \frac{\{ \lambda(R) Y - \bar{Q}(R) \}_+}{1-\tau(R)}\right) \mid R \right] \text{ a.s.}
    \end{align*}

    Combining terms, we obtain:
    \begin{align*}
        Q^+(R) & \in \argmin_{\bar{Q}} \E\left[ \underline{w}(R) \lambda(R) Y + (1-\underline{w}(R)) \left( \bar{Q}(R) + \frac{\{ \lambda(R) Y - \bar{Q}(R) \}_+}{1-\tau(R)}\right) \mid R \right] \text{ a.s.}
    \end{align*}
    Completing the proof. 
\end{proof}

The following lemmas are used for the proofs of \Cref{lemma:EstimandExpressions} and \Cref{prop:RDEstimands} in Section 3. 

\begin{lemma}\label{lemma:NonManipRandomization}
    Consider the Regression Discontinuity application in \Cref{sec:RDManipulation}. Suppose $F_{X \mid M = 0}(x)$ is differentiable in $x$ at $c$ with a positive derivative. Then $\PTrue(X = c^+ \mid X = c, M = 0) = 1/2$. 
\end{lemma}

\begin{lemma}\label{lemma:ImpliedLikelihoodRatios}
    Consider the Regression Discontinuity application in \Cref{sec:RDManipulation}. Suppose $\Q \in \mathcal{M}_{\Lambda}(\infty, \infty)$. Further suppose that the distribution of $(Y(1), Y(0), T(0), M) \mid X = c$ under $\Q$ has associated manipulation selection functions $q_1(y_1) \equiv \Q(M = 1 | Y(1) = y_1, X = c)$ and $q_0(y_0) \equiv \Q(M = 1 | Y(0) = y_0, X = c)$. 
    Then the Radon-Nikodym derivatives are as follows:
    \begin{align*}
        \frac{d \Q(Y(1) \mid X = c, M = 0)}{d \PObs(Y \mid X = c^+)} & = \frac{1}{1-\ProbManip} \frac{1 - q_1(Y(1))}{1 + q_1(Y(1))}  \\
        \frac{d \Q(Y(0) \mid X = c, M = 1)}{d \PObs(Y \mid X = c^-)} & =  \frac{2(1-\ProbManip)}{\ProbManip} \frac{q_0(Y(0))}{1-q_0(Y(0))}.
    \end{align*} 
    As a result:
    \begin{align}
        E_{\PTrue}\left[ Y(1) \mid X=c, M=0 \right] & = \frac{1}{1-\ProbManip}E_{\PObs}\left[Y\frac{1-q_{1}(Y)}{1+q_{1}(Y)}\frac{1\left\{ X=c^{+}\right\} }{\PObs\left(X=c^{+}\right)}\right] \label{eqn:Y1M0} \\
        E_{\PTrue}\left[Y(0)|X = c,M=1\right] & = \frac{2\left(1-\ProbManip\right)}{\ProbManip}E_{\PObs}\left[Y\frac{q_{0}(Y)}{1-q_{0}(Y)}\frac{1\left\{ X=c^{-}\right\} }{\PObs\left(X=c^{-}\right)}\right]. \label{eqn:Y0M1}
    \end{align}
\end{lemma}

\begin{proof} [Proof of \Cref{lemma:EstimandExpressions}]
    For our proofs for the Regression Discontinuity application of \Cref{sec:RDManipulation}, it is useful to use $\ProbManip_0 \equiv \PTrue(M=1 \mid X = c)$.
    
    For the CATE, first observe that $\PTrue(X = c^- \mid X = c) = \PObs(X = c^- \mid X = c) = (1-\ProbManip_0)/2$. To see this, $\PTrue(X = c^- | X = c) = \PTrue (X = c^-| M=0, X = c) \PTrue (M=0 | X = c) = \PTrue(X = c^- | M=0, X = c) (1- \ProbManip_0) = (1-\ProbManip_0)/2$, due to \Cref{lemma:NonManipRandomization}. Consequently, $\PTrue(X = c^+ | X = c) = (1+\ProbManip_0)/2$, so
    \begin{align*}
    \psi_{CATE} &\equiv E[Y(1) - Y(0)| X = c] \\
    &= E[Y(1)| X = c^+, X = c] P(X = c^+| X = c) + E[Y(1) | X = c^- , X = c] P(X = c^-| X = c) \\
    &\qquad - E[Y(0)| X = c, M=1] P( M= 1|X = c) - P(M=0| X = c) E[Y(0)| X = c, M=0] \\
    &=  E[Y(1)| X = c^+, X = c] \frac{1 + \ProbManip_0}{2} + E[Y(1) | X = c^- , X = c] \frac{1 - \ProbManip_0}{2} \\
    &\qquad - E[Y(0)| X = c, M=1] \ProbManip_0 - (1-\ProbManip_0) E[Y(0)| X = c, M=0]. 
\end{align*}
The final equality uses the definition that $\ProbManip_0 = \PTrue(M=1| X = c)$ and the result that $\PTrue(X = c^-| X = c) = (1-\ProbManip_0)/2$. 
Further, $E[Y| X = c^+] = E[Y(1)| X = c^+]$ and $E[Y| X = c^-] = E[Y(0)| X = c^-]$ due to the sharp RD design of Assumption 1.
To obtain the expression in the lemma, observe that $\ProbManip_0 = \PTrue(M=1 \mid X=c^+) \PTrue(X=c^+ \mid X=c) = \ProbManip (1+\ProbManip_0)/2$ as $\PTrue(M=1 \mid X=c^-) =0$. Then, we can substitute $\ProbManip_0 = \ProbManip/(2-\ProbManip)$.

For the CATT, observe:
\begin{align*}
    \psi_{CATT} & = \frac{E[T Y \mid X = c] - E[T Y(0) \mid X = c]}{E[T \mid X = c]} \\
    & = \frac{E[T Y \mid X = c]}{E[T \mid X = c]} - \frac{\PTrue[T = 1, M = 0 \mid X = c] E[Y \mid X = c^-]}{E[T \mid X = c]} \\
    & - \frac{\PTrue[T = 1, M = 1 \mid X = c] E[Y \mid X = c, M = 1]}{E[T \mid X = c]} \\
    & = \frac{E[T Y \mid X = c]}{E[T \mid X = c]}  - \frac{\PTrue[T = 0, M = 0 \mid X = c] E[Y \mid X = c^-]}{E[T \mid X = c]} \\
    & - \frac{\ProbManip_0 E[Y \mid X = c, M = 1]}{E[D \mid X = c]} \\ 
    & = \frac{E[T Y \mid X = c] - E[(1-T) Y \mid X = c] - \ProbManip_0 E[Y \mid X = c, M = 1]}{E[T \mid X = c]} \\
    & = \frac{E[(2 T - 1) Y \mid X = c] - \ProbManip_0 E[Y \mid X = c, M = 1]}{\ProbManip_0 + (1-\ProbManip_0) / 2}.
\end{align*}

The CLATE is immediate. 
\end{proof}

\begin{proof} [Proof of \Cref{lem:excl_bound_transform}]

We drop conditioning on $X$ for exposition. Observe that:
\begin{align*}
d&\PTarget(Y  \mid T=0,Z=0)\\
 & =\omega(1)\left(\PObs(Nt|T=Z=0)d\PTrue(Y(0,1)|Nt)+\PObs(Co|T=Z=0)d\PTrue(Y(0,1)|Co)\right)\\
 &  + \omega(0)\left(\PObs(Nt|T=Z=0)d\PTrue(Y(0,0)|Nt)+\PObs(Co|T=Z=0)d\PTrue(Y(0,0)|Co)\right) \\ 
 & = \omega(1) \left( \PObs(Nt|T=Z=0) d\PObs( Y \mid T=0, Z=1) + \PObs(Co|T=Z=0) d\PTrue(Y(0,1)|Co)\right) \\
 & + \omega(0) d \PObs(Y \mid T=0, Z=0). 
\end{align*}
We continue to analyze the unobserved term $d\PTrue(Y(0,1)|Co)$, which we factor as:
\begin{align*}
    \frac{d\PTrue(Y(0,1)|Co)}{d \PObs(Y \mid T=0, Z=0)} & = \frac{d\PTrue(Y(0,1)|Co)}{d \PTrue( Y(0, 0) \mid Co)} \frac{d \PTrue( Y(0, 0) \mid Co)}{d \PObs(Y \mid T=0, Z=0)}
\end{align*}

Note that:
\[
\frac{d\PTrue(Y(0,0)|Co)}{d\PTrue(Y(0,0)|Nt)}=\frac{1}{\PObs\left(Co|Z=T=0\right)}\left(\frac{d\PObs(Y\mid T=0,Z=0)}{d\PTrue(Y(0,0)|Nt)}-\PObs\left(Nt|Z=T=0\right)\right),
\]
because 
\[
\frac{d\PObs(Y\mid T=0,Z=0)}{d\PTrue(Y(0,0)|Nt)}=\PObs\left(Nt|Z=T=0\right)+\PObs\left(Co|Z=T=0\right)\frac{d\PTrue(Y(0,0)|Co)}{d\PTrue(Y(0,0)|Nt)}. 
\]

Combining the results,
\begin{align*}
&\frac{d\PTarget(Y|T=0,Z=0)}{d\PObs(Y|T=0,Z=0)}  \\
& = \omega(0)+\omega(1) \PObs\left(Nt|Z=T=0\right) \frac{d\PObs(Y\mid T=0,Z=1)}{d\PObs(Y\mid T=0,Z=0)} \\
 & \quad+\omega(1)\frac{d\PTrue(Y(0,1)|Co)}{d\PTrue(Y(0,0)|Co)}\left(1-\PObs\left(Nt|Z=T=0\right)\frac{d\PTrue(Y(0,0)|Nt)}{d\PTrue(Y(0,1)|Nt)} \frac{d\PObs(Y\mid T=0,Z=1)}{d\PObs(Y\mid T=0,Z=0)}\right),
\end{align*}
where 
\begin{align*}
   \PObs\left(Nt|Z=T=0\right) & = \PObs( Nt \mid T(0) = 0) = \frac{\PObs(Nt)}{\PObs(Co) + \PObs(Nt)}. 
\end{align*}

Using an analogous argument,
\begin{align*}
&\frac{d\PTarget(Y|T=1,Z=1)}{d\PObs(Y|T=1,Z=1)} =\omega(1) + \omega(0) \frac{\PObs(At)}{\PObs(Co) + \PObs(At)} \frac{d\PObs(Y\mid T=1,Z=0)}{d\PObs(Y\mid T=1,Z=1)}\\
 & \quad+\omega(0)\frac{d\PTrue(Y(1,0)|Co)}{d\PTrue(Y(1,1)|Co)}\left(1-\frac{\PObs(At)}{\PObs(Co) + \PObs(At)} \frac{d\PTrue(Y(1,1)|At)}{d\PTrue(Y(1,0)|At)}\frac{d\PObs(Y\mid T=1,Z=0)}{d\PObs(Y\mid T=1,Z=1)}\right). 
\end{align*}

We continue to analyze $d\PTarget\left(Y|T=0,Z=1\right)$.
\begin{align*}
d\PTarget &\left(Y|T=0,Z=1\right) = \omega(0)d\PTrue\left(Y(0,0)\mid Nt\right) + \omega(1)d\PTrue\left(Y(0,1)\mid Nt\right) \\
 & =\omega(0)\frac{d\PTrue\left(Y(0,0)\mid Nt\right)}{d\PTrue\left(Y(0,1)\mid Nt\right)}d\PTrue\left(Y(0,1)\mid Nt\right)+\omega(1)d\PObs\left(Y\mid T=0,Z=1\right)\\
 & =\left(\omega(0)\frac{d\PTrue\left(Y(0,0)\mid Nt\right)}{d\PTrue\left(Y(0,1)\mid Nt\right)}+\omega(1)\right)d\PObs\left(Y\mid T=0,Z=1\right).
\end{align*}

Analogously,
\begin{align*}
    d \PTarget&(Y \mid T = 1, Z = 0) = \omega(1) d \PTrue (Y (1,1)\mid At) + \omega(0) d \PTrue (Y(1,0)\mid At) \\
    &= \omega(1) \frac{d \PTrue (Y (1,1)\mid At)}{ d \PObs (Y\mid T=1,Z=0)} d \PObs (Y\mid T=1,Z=0) + \omega(0) d \PObs (Y\mid T=1,Z=0) \\
    &= \left( \omega(1) \frac{d \PTrue (Y (1,1)\mid At)}{ d \PTrue (Y (1,0)\mid At)} + \omega(0) \right) d \PObs (Y\mid T=1,Z=0). 
\end{align*}
\end{proof}

\begin{proof}[Proof of \Cref{lemma:IMCalibration}]
\emph{Step 1: the critical value is a continuous decreasing function of the width-to-scale ratio.} For $\rho \in [0, \infty]$, the map $r \mapsto \Phi(r + \rho) - \Phi(-r)$ is continuous and strictly increasing on $[0, \infty)$, equals $\Phi(\rho) - 1/2 \leq 1/2 < 1-\alpha$ at $r = 0$, and tends to one. So $r_\alpha(\rho)$ is well defined, unique and strictly positive; by the implicit function theorem it is continuous and strictly decreasing in $\rho$, with $r_\alpha(0) = \Phi^{-1}(1-\alpha/2)$ and $r_\alpha(\infty) = \Phi^{-1}(1-\alpha)$.

\emph{Step 2: reduction to a subsequence.} For each $n$ pick $\psi_n \in [L_n, U_n]$ within $1/n$ of the infimum of the coverage probability over that interval. Take an arbitrary subsequence and pass to a further subsequence along which
\begin{align*}
    \frac{\Delta_n}{\sigma^{\max}_n} \to \rho \in [0,\infty], \quad \frac{\psi_n - L_n}{\sigma^{\max}_n} \to \delta^{L}, \quad \frac{U_n - \psi_n}{\sigma^{\max}_n} \to \delta^{U}, \quad \frac{\sigma^{j}_n}{\sigma^{\max}_n} \to s^{j},
\end{align*}
with $\delta^{L}, \delta^{U} \geq 0$, $\delta^{L} + \delta^{U} = \rho$, $s^{j} \in [0,1]$ and $\max\{s^{L}, s^{U}\} = 1$. Then $\hat{\sigma}_{\max} / \sigma^{\max}_n \xrightarrow{p} \bar{c} \equiv \max\{ c^{L} s^{L}, c^{U} s^{U} \}$, which is finite and, because $c^{j} \geq 1$, satisfies $\bar{c} \geq \max\{s^{L}, s^{U}\} = 1$. Under the first alternative in (iii),
\begin{align*}
    \hat{\rho} & = \frac{\Delta_n}{\sigma^{\max}_n} \frac{\sigma^{\max}_n}{\hat{\sigma}_{\max}} + \frac{\hat{\Delta} - \Delta_n}{\hat{\sigma}_{\max}} \xrightarrow{p} \rho / \bar{c} ;
\end{align*}
under the second, $\rho = \infty$ and $\hat{\rho} = ( \hat{\Delta} / \Delta_n ) ( \Delta_n / \sigma^{\max}_n ) ( \sigma^{\max}_n / \hat{\sigma}_{\max} ) \xrightarrow{p} \infty = \rho / \bar{c}$. Either way $r_\alpha(\hat{\rho}) \xrightarrow{p} r^{\infty} \equiv r_\alpha( \rho / \bar{c} ) > 0$ by Step 1.

\emph{Step 3: coverage along the subsequence.} Non-coverage of $\psi_n$ is $A^{L}_n \cup A^{U}_n$ with $A^{L}_n = \{ \hat{L} - r_\alpha(\hat{\rho}) \hat{\sigma}^{L} > \psi_n \}$ and $A^{U}_n = \{ \hat{U} + r_\alpha(\hat{\rho}) \hat{\sigma}^{U} < \psi_n \}$. Normalizing,
\begin{align*}
    A^{L}_n & = \left\{ \frac{ \hat{L} - L_n }{ \sigma^{L}_n } > r_\alpha(\hat{\rho}) \frac{ \hat{\sigma}^{L} }{ \sigma^{L}_n } + \frac{ \psi_n - L_n }{ \sigma^{L}_n } \right\} .
\end{align*}
The first term on the right converges in probability to $c^{L} r^{\infty} \geq r^{\infty}$, and the second is bounded below by $(\psi_n - L_n)/\sigma^{\max}_n \to \delta^{L}$ because $\sigma^{L}_n \leq \sigma^{\max}_n$. The left side converges in distribution to a standard normal, so $\limsup_n \PP(A^{L}_n) \leq 1 - \Phi( r^{\infty} + \delta^{L})$. Symmetrically $\limsup_n \PP(A^{U}_n) \leq \Phi( -r^{\infty} - \delta^{U})$. Hence
\begin{align*}
    \liminf_n \PP\left( \psi_n \in \mathcal{C} \right) & \geq \Phi(r^{\infty} + \delta^{L}) - \Phi(-r^{\infty} - \delta^{U}) = \Phi(r^{\infty} + \delta^{L}) + \Phi(r^{\infty} + \delta^{U}) - 1 .
\end{align*}
Since $r^{\infty} > 0$, the map $x \mapsto \Phi(r^{\infty} + x)$ is concave on $[0,\infty)$, so $x \mapsto \Phi(r^{\infty}+x) + \Phi(r^{\infty} + \rho - x)$ is concave on $[0,\rho]$ and is minimized at an endpoint. The displayed lower bound is therefore at least $\Phi(r^{\infty}) + \Phi(r^{\infty} + \rho) - 1$, which since $\bar{c} \geq 1$ is at least $\Phi(r^{\infty}) + \Phi(r^{\infty} + \rho/\bar{c}) - 1 = 1-\alpha$ by the definition of $r^{\infty}$. When $\rho = \infty$ the same conclusion follows from $\Phi(r^{\infty} + \delta^{L}) + \Phi(r^{\infty} + \delta^{U}) - 1 \geq \Phi(r^{\infty}) = \Phi(\Phi^{-1}(1-\alpha))$. As the subsequence was arbitrary, the bound holds along the full sequence, and by the choice of $\psi_n$ it holds uniformly over $[L_n, U_n]$.

\emph{Step 4: coverage of the interval.} Fix $r = \Phi^{-1}(1-\alpha/2)$. Failure to contain $[L_n, U_n]$ is contained in $\{ \hat{L} - r \hat{\sigma}^{L} > L_n \} \cup \{ \hat{U} + r \hat{\sigma}^{U} < U_n \}$. By (i) and (ii) and $c^{j} \geq 1$, their probabilities converge to at most $1 - \Phi(c^{L} r) \leq \alpha/2$ and $\Phi(-c^{U} r) \leq \alpha/2$, and a union bound gives the claim. Running the same argument with $r_\alpha(\hat{\rho}) \geq \Phi^{-1}(1-\alpha)$ in place of $r$ gives $1-2\alpha$, which is all the \citet{imbens2004confidence} calibration guarantees for the interval.
\end{proof}

\begin{proof}[Proof of \Cref{lemma:RewriteA}]
    Since $1-\tau = \frac{(\bar{w}-\underline{w}) - (\bar{w}-1)}{\bar{w}-\underline{w}} = \frac{1-\underline{w}}{\bar{w}-\underline{w}}$, we have $\frac{1-\underline{w}}{1-\tau} = \bar{w}-\underline{w}$, and the claim follows by substitution.
\end{proof}

\begin{proof}[Proof of \Cref{prop:IPWIdentifiedSetBound}]
    We show that the upper bound is at most $E[ \lambda(R) \mu(R) ]$ plus one-half the proposed width; the lower bound follows symmetrically.

    Note that as we argue in \Cref{thm:IDSetUB}, the upper bound for the identified set can be written as:
    \begin{align*}
        \psi^+ & = E_{\PObs}\left[ \underline{w}(R) \lambda(R) Y + (1-\underline{w}(R)) CVaR_{\tau(R)}^+(R) \right].
    \end{align*}
    In the Normal-residual case, we can write $CVaR_{\tau(R)}^+(R) = \lambda(R) \mu(R) + \lambda(R) \sigma(R) \frac{\phi(q_{\tau(R)})}{1-\tau(R)},$
    where $q_\tau$ is the $\tau^{\text{th}}$ quantile of a standard normal distribution and $\phi$ is the standard normal CDF. 
    By existing arguments (e.g. \citet{pinelis2019exact}), the inverse Mills ratio $\phi(q) / (1-\Phi(q))$ has the upper bound $\sqrt{2/\pi} + q$. 
    
    Therefore the APO upper bound can be further bounded as:
    \begin{align*}
        \psi^+ & = E_{\PObs}\left[ \lambda(R) \mu(R) + (1-\underline{w}(R)) \sigma(R) \lambda(R) \frac{\phi(q_{\tau(R)})}{1-\tau(R)} \right] \\
        & \leq E[ \lambda(R) \mu(R)] + E \left[ (1 - \underline{w}(R)) \sigma(R) \lambda(R) \left( \sqrt{2/\pi} + q_{\tau(R)} \right) \right].
    \end{align*}
    
    It remains to bound $q_{\tau(R)}$.  By standard arguments, if $S \sim N(0, 1)$, then $P(S > s) \leq exp( - s^2 / 2)$. We substitute $s = q_{\tau(R)}$ to obtain:
    \begin{align*}
        1 - \tau(R) = P( S > q_{\tau(R)} ) & \leq exp( - q_{\tau(R)}^2 / 2) \\
        \log(1 - \tau(R)) & \leq -  q_{\tau(R)}^2 / 2 \\
        \sqrt{ \log\left( \frac{1}{1 - \tau(R))^2} \right) } & \geq q_{\tau(R)}.
    \end{align*}
    
    Therefore we have bounded the identified set as:
    \begin{align*}
        \psi^+ & \leq E \left[ \lambda(R) \mu(R) +  (1-\underline{w}(R)) \sigma(R) \lambda(R) \left( \sqrt{2/\pi} + \sqrt{ 2 (1 - \underline{w}(R))^2 \log\left( \frac{1}{1 - \tau(R))} \right) } \right) \right] .
    \end{align*}

    Now we bound the second square root, using the identity:
    \begin{align*}
        \frac{1}{1-\tau(R)} = \frac{\bar{w}(R) - \underline{w}(R)}{1 - \underline{w}(R)} & = 1 + \frac{\bar{w}(R) - 1}{1 - \underline{w}(R)}.
    \end{align*}
    Therefore:
    \begin{align*}
        2 &(1 - \underline{w}(R))^2 \log\left( \frac{1}{1 - \tau(R))} \right) \\
        & = 2 (1 - \underline{w}(R))^2 \log( \bar{w}(R) - \underline{w}(R)) - 2 (1 - \underline{w}(R))^{2-2\epsilon} (1-\underline{w}(R))^{2 \epsilon} \log(1 - \underline{w}(R)) \\
        & \leq 2 (1 - \underline{w}(R))^2 \log( \bar{w}(R)  ) + \frac{(1 - \underline{w}(R))^{2-2 \epsilon} }{e * \epsilon}.
    \end{align*}
    So that we now have the bound:
    \begin{align*}
        \psi^+ & \leq E \left[ \lambda(R) \mu(R) \right] \\
        &+ E \left[ \sigma(R) \lambda(R) \left( (1-\underline{w}(R)) \sqrt{2/\pi} + \sqrt{ 2 (1-\underline{w}(R))^2 \log( \bar{w}(R) ) + \frac{(1 - \underline{w}(R))^{2-2 \epsilon} }{e * \epsilon} } \right) \right] \\
        & \leq E \left[ \lambda(R) \mu(R) \right] \\
        &+ E\left[ \sigma(R) \lambda(R) (1-\underline{w}(R))  \left( \sqrt{ 2 \log( \bar{w}(R) ) } +  \sqrt{2/\pi} + (1 - \underline{w}(R))^{\epsilon} \sqrt{ 1 / (e * \epsilon) } \right)  \right] .
    \end{align*}
    Applying the same argument to the symmetric lower bound completes the proof.
\end{proof}

\begin{proof} [Proof of \Cref{lemma:NonManipRandomization}]
Define $f_{x \mid M = 0}(c) > 0$ to be the derivative of $F_{X \mid M = 0}(x)$ at $c$. Then we have:
    \begin{align*}
        &\PTrue(X = c^+ \mid X = c, M = 0)\\
        %&=\PTarget(X = c^+ \mid X = c, M = 0) \\
        & = \lim_{\varepsilon \to 0^+} \frac{F_{X \mid M = 0}(c + \varepsilon) - F_{X \mid M = 0}(c)}{F_{X \mid M = 0}(c + \varepsilon) - F_{X \mid M = 0}(c - \varepsilon)} \\
         & = \lim_{\varepsilon \to 0^+} \frac{F_{X \mid M = 0}(c + \varepsilon) - F_{X \mid M = 0}(c)}{F_{X \mid M = 0}(c + \varepsilon) - F_{X \mid M = 0}(c) + F_{X \mid M = 0}(c) - F_{X \mid M = 0}(c - \varepsilon)}  \\
         & = \lim_{\varepsilon \to 0^+} \frac{\frac{F_{X \mid M = 0}(c + \varepsilon) - F_{X \mid M = 0}(c)}{\varepsilon}}{\frac{F_{X \mid M = 0}(c + \varepsilon) - F_{X \mid M = 0}(c)}{\varepsilon} + \frac{F_{X \mid M = 0}(c) - F_{X \mid M = 0}(c - \varepsilon)}{\varepsilon}} \\
         & = \frac{\lim_{\varepsilon \to 0^+} \frac{F_{X \mid M = 0}(c + \varepsilon) - F_{X \mid M = 0}(c)}{\varepsilon}}{\lim_{\varepsilon \to 0^+} \frac{F_{X \mid M = 0}(c + \varepsilon) - F_{X \mid M = 0}(c)}{\varepsilon} + \lim_{\varepsilon \to 0^-} \frac{F_{X \mid M = 0}(c + \varepsilon) - F_{X \mid M = 0}(c)}{\varepsilon}} \\
         & = \frac{f_{x \mid M = 0}(c)}{2 f_{x \mid M = 0}(c)} = 1 / 2
    \end{align*}
\end{proof}

\begin{proof} [Proof of \Cref{lemma:ImpliedLikelihoodRatios}]
    For our proofs, it is useful to use $\ProbManip_0$ instead. We define:
    \begin{align*}
        \ProbManip_0 \equiv \PTrue( M=1 \mid X = c) = \PObs(X = c^+ \mid X = c) \ProbManip. 
    \end{align*}
    Notice that $\PObs(X = c^+ \mid X = c) = 1 / (2-\ProbManip)$, because $\PObs(X = c^- \mid X = c) = \PTrue( M = 0 \mid X = c) / 2 = (1-\ProbManip) \PObs(X = c^+ \mid X = c)$. As a result, $\ProbManip_0 = \eta / (2-\eta)$, $\eta=2\eta_0/(1+\eta_0)$, and $\PObs(X=c^+ \mid X=c) = (1+\eta_0)/2$.

    Also notice that in Equation (\ref{eq:RDManipulationBounds}), the constraints on manipulation probabilities are defined relative to:
    \begin{equation} \label{eq:M1M0ratio}
        \frac{\PTrue( M = 1 \mid X = c)}{\PTrue( M = 0 \mid X = c)} = \frac{\ProbManip_0}{1-\ProbManip_0} = \frac{\ProbManip}{2(1-\ProbManip)}.
    \end{equation}

    We begin with $Y(1)$. Since $\Q$ produces the observable distribution $\PObs(X = c^+ \mid X = c) d \PObs( Y \mid X = c^+)$ and $\PObs(X=c^+ \mid X=c) = (1+\eta_0)/2$, we have:
    \begin{align*}
        \Q(X = c^+ \mid X = c) d \Q( Y(1) \mid X = c^+) & = d \Q(Y(1) \mid X = c) \Q(X = c^+ \mid X = c, Y(1)) \\
        & = d \Q(Y(1) \mid X = c) \left( 1+q_1(Y(1)) \right)/2 \\
        d \Q(Y(1) \mid X = c) & = \frac{2 \PObs(X = c^+ \mid X = c)}{1 + q_1(Y(1))} d \Q(Y(1) \mid X = c^+) 
    \end{align*} 

    We can also derive the probability of a treated observation being manipulated under $\Q$ through Bayes' Rule:
    \begin{align*}
        \Q&(M = 0 \mid Y(1), X = c^+) \\
        & = \frac{d \Q(Y(1) \mid X = c) \Q(M = 0 \mid X = c, Y(1)) \Q(X = c^+ \mid X = c, Y(1), M = 0)}{\Q(X = c^+ \mid X = c) d \Q(Y(1) \mid X = c^+)} \\
        & = \frac{d \Q(Y(1) \mid X = c) * (1-q_1(Y(1))) / 2}{d \Q(Y(1) \mid X = c) \left( 1+q_1(Y(1)) \right)/2 } \\
        & = \frac{1 - q_1(Y(1))}{1 + q_1(Y(1))}
    \end{align*} 

    As a result:
    \begin{align*}
        d \Q(Y(1) \mid X = c, M = 0) & = d \Q(Y(1) \mid X = c^+, M = 0) \\
        & = \frac{d \Q(Y(1) \mid X = c^+) \Q(M = 0 \mid X = c^+, Y(1))}{\Q(M = 0 \mid X = c^+)} \\
        & = \frac{\frac{1 - q_1(Y(1))}{1 + q_1(Y(1))}}{1-\ProbManip} d \PTrue(Y(1) \mid X = c^+)  
    \end{align*}
    This is our first equality. 

    We now turn our attention to $Y(0)$. By a similar observed-untreated-outcome argument, we have:
    \begin{align*}
        & \Q(X = c^- \mid X = c) d \Q(Y(0) \mid X = c^-) \\
        = & d \Q(Y(0) \mid X = c) \Q(X = c^- \mid X = c, Y(0)) \\
        = & d \Q(Y(0) \mid X = c) \Q(M = 0 \mid X = c, Y(0)) \Q(X = c^- \mid X = c, Y(0), M = 0) \\
        = & d \Q(Y(0) \mid X = c) ( 1 - q_0(Y(0)) ) / 2 
    \end{align*} 

    Since $\Q(X = c^- \mid X = c) = \PTarget( X = c^- \mid X = c) = \frac{1-\ProbManip_0}{2}$, we can then obtain:
    \begin{align*}
        d \Q(Y(0) \mid X = c) & = \frac{1-\ProbManip_0}{1-q_0(Y(0))} d \Q(Y(0) \mid X = c^-) 
    \end{align*}

    We can also split up $d \Q(Y(0) \mid X = c)$ as:
    \begin{align*}
        d \Q(Y(0) \mid X = c) & = \ProbManip_0 d \Q(Y(0) \mid X = c, M = 1) + (1-\ProbManip_0) d \Q(Y(0) \mid X = c, M = 0) \\
         & = \ProbManip_0 d \Q(Y(0) \mid X = c, M = 1) + (1-\ProbManip_0) d \Q(Y(0) \mid X = c^- )
    \end{align*}
    So that we can combine terms to obtain:

    \begin{align*}
        d \Q(Y(0) \mid X = c, M = 1) & = \frac{1-\ProbManip_0}{\ProbManip_0} \frac{q_0(Y(0))}{1-q_0(Y(0))} d \Q(Y(0) \mid X = c^- ) \\
        \frac{1-\ProbManip_0}{\ProbManip_0} & = \frac{2(1-\ProbManip)}{\ProbManip} \\
        d \Q(Y(0) \mid X = c, M = 1) & = \frac{2(1-\ProbManip)}{\ProbManip} \frac{q_0(Y(0))}{1-q_0(Y(0))} d \Q(Y(0) \mid X = c^- ) 
    \end{align*} 
    Which is the final equality after substituting in  $d \Q(Y(0) \mid X = c^-) = d \PTrue(Y(0) \mid X = c^-)$.
\end{proof}

%\subsection{Proofs of Corollary}

%%%%%%%%%%%%%%%%%%%%
\end{document}